\PassOptionsToPackage{x11names}{xcolor}
\RequirePackage{xr}

\documentclass[acmsmall,screen]{acmart}
\usepackage{filecontents}

\def\ispopl{}
\def\includesequetialscheduleappendix{}
\let\includesequetialscheduleappendix\undefined %
\usepackage{./common}
\makeatletter%
\RequirePackage[disable]{todonotes}

\RequirePackage{rline}
\RequirePackage{sort}

\RequirePackage{endnotes}

\RequirePackage{microtype}

\PassOptionsToPackage{framemethod=tikz}{mdframed}
\RequirePackage{mdframed}
\RequirePackage{soul}
\RequirePackage{colortbl}
\definecolor{morange}{RGB}{255,100,0}
\definecolor{mblue}{RGB}{86,180,233}
\definecolor{mgreen}{RGB}{0,158,115}
\definecolor{myellow}{RGB}{255,255,0}
\definecolor{mred}{RGB}{204,51,139}
\definecolor{mpurple}{RGB}{238,130,238}
\definecolor{msilver}{RGB}{192,192,192}

\newenvironment{changed}[1]{}{}
\newenvironment{changednospace}[1]{}{}
\renewcommand\hl[1]{}

\newcommand\hlone[1]{}

\renewcommand{\lsttabsize}{2}

\RequirePackage{array}
\newcolumntype{q}{>{\setbox0=\hbox\bgroup}c<{\egroup}@{}}

\newcommand{\appendixlocationnote}{All appendices can be found in the supplementary materials section of the ACM Digital Library.}

\makeatletter
\newcommand*{\addFileDependency}[1]{%
  \typeout{(#1)}
  \@addtofilelist{#1}
  \IfFileExists{#1}{}{\typeout{No file #1.}}
}
\makeatother
\newcommand*{\myexternaldocument}[1]{%
    \externaldocument{#1}%
    \addFileDependency{#1.tex}%
    \addFileDependency{#1.aux}%
}\makeatother

\myexternaldocument{popl-appendix}

\setcopyright{rightsretained}
\acmPrice{}
\acmDOI{10.1145/3434301}
\acmYear{2021}
\copyrightyear{2021}
\acmSubmissionID{popl21main-p145-p}
\acmJournal{PACMPL}
\acmVolume{5}
\acmNumber{POPL}
\acmArticle{20}
\acmMonth{1}

\graphicspath{{.}}
\begin{document}

\makeatletter\@input{yy.tex}\makeatother

\title{Simplifying Dependent Reductions in the Polyhedral Model}         
\author{Cambridge Yang}
\affiliation{
  \institution{MIT CSAIL}            %
  \country{USA}                    %
}
\email{camyang@csail.mit.edu}          %

\author{Eric Atkinson}
\affiliation{
  \institution{MIT CSAIL}            %
  \country{USA}                    %
}
\email{eatkinson@csail.mit.edu}          %

\author{Michael Carbin}
\affiliation{
  \institution{MIT CSAIL}            %
  \country{USA}                    %
}
\email{mcarbin@csail.mit.edu}          %

\begin{abstract}
A Reduction -- an accumulation over a set of values, using an associative and
commutative operator --  is a common computation in many numerical
computations, including scientific computations, machine learning,
computer vision, and financial analytics.
	
Contemporary polyhedral-based compilation techniques make it possible to optimize 
reductions, such as prefix sums, in which each component of the reduction's output 
potentially shares computation with another component in the reduction.
Therefore an optimizing compiler can identify the computation shared between multiple components and generate code that computes the shared computation only once. 

These techniques, however, do not support reductions that -- when phrased in the language of the polyhedral model -- span multiple 
dependent statements.
In such cases, existing approaches can generate incorrect code that violates the data dependences of the original, unoptimized program.

In this work, we identify and formalize the optimization of dependent reductions as an integer bilinear program.
We present a heuristic optimization algorithm that uses an affine sequential schedule of the program to determine how 
to simplfy reductions yet still preserve the program's dependences.

We demonstrate that the algorithm provides optimal complexity for a set of 
benchmark programs from the literature on probabilistic inference algorithms, whose 
performance critically relies on simplifying these reductions.
The complexities for 10 of the 11 programs improve
siginifcantly by factors
at least of the sizes of the input data, which are
in the range of $10^4$ to $10^6$ for typical real application inputs. 
We also confirm the significance of the improvement by showing 
speedups in wall-clock time that range from $1.1\text{x}$ to over 
$10^6\text{x}$.

\end{abstract}

\begin{CCSXML}
<ccs2012>
<concept>
<concept_id>10011007.10011006.10011041</concept_id>
<concept_desc>Software and its engineering~Compilers</concept_desc>
<concept_significance>500</concept_significance>
</concept>
<concept>
<concept_id>10003752.10010124.10010131</concept_id>
<concept_desc>Theory of computation~Program semantics</concept_desc>
<concept_significance>300</concept_significance>
</concept>
<concept>
<concept_id>10003752.10003753.10003757</concept_id>
<concept_desc>Theory of computation~Probabilistic computation</concept_desc>
<concept_significance>300</concept_significance>
</concept>
<concept>
<concept_id>10003752.10010124.10010138.10010143</concept_id>
<concept_desc>Theory of computation~Program analysis</concept_desc>
<concept_significance>500</concept_significance>
</concept>
</ccs2012>
\end{CCSXML}

\ccsdesc[500]{Software and its engineering~Compilers}
\ccsdesc[500]{Theory of computation~Program analysis}
\ccsdesc[300]{Theory of computation~Program semantics}
\ccsdesc[300]{Theory of computation~Probabilistic computation}

\keywords{reductions, polyhedral model, program dependence}

\maketitle

\section{Introduction}
\label{sec:introduction}

A {\em reduction} -- an accumulation over a set of values, using an associative and
commutative operator --  is a common computation in many areas, including 
scientific computations, machine learning, computer
vision, and financial analytics. 
For example, consider the prefix sum defined mathematically by \Cref{eq:scan} and 
presented by \Cref{lst:naiveprefixsum} in an imperative language with loops.

The value at each index $i$ of the array $B$ is the summation of values at indices $j$ 
before and up to $i$ of array $A$. 
The complexity of this na\"{i}ve prefix sum is $\BigO(N^2)$: $\BigO(N)$ for iterating 
over "$\forall i$", and $\BigO(N)$ for the summation over $j$.

\begin{minipage}[t]{0.42\linewidth}
\begin{center}
\begin{equation}
\label{eq:scan}
B[i] = \sum_{j=0}^{j \le i} A[j] \quad \forall i, 0 \le i < N
\end{equation}
\end{center}
\end{minipage}\hfill
\begin{minipage}[t]{0.45\linewidth}
\begin{clisting}[caption={Na\"{i}ve prefix sum}, 
label=lst:naiveprefixsum,columns=fullflexible]
// B is array of ints initialized to all 0
for(i = 0; i < N; i++)
	for(j = 0; j<=i; j++)
		B[i] += A[j]`\label{cline:reduction}`
\end{clisting}
\end{minipage}

\begin{changed}{morange}
\paragraph{Shared Computation and Reuse}
It is possible to implement prefix 
sum with $\BigO(N)$ complexity, which is a linear speedup over this 
na\"{i}ve implementation.
The optimization relies on the fact that 
consecutive iterations of the loop 
that computes $B$ (indexed by $i$) {\em share} equivalent 
computations: for any pair of consecutive iterations $[i]$ and $[i + 1]$, 
the values of the entire set of computations \{$A[j] \mid j < i\}$ are the 
same. 
Therefore, the latter iteration shares the former iteration's entire computation: $\sum_{j=0}^{j < i} A[j]$.
In principle, that shared computation can be {\em reused}: computed once in the former iteration, stored, and then reused in the latter iteration.
\end{changed}

\begin{minipage}[t]{0.45\linewidth}
\begin{subequations}
\allowdisplaybreaks
\begin{align}
B[0] &= A[0] \label{eq:scan-opt-init}\\
B[i] &= B[i-1] + A[i] \quad \forall i, 1 \le i < N \label{eq:scan-opt-update}
\end{align}
\label{eq:scan-opt}
\end{subequations}
\end{minipage}\hfill
\begin{minipage}[t]{0.45\linewidth}
\begin{clisting}[caption={Optimized prefix sum}, 
label=lst:optprefixsum,columns=fullflexible]
// B is array of ints initialized to all 0
B[0] = A[0] `\label{cline:reductioninit}`
for(i = 1; i < N; i++)
	B[i] = B[i-1] + A[i] `\label{cline:reductionincrement}`
\end{clisting}
\end{minipage}

\paragraph{Optimized Reductions}
\Cref{eq:scan-opt}  -- and, correspondingly, \Cref{lst:optprefixsum} --  presents 
an $\BigO(N)$ implementation of prefix sum that reuses the shared computation between iterations. 
Instead, an iteration calculates its result, stores its result into $B$ (as normal), and the subsequent iteration reuses the result of the previous iteration (via $B[i - 1]$) to compute its own result.

\subsection{Simplifying Reductions}

\citet{sr} developed a suite of polyhedral compilation techniques for {\em 
Simplifying Reductions} (SR) that can be applied to automatically 
transform \Cref{eq:scan} to \Cref{eq:scan-opt} for an array equational 
language that supports reductions as a first class operation 
\citep{alphaz}.  Several of the challenges that these techniques solve are 
1) identifying shared computation, 2) identifying if shared computation 
is {\em reusable}: if it's possible to transform the program to exploit 
the reuse), and 3) identifying if reusing shared computation is {\em 
profitable}: the transformation reduces the complexity of a program.
SR provides a suite of specifications and techniques to identify shared 
computation as well as identify when shared computation is profitably 
reusable.\footnote{We refer readers to \Cref{sec:sr} for additional 
explanation of how SR solves these problems. \appendixlocationnote}

\paragraph{Reuse Vector} The core of SR, the {\em Simplification Transformation} (ST), codifies the set of profitably reusable computations as a set of {\em reuse vectors}. 
For \Cref{eq:scan}, the reuse vector $[1, 0]^\transpose$ denotes the shared computation: changing $i$ to $i+1$ and $j$ to $j+0$ (i.e. not changing $j$) does not change the value ofthe reduction body $A[j]$.
Given an equational statement and a reuse vector, ST automatically transforms the statement into a set of statements that together is semantically equivalent to the original statement, but reuses shared computation. 
For example, given \Cref{eq:scan} and the reuse vector $[1, 0]^\transpose$, ST transforms \Cref{eq:scan} to \Cref{eq:scan-opt}.

\paragraph{Choosing a Reuse Vector}
The space of valid reuse vectors is infinite in general, e.g.,  any vector $[c, 
0]^\transpose$ with constant $c$ is a valid choice for the reuse  vector for 
\Cref{eq:scan}, since they all satisfy that changing from $i$ to $i+c$ and not changing 
$j$ does not change the evaluation of $A[j]$. 
Moreover, different reuse vectors result in different output equations by ST and, 
ultimately, different programs. 
As a concrete example, applying ST to \Cref{eq:scan} with the reuse vector $[-1, 0]^\transpose$ produces \Cref{eq:scan-opt-rev}.
\begin{subequations}
\label{eq:scan-opt-rev}
\begin{align}
B[N-1] &= \sum_{j=0}^{j<N} A[j] \label{eq:scan-opt-rev-init} \\
B[i] &= B[i+1] - A[i] \quad \forall i, 0 \le i < N-1 \label{eq:scan-opt-rev-update}
\end{align}
\end{subequations}
Instead of initializing $B[0]$ and computing $B[i]$ from lower indices to higher indices (i.e. left to right) as in \Cref{eq:scan-opt}, \Cref{eq:scan-opt-rev} initializes $B[N-1]$ and computes $B[i]$ from higher to lower indices (i.e. right to left). However, this still has complexity $\BigO(N)$, just as does \Cref{eq:scan-opt}.

\todo{Explain how they propose to choose a reuse vectors?}

\newcommand{\refeqscanopt}{%
	\Cref{eq:scan-opt-init,eq:scan-opt-update,eq:ms-scan-iter-update}%
}
\newcommand{\refeqscanoptrev}{%
\Cref{eq:scan-opt-rev-init,eq:scan-opt-rev-update,eq:ms-scan-iter-update}%
}

\subsection{Dependent Reductions}
The SR framework, including ST, only optimizes reductions that are independent.
As a point of contrast, consider the following {\em dependent} reduction:
\noindent\begin{subequations}
\label{eq:ms-scan}
\allowdisplaybreaks
\begin{align}
B[i] &= \sum_{j=0}^{j\le i} A[j] \quad \forall i, 0 \le i < N  
\label{eq:ms-scan-reduction}  \\
A[i+1] &= f(B[i]) \quad \forall i, 0 \le i < N-1 \label{eq:ms-scan-iter-update}
\end{align}
\end{subequations}
\Cref{eq:ms-scan} extends \Cref{eq:scan} (equivalent to \Cref{eq:ms-scan-reduction})  with an additional statement (\Cref{eq:ms-scan-iter-update}).
The reduction in \Cref{eq:ms-scan-reduction} is {\em dependent}: 
the value of the reduction $B[i]$ depends on the set of values $\{ A[j] | j 
\le i \}$, while $A[i]$ depends on the previous value of the reduction 
$B[i-1]$.

Dependent reductions pose a challenge to ST because 1) applying ST introduces new dependences, and 2) the newly introduced dependences together with the program's existing dependences may incorrectly form a dependence cycle in the resultant program.

For example, applying ST to \Cref{eq:ms-scan-reduction} with the reuse vector $[-1, 
0]^\transpose$ produces a program consisting of three statements: 
\refeqscanoptrev, which, together form a dependence cycle.
Specifically, let $E_1 \overset{S}{\longrightarrow} E_2 $ denote a dependence induced by statement $S$ between array entries $E_1$ and $E_2$. The path 
$
B[N-1] \overset{\cref{eq:scan-opt-rev-init}}{\longrightarrow} 
A[N-1]  \overset{\cref{eq:ms-scan-iter-update}}{\longrightarrow} 
B[N-2] \overset{\cref{eq:scan-opt-rev-update}}{\longrightarrow}  
B[N-1]$ 
forms a~cycle.

\begin{wrapfigure}{r}{.4\linewidth}
\begin{clisting}[caption={Optimized prefix sum with dependent 
reductions},label=lst:ms-scan-opt]
B[0] = A[0]`\label{line:ms-scan-opt-init}`
for(i=1; i < N; i++)
	B[i] = B[i-1] + A[i]`\label{line:ms-scan-opt-update}`
	A[i+1] = f(B[i])
\end{clisting}
\end{wrapfigure}

On the other hand, applying ST to \Cref{eq:ms-scan-reduction} with reuse 
vector 
$[1, 0]^\transpose$ produces a valid program consisting of three statements: 
\refeqscanopt, and without any dependence cycle. 
\Cref{lst:ms-scan-opt} presents a translation of this program to an imperative 
language with loops, which correctly computes array $A$ and $B$ with complexity 
$\BigO(N)$.

In summary, one key challenge for optimizing the dependent reductions is to augment 
ST to 
choose reuse vectors such that the augmented transformation produces programs 
that have no dependence cycles.

\paragraph{Approach}
\label{sec:contributions}
In this work, present a new technique to automatically optimize dependent reductions 
while soundly handling dependences
that can automatically generate the code in 
\Cref{lst:ms-scan-opt}. 
We present a heuristic algorithm whose key idea is to use\rline{finalintro} an
\hlone{myellow}{affine sequential schedule} of the program as a guide to choose 
among the multiple choices that can be made 
during the optimization process.
Our results show that even 
though the algorithm does not consider other
viable choices during optimization, 
given an \hlone{myellow}{affine sequential schedule} of the program and all 
left-hand-side arrays of reductions,
the algorithm is still optimal for reductions with 
operators that have inverses.

\hl{morange}{We note\rline{notenotfindreuse} that our work relies on the techniques of
\mbox{\citet{sr}} to find 
resuable shared computation -- i.e., reuse vectors -- for a single reduction.}\footnote{We include a description of the techniques in 
\Cref{sec:simplifying-reduction}}
\hl{morange}{Our work addresses dependent reductions by further constraining reuse vectors 
to satisfy (intra/inter)-statement(s) dependences.
}

 \paragraph{Applications} Simplifying Reductions is a classic problem in 
 the compiler optimization literature\todo{add citations}~ that has reemerged as a primary 
 concern for modern applications. 
 In this work, we study a suite of 11 probabilistic inference algorithms that have been established as widely studied and used algorithms across data science, artificial intelligence, machine 
 learning, computer vision, physics, and medicine. 
 We demonstrate that dependent reductions exist in these 
 algorithms' natural, mathematical specifications. 
 Moreover, delivering efficient implementations of these algorithms by 
 hand  -- as is current practice -- requires solving the 
 simplifying dependent reduction problem by hand, 
 which is a tedious and error-prone endeavor.  
 Our approach shows that it is possible to automatically generate 
 optimized and efficient algorithms from their mathematical specifications 
 alone.
 
\paragraph{Contributions}
In this work, we present the following contributions:
\begin{itemize}[leftmargin=*]
	\item  We identify the problem of simplifying dependent reductions, a problem that was not addressed by the Simplifying Reductions framework \citep{sr}, which did not consider dependences. 
	We illustrate the importance of this problem with examples from real~applications.
	\item  We formalize the task of optimizing a dependent reduction by combining the 
	insights of the Simplifying Reductions framework with insights from ILP 
	scheduling~\cite{ilpschedulemultidim}. 
	We formulate a specification of the problem as a integer 
	bilinear program.
	\item  We propose a heuristic algorithm to solve the above optimization 
	problem.
	\item  We evaluate our proposed method on a benchmark suite 
	consisting of
	standard 
	probabilistic inference algorithms and probabilistic models.
	Our results show that our approach reduces the complexity of the reductions in 
	our programs to their optimal complexity for all of the 11 programs that we evaluate.
	In 10 out of the 11 programs, the complexity improves by a 
	(multiplicative) factor of at least $N$, where $N$ is the size of the input data.\footnote{For programs we consider, for example, this is usually the number of 
	data points or the number of words of a text corpus. We include a 
	more detailed review of input sizes for each benchmark in 
	\Cref{sec:runtimevalidation}}
	This is significant because for typical real application inputs of the 
	programs 
	in consideration, $N$ is in the range of $10^4$ to $10^6$ -- a factor that 
	subsumes 
	other potential constant-factor improvements. 
	We also confirm this significance by showing that the speedups in wall-clock 
	time ranges from $1.1\text{x}$ to over $10^6\text{x}$, with a median 
	of 
	$37\text{x}$.
	We also outline the limits of the optimality of our approach, noting that our 
	technique is not optimal if a reduction operator lacks an inverse operation. 

\end{itemize}

In summary, dependent reductions are a key ingredient of probabilistic 
inference algorithms, which are driving an emerging class of new programming 
languages and 
systems~\cite{stan,augurv2,webppl,pyro,gen,hakaru,venture,edward} designed 
to streamline data science and enable new applications.
Optimizing these algorithms has historically either been done by hand or has 
been 
baked in as a domain/algorithmic-specific optimization for a single problem 
model~\cite{dmm,gibbscollapsed}.
To the best of our knowledge, our results are the first to identify and formulate 
the dependent reduction as a general program pattern, detail its challenges, 
and propose a technique to optimize its performance.

\paragraph{Road Map}
\label{sec:roadmap}
In \Cref{sec:example}, 
we illustrate a heuristic algorithm to address the simplifying dependent 
reduction problem described in \Cref{sec:introduction}.
In addition, to further motivate the problem in the context of existing well-known 
algorithms, we present another motivating example which will be used for evaluation 
later in the paper.
In \Cref{sec:background,sec:background-sr}, we review background on the 
polyhedral model and SR, respectively.
In \Cref{sec:mssrps} we formalize our problem as a integer bilinear program.
In \Cref{sec:mssrsol} we introduce our heuristic algorithm.
In \Cref{sec:implementation,sec:eval} we discuss the implementation of our algorithm and its evaluation.
In \Cref{sec:relatedwork,sec:conclusion} we summarize related work and end with 
concluding remarks, respectively.

\section{Example}
\label{sec:example}

In this section we give two examples. 
In \Cref{sec:mssr_example}, we walk through our approach with the 
prefix sum example.\todo{point to the equation for this example}~
In \Cref{sec:motivatingexample}, we use a practical application example 
to further motivate the importance of the dependent reduction problem.
\subsection{Walk Through}
\label{sec:mssr_example}

\ifthenelse{\isundefined{\ispopl}}%
{
\begin{figure}[ht]
	\centering
	{\includegraphics[width=0.8\linewidth]{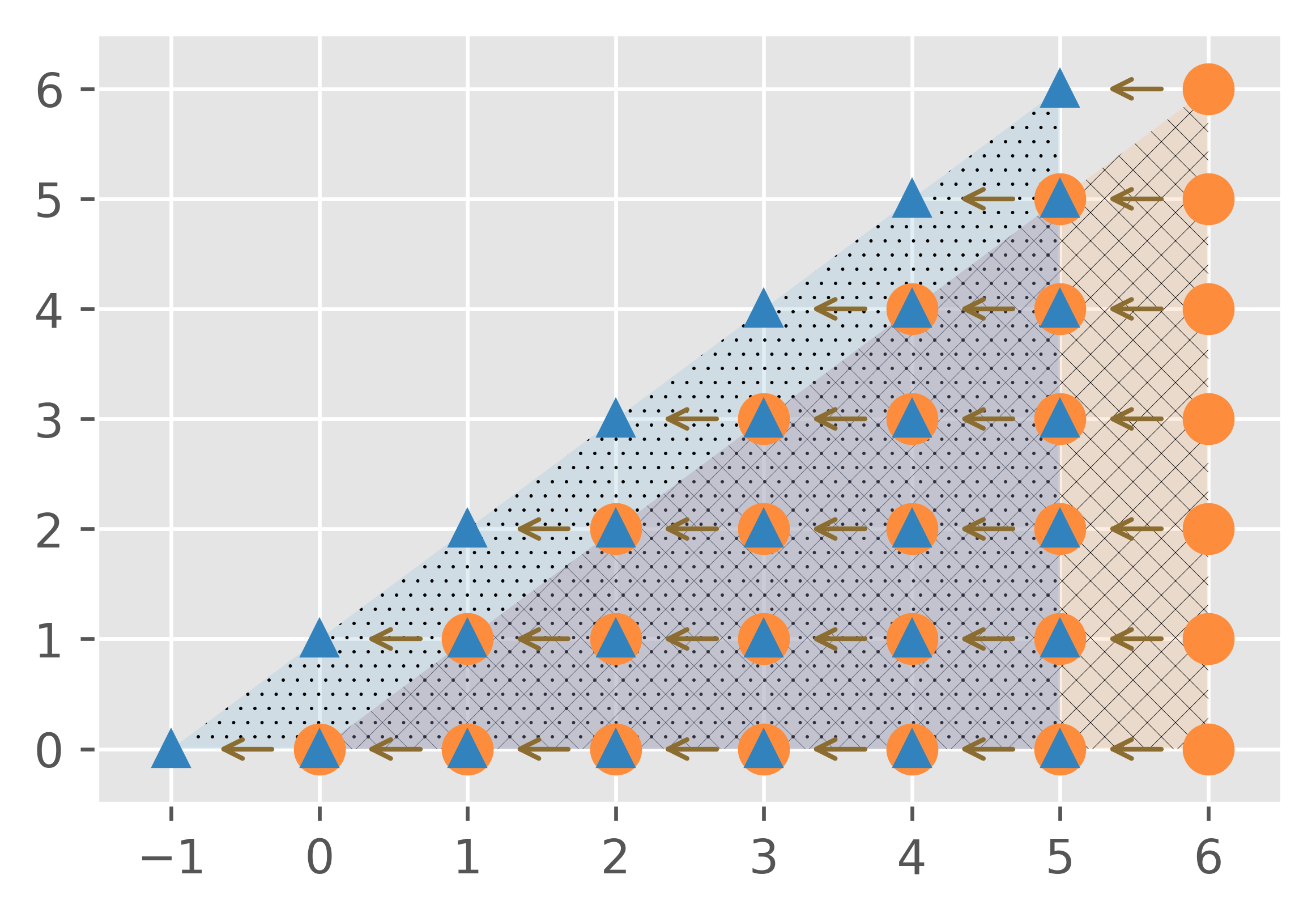}}%
	\caption{Naive prefix sum (\Cref{eq:ms-scan})}
	\label{fig:ms-orig}
\end{figure}%
}%
{
\begin{wrapfigure}[13]{r}{0.4\linewidth}
	\vspace{-2em}
	\centering
	\raisebox{0pt}[\dimexpr\height-\baselineskip\relax]%
	{\includegraphics[width=\linewidth]{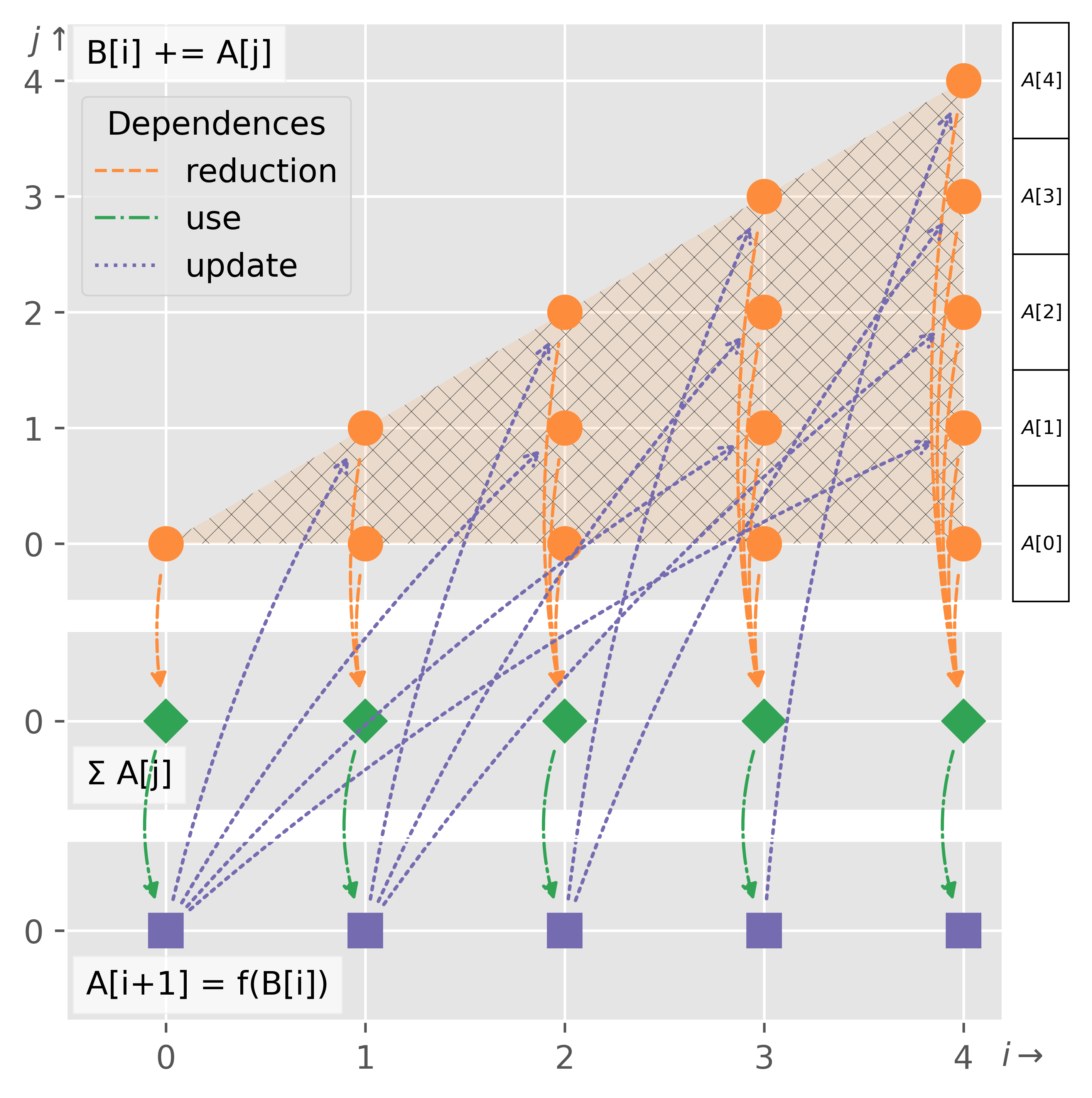}}
	\caption{Naive prefix sum (\Cref{eq:ms-scan})}
	\label{fig:ms-orig}
\end{wrapfigure}%
}

In this section, we use the example of \Cref{eq:ms-scan} to 
1) illustrate the steps of applying ST given a reuse direction,
2) illustrate an invalid reuse direction that leads ST (using 
the algorithm of \citet{sr}) to induce a    
dependence cycle, and compare it to a valid reuse direction, and 
3) describe the mechanism of our proposed heuristic algorithm.\todo{elaborate on what the heuristic does. e.g., " and it chooses a valid reuse direction."}

\paragraph{Naive Prefix Sum}
For ease of comparison and better visualization, we present the input in 
\Cref{eq:ms-scan} with \Cref{fig:ms-orig}, a visual, polyhedral interpretation of the 
naive prefix sum 
program in~\Cref{eq:ms-scan}.
In \Cref{fig:ms-orig}, the top polyhedron with round dots represents the 
iteration 
domain of the reduction statement, $B[i] \mathrel{+}= A[j]$, 
with each round dot denoting an iteration instance of the 
statement.
To the right of the top polyhedron, we have labeled each round dot in 
the top polyhedron at coordinate $(i, j)$ by the array element $A[j]$ 
that should be accumulated into $B[i]$.
The bottom polyhedron with squares represents the iteration domain for 
the 
statement $A[i + 1] = f(B[i])$. 
The middle polyhedron with diamonds is an additional polyhedron 
that our technique inserts 
into the program's polyhedral representation to denote the completion 
of each reduction~$B[i]$, which\rline{bcompletionlabel} we label as 
$\sum\texttt{A[j]}$ in the 
diagram.

\paragraph{Data Dependences}
Each arrow in \Cref{fig:ms-orig} represents a data dependence between iteration 
	instances.
An arrow from iteration instance $a$ to instance $b$ represents a data 
dependence from $a$ to $b$. 
The implication is that $a$ needs to execute before $b$.

There are three sources of data dependences:
\begin{itemize}[leftmargin=*]
	\item {\bf Reduction.} Each point in the middle polyhedron depends 
	on all the points
	in the respective column of the top polyhedron. These dependences 
	are those of 
	the reduction.
	\item {\bf Use.} Each point in the bottom polyhedron depends on the 
	point in the 
	corresponding column 
	of the middle polyhedron. These dependences are those from the 
	use of the 
	reduction results.
	\item {\bf Update.} 
	Points in each row of the top polyhedron depend on the point in the 
	bottom 
	polyhedron that is one to the left of the leftmost point of the row. 
	These dependences are those induced by the update to $A[i+1]$ in 
	\Cref{eq:ms-scan-iter-update} and use by \Cref{eq:ms-scan-reduction}.
\end{itemize}

\paragraph{Correct Optimization}

{
\captionsetup[sub]{skip=0pt}
\captionsetup[figure]{skip=0pt}
\begin{figure*}
	\centering
	\begin{subfigure}{.5\linewidth}
		\centering
		\includegraphics[width=0.8\linewidth]{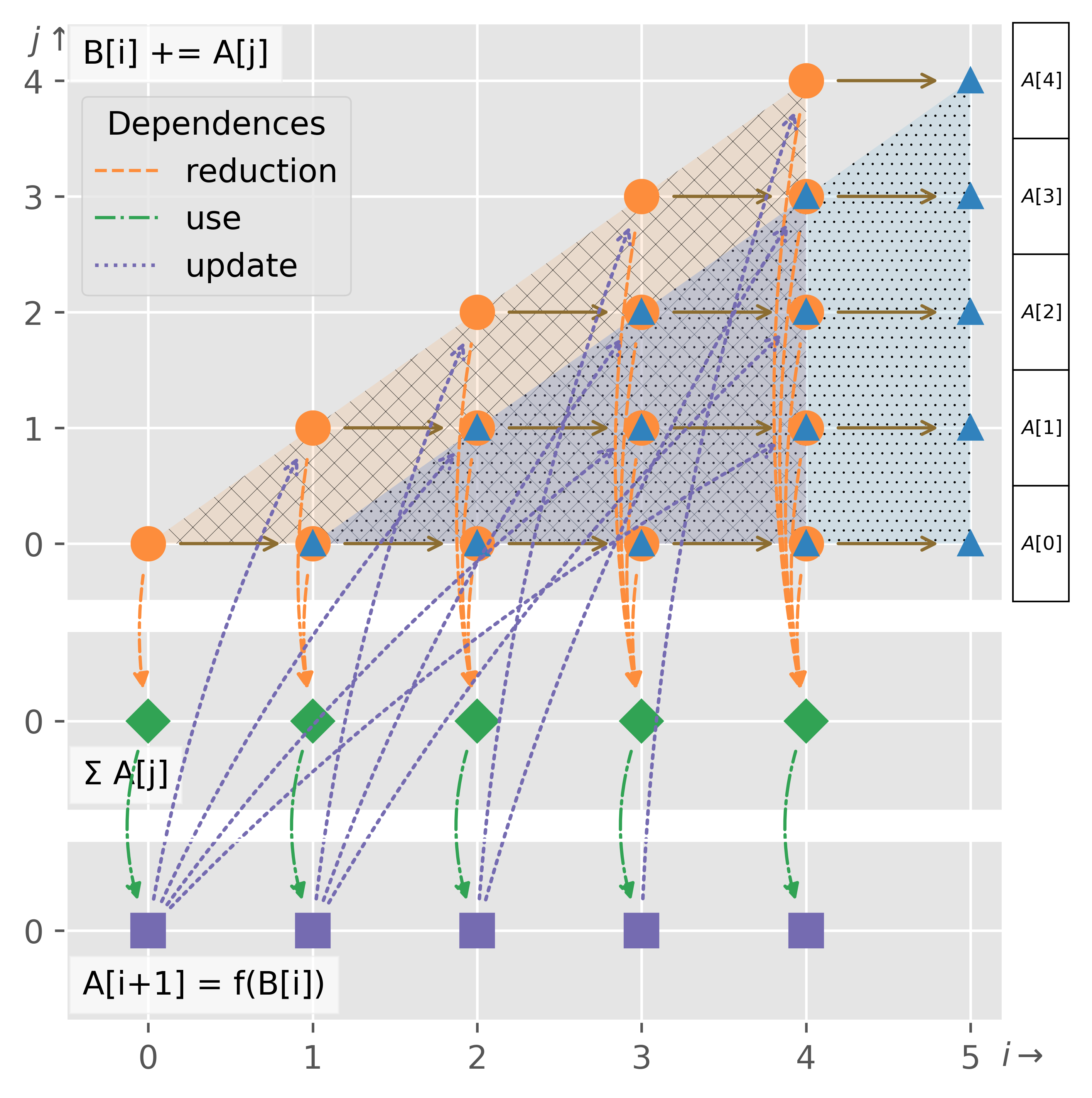}
		\caption{After shift}
		\label{fig:s1}
	\end{subfigure}%
	\begin{subfigure}{.5\linewidth}
		\centering
		\includegraphics[width=0.8\linewidth]{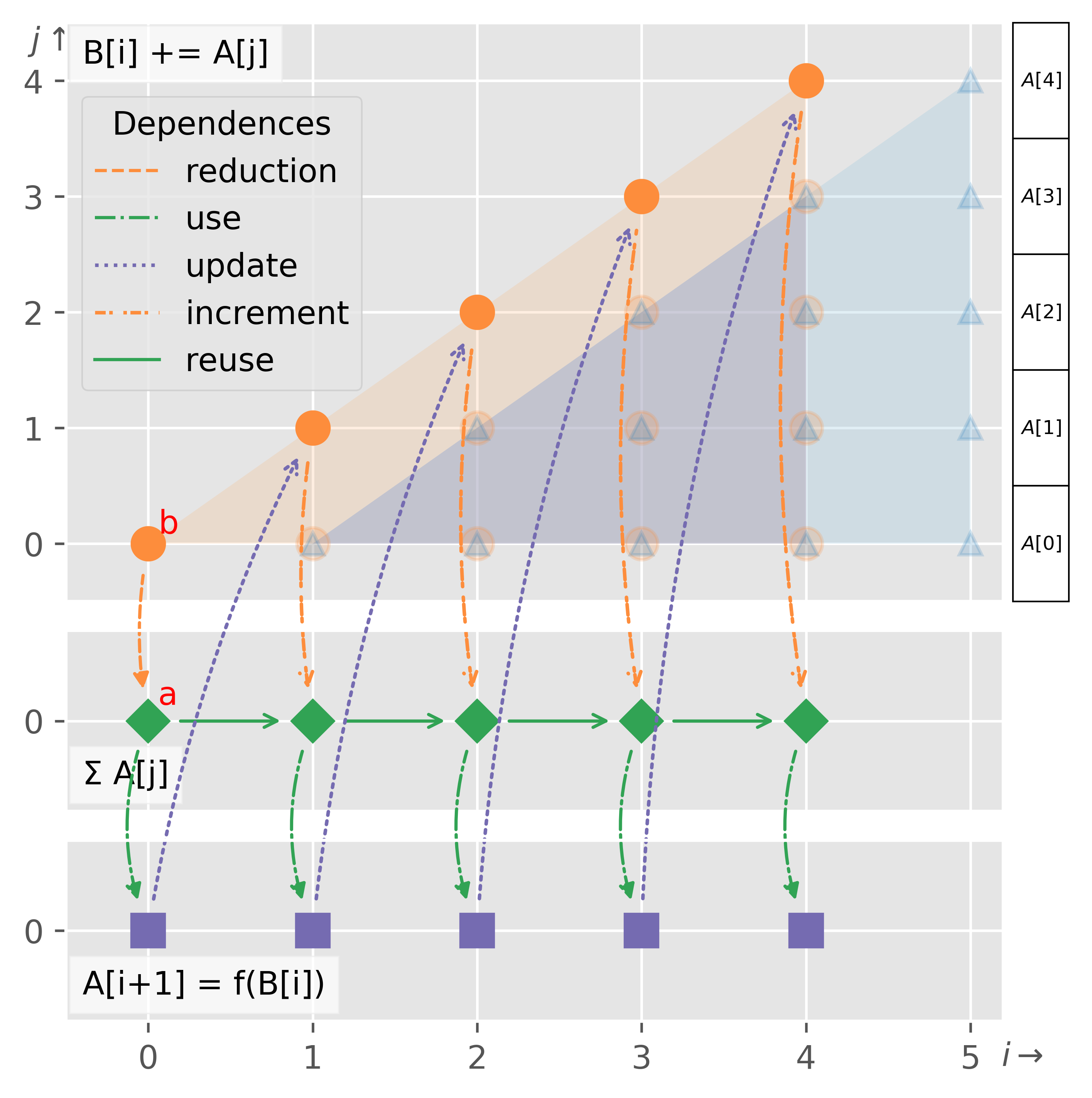}
		\caption{After transformation}
		\label{fig:s1dep}
	\end{subfigure}
	\caption{Correct optimization of prefix sum with depedent reduction}
	\label{fig:correctoptdiagram}
\end{figure*}
}

\Cref{fig:correctoptdiagram} presents two diagrams, corresponding to an intermediate 
step and the result of a correct application of ST, respectively.

\Cref{fig:s1} presents the intermediate step of ST.
In this step, the algorithm chooses a reuse vector and shifts the reduction statement's 
iteration domain along the vector.
\Cref{fig:s1} illustrates the shift along a correct reuse vector, $[1, 
0]^\transpose$, which maps iteration instances $[i, j]$  (round dots in the 
crosshatched 
polyhedron) to instances $[i+1, j]$ (triangles in the dot-shaded polyhedron).
Each solid arrow represents the mapping from an instance in the
polyhedron with round dots to its counterpart in the polyhedron with 
triangles.
The triangles outlined in round dots (i.e. overlapped triangles and round dots) are 
the points corresponding to redundant computations.
Specifically, because the reuse vector has the property that the evaluations of 
the reduction's 
body, $A[j]$, are the same for any two points in the same row, the 
evaluation of 
a reduction over any column $\textit{col}$ in this intersection  
must have the same value as the evaluation of a reduction over the column to 
the left of $\textit{col}$. 
ST therefore eliminates this intersection part of the domain by 
reusing previously computed reductions (i.e. computes $B[i]$ from 
$B[i-1]$ by incrementally using points not in the intersection).

\Cref{fig:s1dep} presents the result of ST, where ST has eliminated the redundant 
computations. 
\Cref{fig:s1dep} corresponds to the resulting polyhedron and 
dependences after ST eliminates redundant computations. 
Specifically, each instance in the intersection of the two polyhedrons has been 
eliminated, along with its induced dependences.

To map \Cref{fig:s1dep} to \Cref{eq:scan-opt}:
the point $a$ in the middle polyhedron and the point $b$ in the top polyhedron in 
\Cref{fig:s1dep} correspond to 
the reduction that initializes $B[0]$ in \Cref{eq:scan-opt-init}.
All points in the middle polyhedron except $a$ then correspond to 
\Cref{eq:scan-opt-update}, i.e., each $B[i]$
is computed by adding the predecessor point $B[i-1]$ with $A[i]$.

Dependences in \Cref{fig:s1dep} are preserved from 
\Cref{fig:s1} for all the non-eliminated instances. 
\Cref{fig:s1dep} introduces the following new dependences: %
\begin{itemize}[leftmargin=*]
\item {\bf Reuse.} Each instance in the middle 
polyhedron (except the leftmost instance) now depends on the instance to its 
left, along the reuse vector. 
These dependences are those from reusing.
\item {\bf Increment.} Each instance in the middle polyhedron now  
depends on the corresponding instance in the top polyhedron in the same 
column. These dependences are those from incrementalizing.
\end{itemize}

\paragraph{Incorrect Optimization}
{
\captionsetup[sub]{skip=0pt}
\captionsetup[figure]{skip=0pt}
\begin{figure*}
	\centering
	\begin{subfigure}[b]{0.5\linewidth}
		\centering
		\includegraphics[width=0.8\linewidth]{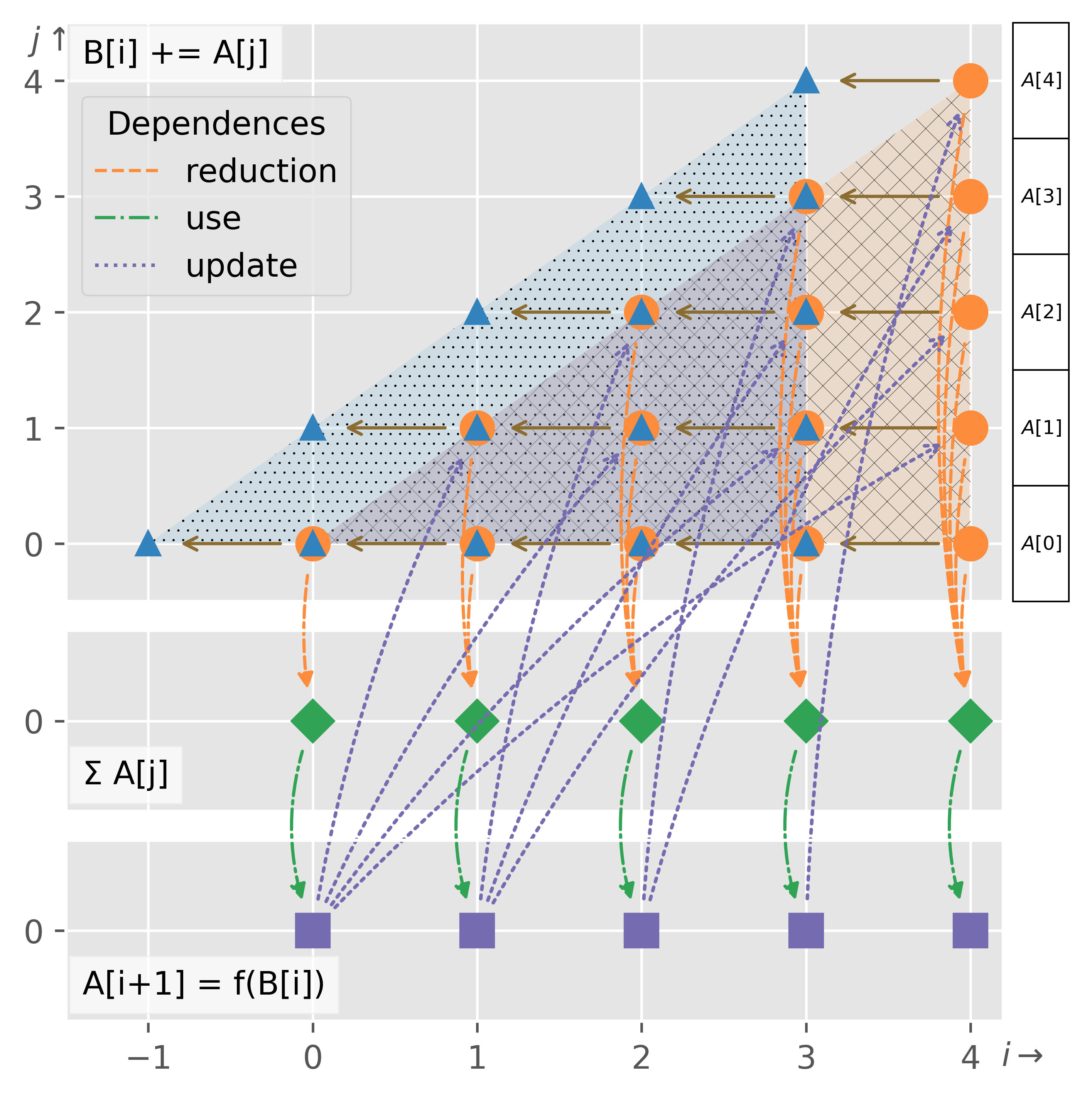}
		\caption{After shift}
		\label{fig:sminus1}
	\end{subfigure}%
	\begin{subfigure}[b]{0.5\linewidth}
		\centering
		\includegraphics[width=0.8\linewidth]{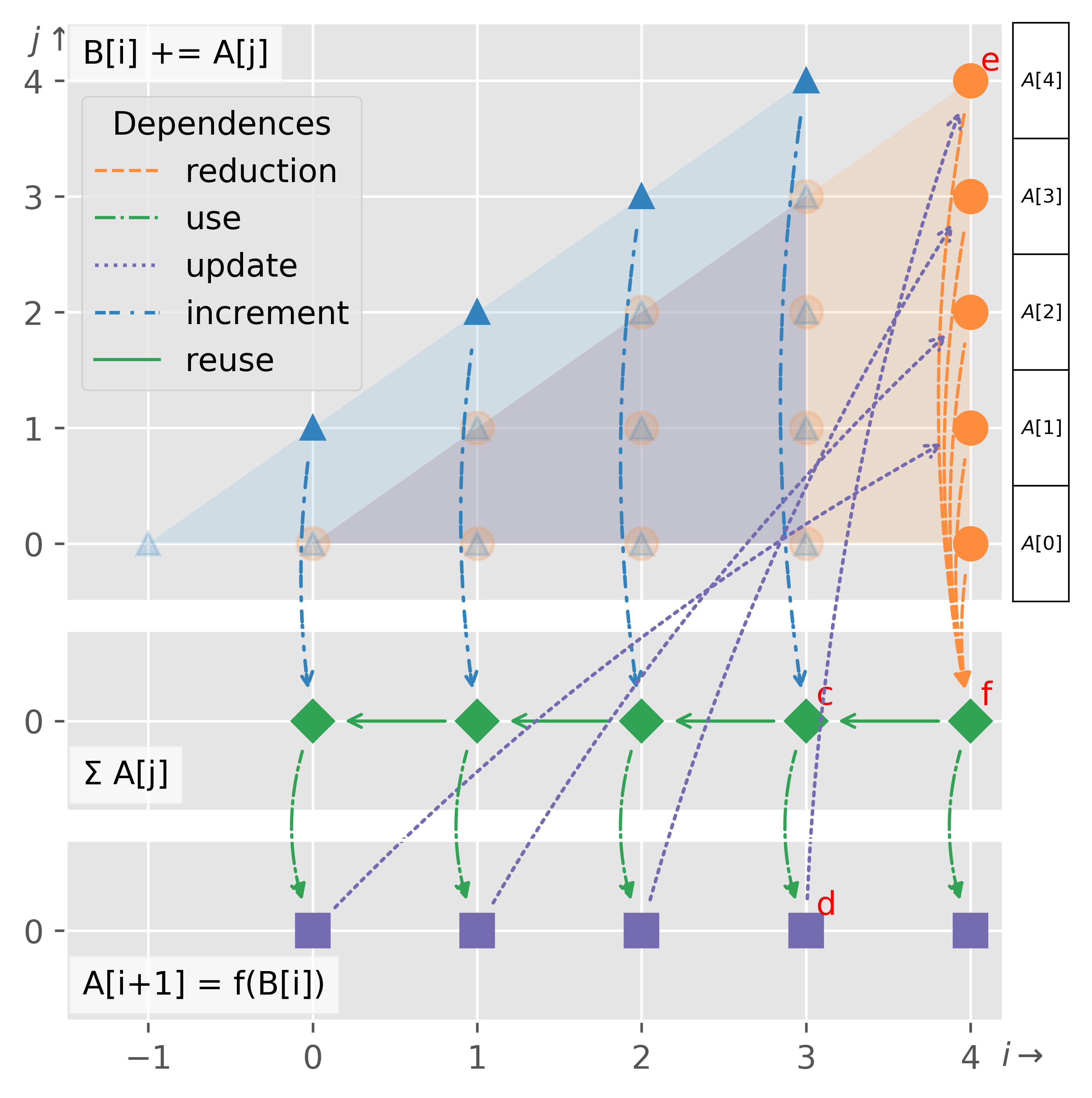}
		\caption{After transformation}
		\label{fig:sminus1dep}
	\end{subfigure}
	\caption{Incorrect optimization of prefix sum with depedent reduction}
	\label{fig:incorrectoptdiagram}
\end{figure*}
}
The two diagrams in
\Cref{fig:incorrectoptdiagram} illustrate an incorrect 
application of ST using \citet{sr}, which ignores the dependences due to 
dependent reduction.
In this case, instead of using the correct reuse vector $[1, 0]^\transpose$, this 
application of ST uses the vector $[-1, 0]^\transpose$.
This vector maps iteration instances $[i, j]$ to instances $[i-1, j]$.

\Cref{fig:sminus1} presents the intermediate step of ST.
Same as in the correct optimization's case, in this step, the algorithm chooses the a 
reuse vector and shifts the reduction's domain along the vector.
\Cref{fig:sminus1} illustrates the shift along the incorrect reuse vector, $[-1, 
0]^\transpose$, which maps iteration instances $[i,j]$ (round dots in the crosshatched 
polyhedron) to instances $[i-1, j]$ (triangles in the dot-shaded polyhedron).
Each solid arrow again represents the maping between the corresponding instances 
before and after the shift.
The triangles outlined in round dots again are the instances corresponding to 
redundant computations to be eliminated by ST.  

\Cref{fig:sminus1dep} presents the result of ST, where ST has, again, eliminated the 
redundant computations.
The top polyhedron now has the round dots at the rightmost column 
and the triangles along the hypotenuse of the shifted domain.
Note that the hypotenuse is restricted to the domain of projected 
domain of the reduction and does not include the point $[i, j] = [-1, 
0]$. 

To map the \Cref{fig:sminus1dep} to \Cref{eq:scan-opt-rev}:
the point $f$ and the round dot column in the top polyhedron in 
\Cref{fig:sminus1dep} correspond to 
the reduction that initializes $B[N-1]$ in \Cref{eq:scan-opt-rev-init}.
All points in the middle polyhedron except $d$ then correspond to 
\Cref{eq:scan-opt-rev-update}, i.e., each $B[i]$
is computed by subtracting the successor point $B[i+1]$ by $A[i]$. 

However, as mentioned in \Cref{sec:introduction}, \Cref{fig:sminus1dep}'s 
dependences 
form cycles (e.g., points $c, d, e, f$ form a cycle). 
Therefore, the transformed program does not have a 
valid 
schedule, and consequently the ST application along the reuse vector with mapping 
$[i, j] \rightarrow [i-1, j]$ is an incorrect optimization. 

\paragraph{Heuristic for Choosing a Valid Direction}
As we have seen from the previous illustration, it is important to choose a valid reuse 
vector for dependent reductions. 
In this work, we propose a heuristic algorithm for choosing a valid reuse vector.
Notably, one key difference between  
\Cref{fig:incorrectoptdiagram,fig:correctoptdiagram}
is the dependences drawn on the middle polyhedron. 
Specifically, in the middle polyhedron of \Cref{fig:correctoptdiagram}, the drawn 
dependences on $\texttt{B}[i]$ respect a \hlone{myellow}{sequential}, scheduled computation order of 
$\texttt{B}[i]$ of the 
original program in \Cref{fig:ms-orig}, whereas that of \Cref{fig:incorrectoptdiagram} 
disobeys that scheduled order.
This observation has inspired the heuristic algorithm, which 
always chooses the reuse vector that is consistent with a \hlone{myellow}{sequential} scheduled computation 
order of the left hand side of the reduction. 
We show that the reuse vector chosen with this algorithm is 1) always sound, and 2) 
guarantees optimality if each reduction operator in the target program has an inverse. 

\subsection{Simplying Dependent Reductions in Practice}
\label{sec:motivatingexample}
As we later show in \Cref{sec:eval} by studying a variety of benchmarks, 
dependent reductions appear in the specifications of many 
problems and algorithms across statistics, artificial intelligence (AI), and machine 
learning (ML) with applications to computer vision, physics, and medicine.  
However, the common practice is to develop these algorithms by hand. 
Therefore, our technique offers the opportunity to automatically translate a 
specification to an efficient implementation. 
In this section, we illustrate our technique on a clustering application 
used across statistics, AI, and ML.

\paragraph{Specification and Implementation} 
Consider the following specification of {\em Gibbs 
Sampling} \citep{gibbssampling} on a {\em two-cluster Gaussian Mixture Model}
\citep[see for example,][]{gmm} (GS-2GMM). 
This computation is designed to cluster data points such that similar data points, 
also called {\em observations}, 
are assigned to the same cluster.
The input to GS-2GMM is a float array $\texttt{Obs}$ that represents the
observations. 
The two-cluster Gaussian Mixture Model (GMM) assumes that each 
single observation belongs to one of the two clusters, and that each cluster follows a  
Gaussian distribution.
The Gibbs sampling procedure samples the array $\texttt{Z}$ that represents the 
cluster membership of all given observations. 
It does this by iteratively taking in an old cluster assignment for all 
observations, and resampling 
a new assignment by updating the individual assignment of a single observation. 
This process will produce
a stream of samples of $\texttt{Z}$s that approach the true
distribution of 
$\texttt{Z}$.
The mathematical specification of GS-2GMM is given in \Cref{eq:gs-2gmm}.
\begingroup
\begin{subequations}
\label{eq:gs-2gmm}
\small
\begin{align}
C_{zi} &= \sum_{\forall j \; \text{ s.t. } j \neq i \land Z_j = z } 1 
\label{eq:gs-2gmm-count}  
\quad,\forall z, i \\
S_{zi} &= \sum_{\forall j \; \text{ s.t. } j \neq i \land Z_j = z } 
\textit{obs}_i 
\label{eq:gs-2gmm-sum}  \quad,\forall z, i \\
    P_o(z, i) &= \mathcal{N}\big(\frac{S_{zi}}{C_{zi}}, (1 + C_{zi})^{-1} + 1 \big) \overset{\textit{\tiny{def.}}}{=} 
    P(\textit{obs}_i | \textit{obs}_{\setminus i}, Z_{\setminus i}, Z_i=z)
\label{eq:gs-2gmm-po-dist}  \\
    P_z(i) &=  \frac{ P_o(0, i)}{ P_o(0, i) + P_o(1, i)  } \overset{\textit{\tiny{def.}}}{=} P(Z_i = 0 | Z_{\setminus i}, \textit{obs}) 
\label{eq:gs-2gmm-pz-dist}  \\
    Z_i &\sim P_z(i) \quad, \forall i \in \{1 \ldots N\}
\label{eq:gs-2gmm-sample} 
\end{align}
\end{subequations}
\endgroup

In \Cref{eq:gs-2gmm-count,eq:gs-2gmm-sum}, 
$C_{0i}$ and $S_{0i}$ represent the counts and sums, respectively, of
all the observations except the one with index $i$, for which the current old 
assignment of  cluster membership is 0 
(and similarly for $C_{1i}, S_{1i}$, with membership of 1). 
\Cref{eq:gs-2gmm-po-dist} defines the function $P_o$ as an abbreviation of a normal distribution that represents the distribution of $\textit{obs}_i$ given all values of $\textit{obs}$ except $\textit{obs}_i$ and all current assignments of $Z$s except fixing $Z_i$ to be $z$.
We use the notation $\setminus i$ to denote the set $\{j \;|\; 1 \le j \le N \land j \neq i \}$.
\Cref{eq:gs-2gmm-pz-dist} defines the function $P_z$ as an abbreviation of a probability representing the chance tha $Z_i$ is equal to $0$ given all values of $\textit{obs}$ and all current assignments $Z$ except $Z_i$.
Lastly, \Cref{eq:gs-2gmm-sample} samples each $Z_i$ in order from its distribution.
Note that the exact computations required to perform Gibbs sampling are not important 
for understanding the optimization problem.

\Cref{lst:gs-2gmm-correct} gives an efficient implementation of the above 
mathematical specification -- notice that \Cref{lst:gs-2gmm-correct} computes 
the counts and sums incrementally, instead of forming the full reductions of 
\Cref{eq:gs-2gmm-count,eq:gs-2gmm-sum}. 
Deriving \Cref{lst:gs-2gmm-correct} from \Cref{eq:gs-2gmm} requires manually 
solving the simplifying dependent reductions problem which is tedious and 
error-prone. 

\paragraph{Our Approach}
\begingroup
\newcommand{\ieqtv}[3]{\ensuremath{(#1 = #2 ? #3 : 0)}}
\newcommand{\sampleref}{%
\Cref{eq:gs-2gmm-po-dist,eq:gs-2gmm-pz-dist,eq:gs-2gmm-sample}%
\label{line:gs-2gmm-correct-abstract-sample}}
\begin{clisting}[caption={Correct optimized GS-2GMM with dependent 
reduction; \texttt{(a ? b : c)} denotes\rline{ifthenelseC} if-then-else 
expression as in the C 
language},label=lst:gs-2gmm-correct,float={h},xleftmargin=4ex,aboveskip=1.5ex]
int[N] C0L, C1L, C0R, C1R = {0...} // Zero initialize
float[N] S0L, S1L, S0R, S1R = {0...} // Zero initialize
for(i = 1; i < N; i++)
	C0R[0] += (Z[i] == 0 ? 1 : 0)
	C1R[0] += (Z[i] == 1 ? 1 : 0)
	S0R[0] += (Z[i] == 0 ? Obs[i] : 0)
	S1R[0] += (Z[i] == 1 ? Obs[i] : 0)
for(i = 0; i < N; i++)
	// Sample according to `\sampleref`
	Z'[i] = sample(C0L[i] + C0R[i], C1L[i] + C1R[i], `
	\ifthenelse{\isundefined{\ispopl}}%
	{\newline\text{\phantom{xxxxxxxxxxxxxxx}S0L[i] + S0R[i], S1L[i] + S1R[i])}}%
	{\text{S0L[i] + S0R[i], S1L[i] + S1R[i])}}`
	// Incremental updates
	C0L[i] = C0L[i-1] + (Z'[i] == 0 ? 1 : 0)`\label{line:increment}`
	C1L[i] = C1L[i-1] + (Z'[i] == 1 ? 1 : 0)
	S0L[i] = S0L[i-1] + (Z'[i] == 0 ? 1 : 0)
	S1L[i] = S1L[i-1] + (Z'[i] == 1 ? 1 : 0)
	C0R[i] = C0R[i-1] - (Z[i] == 0 ? 1 : 0)
	C1R[i] = C1R[i-1] - (Z[i] == 1 ? 1 : 0)
	S0R[i] = S0R[i-1] - (Z[i] == 0 ? Obs[i] : 0)
	S1R[i] = S1R[i-1] - (Z[i] == 1 ? Obs[i] : 0)
\end{clisting}
\endgroup
Given an array-based representation of \Cref{eq:gs-2gmm}, our approach 
automatically produces \Cref{lst:gs-2gmm-correct}.
For conciseness of presentation, we consider the variable $S_{zi}$ with 
fixed $z=0$ as an example.
In this case \Cref{eq:gs-2gmm-sum} can be rewritten as sum of two variables $S_{0i} 
= \texttt{S0L}[i] + \texttt{S0R}[i]$, where $\texttt{S0L},\texttt{S0R}$ are given by 
\Cref{eq:gs-2gmm-s0l-reduction,eq:gs-2gmm-s0r-reduction}, respectively.
\rline{ifthenelseequation}\begin{subequations}
\small
\allowdisplaybreaks
\begin{align}
\texttt{S0L}[i] &= \sum_{j=0}^{j<i} (\text{if } \texttt{Z}[j] == 0 \text{  
then  } \texttt{Obs}[j] \text{  else  } 0) 
\label{eq:gs-2gmm-s0l-reduction}\\
\texttt{S0R}[i] &= \sum_{j=i+1}^{j<N} (\text{if } \texttt{Z}[j] == 0 \text{  
then  } \texttt{Obs}[j] \text{  else  } 0) 
\label{eq:gs-2gmm-s0r-reduction} \\
&\vdots\quad\text{Other equations...} \nonumber \\
\texttt{Z'}[i] &= \texttt{sample}(\texttt{S0L}[i] + \texttt{S0R}[i], ...) 
\label{eq:gs-2gmm-iter-update} 
\end{align}
\end{subequations}

The step of rewriting in terms of $\texttt{S0L}$ and $\texttt{S0R}$ is standard in 
polyhedral model compilation: the original domain with constraint $j \neq i$ is 
non-convex and it is standard to break it into two convex polyhedrons with 
constraints $j<i$ and $j>i$.
Further, we make the non-affine constraint $Z_j = z$ into a simple 
if-then-else 
expression guarding the reduction's body -- this is standard approach and same 
as the one proposed by \citet{polyhedralmodelismoreapplicable} to model 
non-affine constraints as control predicates. 

\Cref{eq:gs-2gmm-s0l-reduction,eq:gs-2gmm-iter-update} exactly 
correspond to \Cref{eq:ms-scan-reduction,eq:ms-scan-iter-update}, 
respectively, since they have the same data flow dependences 
\footnote{Although \Cref{eq:gs-2gmm-iter-update} contains 
\texttt{sample} that 
is stochastic and \Cref{eq:ms-scan-iter-update} contains \texttt{f} that 
is 
deterministic, they still have the same data flow dependences. }. 
Thus the technique walked through in \Cref{sec:mssr_example} also applies to 
\Cref{eq:gs-2gmm-s0l-reduction,eq:gs-2gmm-iter-update} to produce
a specification with efficient complexity.
Further, our technique is general in that it handles any dependent reduction, 
including \Cref{eq:gs-2gmm-s0r-reduction} with constraints $i+1\le j < 
N$, which are the reverse of the constraints in 
\Cref{eq:gs-2gmm-s0l-reduction} .
Lastly, the same analysis can be applied to all cases of $C_{zi}$ and $S_{zi}$ with 
$z=0$ or $z=1$.
The analyses in total produces eight intermedieate variables, namely 
$\texttt{C0L}$, 
$\texttt{C1L}$, 
$\texttt{C0R}$, $\texttt{C1R}$, $\texttt{S0L}$, $\texttt{S1L}$, $\texttt{S0R}$, 
$\texttt{S1R}$, which produce \Cref{lst:gs-2gmm-correct} by applying our technique and compiling to exectuable code.

\paragraph{Results} Our evaluation in \Cref{sec:eval} shows that our technique produces an optimal 
complexity algorithm for Gibbs Sampling on the Gaussian Mixture Model, 
matching that of a manually developed implementation, and yielding a $7.1$x 
performance improvement over a naive, unoptimized implementation.
These results demonstrate the opportunity to automatically compile high-level 
specifications that include dependent reductions to efficient 
implementations.

\section{Background: Polyhedral Model}
\label{sec:background}

In this section, we review the terminology from the {\em polyhedral model} 
\todo{citationo} that we use in this work. \todo{MC: explain the point of the 
polyhedral model}
The polyhedral model represents a program by a set of statements, and for each statement, an associated {\em polyhedral set} known as the statement's {\em domain}. Each point in a polyhedral set corresponds to one concrete execution instance of the statement. 

\begingroup
\newcommand{\btmp}{\texttt{BTmp}}
\newcommand{\reusevec}{$[1,0]^\transpose$}
\newcommand{\SOneFinalizeAndSTwo}{
	\texttt{S1Fin}: \texttt{B}[i] = \texttt{\btmp}[i] : 
	\polyhedral{[i] : 0  \le i < N } \\
	\texttt{S2}: \texttt{A}[i+1] = f(\texttt{B}[i])  : 
	\polyhedral{[i] : 0  \le i < N - 1 }}
\newcommand{\lstleft}{
	\arraycolsep=0.3pt
	\begin{array}{l}
	\texttt{S1}: \texttt{\btmp}[i] \mathrel{+}=\texttt{A}[j]  : 
	\polyhedral{[i, j] : 0 \le i < N \wedge 0  \le j \le i } \\
	\SOneFinalizeAndSTwo
	\end{array}
}
\begin{clisting}[caption={\Cref{eq:ms-scan} in the polyhedral IR}, 
numbers=none,xleftmargin=0pt,label={lst:ms-scan-ir},float={t},aboveskip=1.5ex]
`
$\lstleft$
`
\end{clisting}
\endgroup

\Cref{lst:ms-scan-ir} gives an example of \Cref{eq:ms-scan} our polyhedral intermediate representation.
Each line is a polyhedral statement to which we have affixed a label to the left: 
\texttt{S1}, \texttt{S1Fin}, and \texttt{S2}, respectively. The polyhedral statement
\texttt{S1} corresponds to \Cref{eq:ms-scan-reduction}, and  \texttt{S2} corresponds to \Cref{eq:ms-scan-iter-update}.
To aid understanding, we have added a statement, \texttt{S1Fin}, that does not map directly to an equation
in \Cref{eq:ms-scan}. The statement denotes the completion of the computation (the reduction) for each element of  
 \texttt{B} and corresponds to the middle polyhedra in \Cref{fig:ms-orig,fig:incorrectoptdiagram,fig:correctoptdiagram}.

\subsection{Polyhedral Representations}

\paragraph{Polyhedral Set} 
The notation that follows the second colon of each line of \Cref{lst:ms-scan-ir} 
denotes the 
{\em polyhedral set} that defines the statement's domain. The statement
executes once for each point in the set. We introduce the following definition and 
notation for a polyhedral set; the notation is consistent with the
Integer Set Library (ISL) \citep{isl}'s notation. 

\begin{definition}[{\em Polyhedral set}]\label{def:polyhedralset}
A polyhedral set $\mathcal{P}$ is defined as $\polyhedral[\vec{p}]{[\vec{x}] : M \cdot 
[\vec{x}, 
\vec{p}, 
1]^\transpose \ge 
	\vec{0}}$, which consists of  a vector\rline{tupletovector1} of parameters  
	$[\vec{p}]$ , 
	 a 
	vector template $[\vec{x}]$, and  a {\em system of affine inequalities} $M \cdot 
			[\vec{x}, \vec{p}, 1]^\transpose 
			\ge \vec{0}$, where $M$ is an $m \times 
		(|\vec{x}|+ |\vec{p}| + 1)$ matrix of 
	constant 
	integers.
In addition, $\polyhedral[\vec{p}]{[\vec{x}]}$ is the called the space of  $\mathcal{P}$. 
\end{definition}

A polyhedral set provides an intensional description of a set of tuples, templated by $[\vec{x}]$, so that all tuples in the set satisfy the system of affine inequalities.
The set is optionally parametric in $[\vec{p}]$, if $[\vec{p}]$ is not empty.
For example, the polyhedral set for statement \texttt{S1Fin} in \Cref{lst:ms-scan-ir} is $\polyhedral[N]{[i]: \begin{bmatrix} 1 & 0 & 0 \\ -1 & 1 & -1 \end{bmatrix} 
\cdot \begin{bmatrix} i \\ N \\ 1 \end{bmatrix} \ge \vec{0}
  }$ that denotes the set of integer values of $i$ from 0 to $N - 1$. %
Specifically, each row of $M$ denotes an inequality. 
Therefore, the inequalities in this example are $(i \ge 0) \land (i \le N - 1)$ --- or simply the shorthand $0 \le i < N$.
An equality $i = 0$ is shorthand\rline{shorthand} for the conjunction of two inequalities $(i \ge 0) \land (-i \ge 0)$. 

\paragraph{Polyhedral Relation} 
\begin{wrapfigure}[8]{r}{0.35\linewidth}
	\vspace{-4ex}
	\centering
	\raisebox{0pt}[\dimexpr\height-1.\baselineskip\relax]%
{\includegraphics[width=\linewidth]{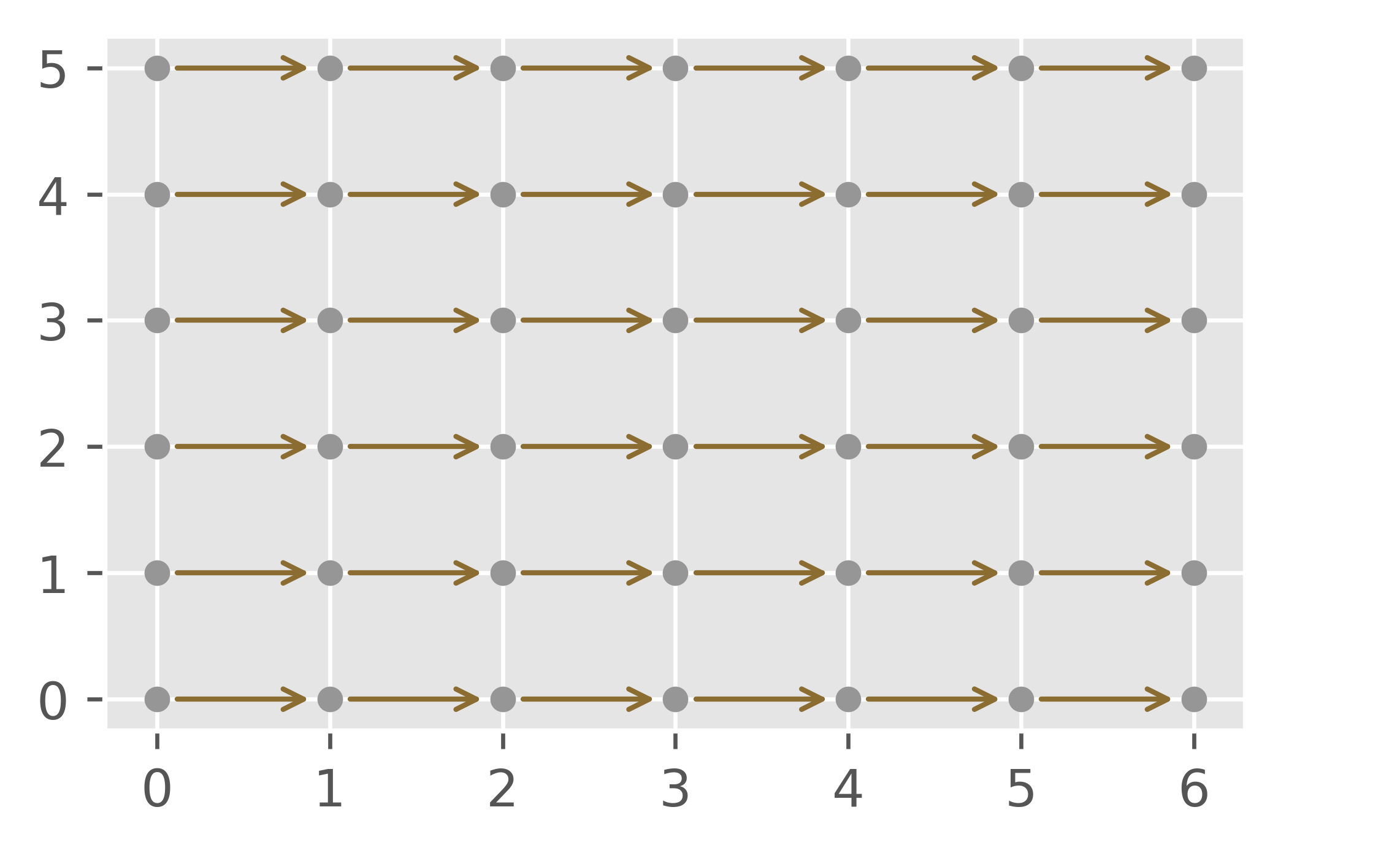}}
	\caption{Plot of example the polyhedral 
	relation\rline{recallpolyhedralrelation} $\polyhedral[N]{[i, j] 
	\rightarrow [i + 1, j]: 0 \le i < N, 0 \le j < N}$}
	\label{fig:polyhedral_relation}
\end{wrapfigure}
To give semantics to our polyhedral intermediate representation (e.g. 
\Cref{lst:ms-scan-ir}), we also introduce the 
following definition for a {\em polyhedral relation}, also in accordance 
with ISL's notation.

\begin{definition}[{\em Polyhedral relation}]\label{def:polyhedralrelation}
A polyhedral relation $\polyhedral[\vec{p}]{ [\vec{x_1}] \rightarrow [\vec{x_2}] : M 
\cdot 
[\vec{x_1}, 
\vec{x_2}, \vec{p}, 1]^\transpose 
\ge \vec{0}}$  contains a vector\rline{tupletovector2} of parameters 
$[\vec{p}]$, vector 
templates $[\vec{x_1}], [\vec{x_2}]$ 
, and a system of affine inequalities $M \cdot [\vec{x_1}, \vec{x_2}, \vec{p}, 
1]^\transpose \ge \vec{0}$.
\end{definition}

A polyhedral relation describes a set of binary relations mapping from $[\vec{x_1}]$ to 
$[\vec{x_2}]$, for every $[\vec{x_1}]$-$[\vec{x_2}]$ pair that satisfies the system of 
affine 
inequalities; a polyhedral relation can also be parametric in $[\vec{p}]$.
For example, 
$\polyhedral[N]{[i, j] \rightarrow [i + 1, j]: 0 \le i < N \land 0 \le j < N}$ 
denotes the 
relation that maps every integer tuple $[i, j]$ to $[i+1, j]$ within an $N$-by-$N$ grid. 
\Cref{fig:polyhedral_relation} visualizes this relation for $N=5$: the 
arrows map points corresponding to integer tuples to their right successors.

For aesthetic reasons, we omit the parameter $[\vec{p}]$ when it is clear which 
identifiers are parameters.

\subsection{Polyhedral Representation of a Program}
\label{sec:ir}

\paragraph{Syntax}
\newcommand{\reduceeq}{\ensuremath{\mathrel{\oplus}=}}
\newcommand{\rdeq}{$\reduceeq$}
\begin{wrapfigure}[4]{r}{0.3\linewidth}
\vspace{-1ex}
\centering
\newcommand{\galt}{\ensuremath{\;\mid\;}}%
\begin{minipage}{0.3\linewidth}
\vspace{-30pt}
\begin{align*}
P &:= S + \\
S &:= X \; ( = \galt \reduceeq ) \; E 
\;\texttt{:}\; 
\mathcal{P} \\
E &:= E ? \oplus E \galt X \galt c \\
X &:= a \; \texttt{[} \; (A \texttt{,}){*} \; A \; \texttt{]}
\end{align*}
\end{minipage}
\caption{IR Grammar}
\label{lst:irstmt}
\end{wrapfigure}
Following the formalization by the original SR work \citep{sr, alphaz}, we 
use 
an equation-based representation of program in this work, presented in grammar by 
\Cref{lst:irstmt}.
We explain each component in turn:

\begin{description}[font=$\bullet$~\normalfont, leftmargin=0.3cm]

\item [$P$] is a program that consists of multiple statements.

\item [$S$] is a statement that consists of a left hand side array access $X$, a middle 
assignment 
operator (i.e. either $=$ or $\reduceeq$ ) , a right hand side expression (i.e. $E$), and 
a domain (i.e. $\mathcal{P}$). 
A statement is called a normal statement when the middle assignment operator is 
plain $=$; 
it is called a {\em reduction} when the middle assignment operator is $\reduceeq$.
The reduction operator $\oplus$ is associative and commutative, and 
has an identity (e.g. $0$ is the identity addition, and $1$ is the identify 
for multiplication).

\item [$E$] is an expression that is either an unary or binary operator applied on 
expression(s), an array access (i.e. $X$), or a constant.

\item[$A$] is an affine expression, a kind of expression that applies an affine 
transformation to 
variables and produces a scalar. It references only variables in $\vec{x}$ or $\vec{p}$, 
where 
$\polyhedral[\vec{p}]{[\vec{x}]}$ is the space of $\mathcal{P}$. 

\item [$X$] is an affine array access that consists of an array identifier (i.e. $a$), 
and a (comma seperated) list of affine expressions (i.e. $A$s).
A list of affine expressions of length $n$ can be expressed mathematically as an 
affine 
transformation $M \cdot [\vec{x}, \vec{p}, 1]^\transpose$, where $M$ is a constant $n 
\times 
(|\vec{x}| + |\vec{p}| + 1)$ integer matrix and $\vec{x}, \vec{p}$ are defined same as 
those 
for an affine expression $A$. 

\item [\(\mathcal{P}\)] is a polyhedral set representing the statement's domain. 
Each point in the domain corresponds to one concrete execution instance of the 
statement. 
If $\mathcal{P}$ is $\polyhedral[p]{[t] : e}$, then $p$ corresponds to the 
set of parameters of the program and $t$ corresponds to the loop variables of 
the statement. 
\end{description}

\paragraph{Semantics}
We use usual semantics from array languages \citep{alphaz} for our IR. Specifically, a 
statement is evaluated under each point of its domain \(\mathcal{P}\). An expression 
is 
evaluated under a point by substituting the free variables of the expression with the 
instantiated 
values of those variables under that point. 
For example,\rline{fixfuntionalnotation} an expression a[N - i + j + 1] 
with domain 
$\polyhedral[N]{[i, 
j] : 0 \le i < N \land 0 \le j < N}$, evaluates to the value of a[9] at the 
point $[i, j]^\transpose = [1, 0]^\transpose$ and given that $N=10$.

If the statement is a normal assignment, for each point in $\mathcal{P}$, 
the right hand side 
expression is evaluated and assigned to the left hand side array. 

If the statement is a reduction, 
we first define the {\em projection} of a reduction to be a polyhedral relation that 
maps points in $\mathcal{P}$ to their accessed indices of the left hand side array.
For example, for the reduction \texttt{S1} in \Cref{lst:ms-scan-ir}, the projection is 
the polyhedral relation $\polyhedral[N]{[i, j] \rightarrow [i] : 0 \le i < N \land 0 \le j 
\le i}$. 
Then, for each point $p \in \mathcal{P}$ the right hand side 
expression is evaluated and accumulated into the left hand side 
array at point $p' = proj(p)$ 
using the operator $\oplus$, where $\textit{proj}$ is the projection of the 
reduction.

\paragraph{Computability}
In general,\rline{ircomputability} programs in the equation-based 
representation 
(\Cref{lst:irstmt}) may not be computable due to cyclic dependences. 
\citet{sarecomputability} 
showed that detecting cycles for these programs is undecidable.

In this work, we focus on programs in the equational-based 
representation (\Cref{lst:irstmt}) that do not have 
cyclic dependences, admit legal schedules, and are always computable.

\newcommand{\conjQ}{\ensuremath{\bigwedge_{\forall i} \mathbf{Q}_i}}

\subsection{Polyhedral Model Scheduling}
\label{sec:ilp_schedule}

Scheduling is a step in the polyhedral model where a scheduling function assigns each point in a statement's domain a timestamp, denoting the order of all execution instances.
The task of scheduling a program in the polyhedral model is to find a schedule $\Theta$ for the program such that the schedule timestamps for all statements satisfy the dependence relations of the program.

\subsubsection{Dependence Analysis}

In scheduling, dependence analysis compute the
{\em dependence relations} of the program.
Here we recall concepts in dependence analysis and breifly describe how to 
compute the dependence relation bewteen two statements in the program.

\paragraph{Access Relation}
An {\em access relation} is a polyhedral relation mapping from the space of the 
domain of a statement to the space of an accessed array. An access relation can either 
be a {\em write access 
relation} (when the access is a write to an array on the left hand side of a statement), 
or a read access relation 
(when the access is a read to an array on the right hand side of a statement).
Let $\textit{array}[A_1, \ldots, A_k]$ be an array reference in a statement with space 
$\polyhedral[\vec{p}]{[\vec{x}]}$, and the list of affine expressions $A_1, \ldots, A_k$ 
is expressed as $M 
\cdot [\vec{x}, \vec{p}, 1]^\transpose$. The 
access relation for this array access is $\polyhedral[p]{[\vec{x}] \rightarrow [\vec{y}] : 
M \cdot 
[\vec{x}, \vec{p}, 1]^\transpose = \vec{y}  }$. 

For example, in \Cref{lst:ms-scan-ir}, since each iteration instance $[i, j]$ of 
\texttt{S1} 
writes 
to  $\texttt{BTmp}[i]$, the write access relation of \texttt{S1} is
$\polyhedral[N]{[i,j] \rightarrow [i] : 0 \le i < N \land 0 \le j < i }$.
Also, since each iteration instance $[i, j]$ of \texttt{S1} reads $\texttt{A}[j]$, the 
read access relation of \texttt{S1} is 
 $\polyhedral[N]{[i,j] \rightarrow [j] : 0 \le i < N \land 0 \le j < i }$.

\paragraph{Array SSA}
Following \citet{sr}, our IR requires the program to be in {\em array 
static-single-assignment} 
(Array SSA) form \citep{feautrierarrayexpansion}; that is, each array element is never 
written 
twice during program execution. 
This means for each unique left hand side array, and the statements 
$\text{S}_0 ... \text{S}_k$ 
that write to it, $\bigcap_i \mathcal{W}_{\text{S}_i} = \emptyset$, where 
$\mathcal{W}_{\text{S}_i}$ is the 
write access relation for $\text{S}_i$. 

\paragraph{Dependence Relation}
\label{sec:dependencerelation}
Any two statements $S_1, S_2$ must satisfy a dependence relation 
represented by a 
polyhedral 
relation 
$\mathcal{D}_{S_1, S_2} = \polyhedral[\vec{p}]{[\vec{x}^{S_1}] 
\rightarrow 
[\vec{x}^{S_2}]: M^{\mathcal{D}_{S_1, S_2}} \cdot \begin{bmatrix} 
\vec{x}^{S_1}, 
\vec{x}^{S_2}, 
\vec{p}, 1 
\end{bmatrix}^\transpose \ge 
\vec{0}}$, where $M^{\mathcal{D}_{S_1, S_2}} $ is the {\em dependence
matrix}. 
The dependence relation $\mathcal{D}_{S_1, S_2}$ describes the 
happens before relation between iterations of $S_1$ and $S_2$. 
For a pair of statements $S_1, S_2$, let $S_1$ write to an array $a$ and
let $S_2$ read from $a$. The dependence relation $\mathcal{D}_{S_1, S_2}$ is 
equal to 
$\mathcal{R}^{-1} \circ \mathcal{W}$, where $\mathcal{R}$ and $\mathcal{W}$ 
are the two statements' read and write access relations w.r.t. $a$, respectively. 
$\mathcal{R}^{-1}$ denotes 
the inverse of the 
polyhedral relation $\mathcal{R}$ and $\circ$ denotes composition. 
Previous work \citep{fuzzyarray, ondemandarraydatafolow} and a textbook 
\citep{presburgerformulaandpolyhedralcompilation} contain detailed 
introductions to dependence analysis techniques, which we refer the 
reader to for deeper exposure.

\subsubsection{Scheduling Function}
\label{sec:schedulefn}
Here we recall the definition of a program's schedule.

\begin{definition}[Schedule Timestamp]\label{def:scheduletimestamp}
A schedule timestamp is an $m$-dimensional vector, where $m$ is the upper bound 
on the 
dimension of 
the schedule. 
For two timestamps $\tau_1$ and $\tau_2$, $\tau_1 < \tau_2$ 
($\tau_1$ happens before $\tau_2$) iff 
$\tau_1[i] < \tau_2[i]$ where $i$ is the first non-equal index between 
$\tau_1, \tau_2$. 
\end{definition}

A schedule $\Theta$ for a program is a collection of {\em scheduling 
functions}, one for each statement.
A scheduling function for a statement $S$ is an affine transformation represented by 
the 
matrix $\Theta^S$, which maps statement $S$'s domain to its 
scheduling 
timestamp.  
For a statement $S$ with domain in space 
$\polyhedral[\vec{p}]{[\vec{x}^S]}$,
the schedule function is an $m \times (|\vec{x}^S|+ |\vec{p}| + 1)$ matrix. 
The statement's $m$ dimensional timestamp $\tau^S$ is given by:
\begin{equation}
\tau^S 
= \Theta^S \cdot \begin{bmatrix} \vec{x}^S \\ \vec{p} \\ 1 
\end{bmatrix} 
= 
\begin{bmatrix}
\Theta_{1, 1} & ... & \Theta_{1, |\vec{x}^S| + |\vec{p}| + 1 } \\
\vdots & ... & \vdots \\
\Theta_{m, 1} & ... & \Theta_{m, |\vec{x}^S| + |\vec{p}| + 1 } \\
\end{bmatrix}\cdot
\begin{bmatrix} \vec{x}^S \\ \vec{p} \\ 1 \end{bmatrix}
\end{equation}

\subsubsection{ILP Formulation Of Scheduling}
\label{sec:ilp_scheduling}
Early works 
\citep{someeffsol2affschedpt1,someeffsol2affschedpt2} 
gave a greedy algorithm to the scheduling problem and provided the 
foundation of the scheduling problem.
\citet{ilpschedulemultidim} formalized the scheduling problem for 
obtaining an $m$-dimensional schedule as the following single convex 
problem:
\begin{subequations}\label{eq:ilpscheduleconvex}
\begin{align}
\forall \mathcal{D}_{S_1, S_2}, & \forall k \in \{1...m\},  
\delta_k^{\mathcal{D}_{S_1, 
S_2}} \in \{0, 1\} \label{eq:ilpscheduleconvexa} \\
\forall \mathcal{D}_{S_1, S_2}, & \sum_{k = 1}^m  
\delta_k^{\mathcal{D}_{S_1, S_2}} = 1 \label{eq:ilpscheduleconvexb} \\
\begin{split}
\label{eq:ilp_scheduling_before_farkas}
\forall \mathcal{D}_{S_1, S_2}, & \forall k \in \{1 ... m\}, \forall 
[\vec{x}^{S_1}, \vec{x}^{S_2}, \vec{p}] 
\in \mathcal{D}_{S_1, S_2} \\
& \Theta^{S_2}_k \cdot \begin{bmatrix} \vec{x}^{S_2} \\ \vec{p} \\ 1 
\end{bmatrix}
- \Theta^{S_1}_k \cdot \begin{bmatrix} \vec{x}^{S_1} \\ \vec{p} \\ 1 
\end{bmatrix} 
\ge \delta_k^{\mathcal{D}_{S_1, S_2}} - \sum_{i=1}^{k-1} 
\delta_i^{\mathcal{D}_{S_1, S_2}}  (K \vec{p} 
+ K)
\end{split}
\end{align}
\end{subequations}
where $\Theta^{S}_k$ denotes the $k$--th row of the matrix $\Theta^{S}$.

In words, \Cref{eq:ilpscheduleconvexa} creates a binary variable 
$\delta_k^{\mathcal{D}_{S_1, S_2}}$ for each dependence relation in the program
and for each of the $k \in \{1 ... m\}$ dimensions of the schedule.\footnote{If 
the dependence\rline{unionofdeprelation} relation between $S_1, S_2$ is 
a union of polyhedral 
relations, then we consider each piece of the union as distinct and set up the 
constraints in \Cref{eq:ilpscheduleconvex} for each piece of the union.}
Each variable models the comparison between 
the two $m$-dimensional timestamps of statements $S_1$ and $S_2$ at the $k$-th dimension.
\Cref{eq:ilpscheduleconvexb} specifies that $\delta^{\mathcal{D}_{S_1, 
S_2}}_k$ is only strongly satisfied once for all $k$.
\Cref{eq:ilp_scheduling_before_farkas} encodes the constraint that the 
schedule functions $\Theta^{S_1}$ and $\Theta^{S_2}$ must satisfy that $\vec{x}^{S_1}$ is scheduled before $\vec{x}^{S_2}$, if the dependence $\vec{x}^{S_1} \rightarrow \vec{x}^{S_2}$ exists. 
Specifically, the dependence $\vec{x}^{S_1} \rightarrow \vec{x}^{S_2}$ exists if $[\vec{x}^{S_1}, \vec{x}^{S_2}, \vec{p}]^\transpose \in \mathcal{D}_{S_1, S_2}$.  
The variable $K$ is a known constant obtainable from the original program, and is an 
upper bound modeling technique to make the problem convex.

\citet{ilpschedulemultidim} show that this problem is equivalent to an ILP thanks to Farkas' Lemma \citep{theoryofilp}.
Specifically, let $M^{\mathcal{D}_{S_1, S_2}}$ denote the constraint 
matrix of 
$\mathcal{D}_{S_1, S_2}$ given in \Cref{def:polyhedralrelation} (i.e. 
$\mathcal{D}_{S_1, S_2}$ is 
$\polyhedral[\vec{p}]{ 
[\vec{x}^{S_1}] \rightarrow [\vec{x}^{S_2}] : M^{\mathcal{D}_{S_1, S_2}} 
\cdot 
[\vec{x}^{S_1}, 
\vec{x}^{S_2}, \vec{p}, 1]^\transpose 
\ge \vec{0}}$).
Farkas' Lemma makes it possible to introduce a vector of integer variables, 
$\Lambda^{\mathcal{D}_{S_1, S_2}}_k$, of length $n$, where $n$ is one plus the 
number of rows of $M^{\mathcal{D}_{S_1, S_2}}$,  for all $\mathcal{D}_{S_1, S_2}$ and 
$k \in \{1...m\}$.
It is then possible to expand \Cref{eq:ilp_scheduling_before_farkas} to the following 
constraints. %
\begin{subequations}
\label{eq:ilp_scheduling_after_farkas}
\medmuskip=2mu
\begin{align}
\nonumber & \forall \mathcal{D}_{S_1, S_2}, \forall k \in \{1 ... m\}, \\
& \Lambda^{\mathcal{D}_{S_1, S_2}}_k \ge 0 \\
& \Theta^{S_2}_k \cdot \begin{bmatrix} \vec{x}^{S_2} \\ \vec{p} \\ 1 
\end{bmatrix}
- \Theta^{S_1}_k \cdot \begin{bmatrix} \vec{x}^{S_1} \\ \vec{p} \\ 1 
\end{bmatrix} 
- \delta_k^{\mathcal{D}_{S_1, S_2}} 
+ \sum_{i=1}^{k-1} \delta_i^{\mathcal{D}_{S_1, S_2}}  (K \vec{p} + K) 
= \left(\Lambda^{\mathcal{D}_{S_1, S_2}}_k\right)^\transpose \cdot 
\begin{bmatrix}
M^{\mathcal{D}_{S_1, S_2}}  \\ \vec{0} \quad 1
\end{bmatrix}\cdot 
\begin{bmatrix} \vec{x}^{S_1} \\ \vec{x}^{S_2} \\ \vec{p} \\ 1 
\end{bmatrix}
\label{eq:ilp_scheduling_after_farkas_before_equate}
\end{align} 
\end{subequations}
By equating the left and right side of
\Cref{eq:ilp_scheduling_after_farkas_before_equate} for all coefficients 
of 
$\vec{x}^{S_1}, \vec{x}^{S_2}, \vec{p}$ and the constant terms produces 
the 
desired affine 
constraints in the ILP formulation of scheduling.
Solving the above formulation produces the desired schedule 
coefficients $\Theta$ in 
\Cref{sec:schedulefn}. \todo{This needs attention.}

\subsubsection{Example}
\label{sec:ilp_schedule_example}
\vspace*{-4pt}
\begin{changed}{mblue}
We\rline{ilpschedulexample} give an example of the above ILP 
formulation of scheduling, using 
the example from \Cref{sec:introduction} (\Cref{eq:ms-scan} and its IR 
form in \Cref{lst:ms-scan-ir}). 
Let $S_1$ refer to statement \texttt{S1} and 
$S_2$ refer to statement 
\texttt{S2} in \Cref{lst:ms-scan-ir}.
The dependence relations between $S_1$ and $S_2$ are:
\begin{subequations}
\begin{align}
\mathcal{D}_{S_1, S_2} &= \polyhedral[N]{ 
[i^{S_1}, j^{S_1}] \rightarrow [i^{S_2}] : i^{S_1} = i^{S_2} \land 0 \le i^{S_1} 
< N 
- 1 \land 0 \le j^{S_1} \le i^{S_1} } \\
\mathcal{D}_{S_2, S_1} &= \polyhedral[N]{ 
[i^{S_2}] \rightarrow [i^{S_1}, j^{S_1}] : i^{S_1} = i^{S_2}+1 \land 0 \le 
i^{S_2} < N - 1 \land 0 \le j^{S_1} \le i^{S_2}  }.
\end{align}
\end{subequations}
The corresponding constraint matrices for $\mathcal{D}_{S_1, S_2} $ and 
$\mathcal{D}_{S_2,S_1} $ are given in \Cref{eq:ilp_example_adst}.

\begin{subequations}
\label{eq:ilp_example_adst}
\small
\begin{minipage}{0.45\linewidth}
\begin{equation}
M^{\mathcal{D}_{S_1, S_2}} = 
\begin{bmatrix}
  -1 & 0 & 0 & 1 & -2\\
  0 & 1 & 0 & 0 & 0\\
  1 & -1 & 0 & 0 & 0\\
  -1 & 0 & 1 & 0 & 0\\
  1 & 0 & -1 & 0 & 0\\
\end{bmatrix}
\end{equation}
\end{minipage}\hfil
\begin{minipage}{0.5\linewidth}
\begin{equation}
M^{\mathcal{D}_{S_2,S_1}} = 
\begin{bmatrix}
  -1 & 0 & 0 & 1 & -2\\
  0 & 0 & 1 & 0 & 0\\
  1 & 0 & -1 & 0 & 0\\
  -1 & 1 & 0 & 0 & -1\\
  1 & -1 & 0 & 0 & 1\\
\end{bmatrix}
\end{equation}
\end{minipage}
\end{subequations}
\vspace{.25em}

Assuming the schedule timestamp's dimension is $m = 3$ and $K$ is a 
large enough constant, we formulate the dependence constraints as 
in \Cref{eq:ilp_scheduling_before_farkas,eq:ilp_scheduling_after_farkas} by \Cref{eq:ilp_scheduling_example}.
\begin{subequations}
\small
\allowbreak
\newcommand{\DST}{\mathcal{D}_{S_1,S_2}}
\newcommand{\DTS}{\mathcal{D}_{S_2,S_1}}
\label{eq:ilp_scheduling_example}
\medmuskip=2mu
\begin{align}
0 \le \delta^{\DST}_k \le 1 \quad\quad 
0 \le \delta^{\DTS}_k \le 1\hphantom{xxxxx} 
& \forall k \in  \{1...3\} \label{eq:ilp_scheduling_example_delta_setup1}\\
\sum_{k=1}^3 \delta^{\DST}_k = 1 \quad\quad
\sum_{k=1}^3 \delta^{\DTS}_k = 1 
\hphantom{xxxxx} 
& \label{eq:ilp_scheduling_example_delta_setup2}\\
\Lambda^{\DST}_k \ge 0 \quad\quad
\Lambda^{\DTS}_k \ge 0 
\hphantom{xxxxxxxx} & \forall k 
\in  \{1...3\} \label{eq:ilp_scheduling_example_lambda_setup}\\
\left\lbrace
\begin{array}{@{}l@{}r@{}l}%
i^{S_1}: & - \Theta_{k,1}^{S_1} &= \Lambda^{\DST}_{k} \cdot 
[-1, 0, 1, -1, 1, 0] \\
j^{S_1}: & - \Theta_{k,2}^{S_1} &= \Lambda^{\DST}_{k} \cdot 
[0, 1, -1, 0, 0, 0] \\
i^{S_2}: & \Theta_{k,1}^{S_2} &= \Lambda^{\DST}_{k} \cdot 
[0, 0, 0, 1, -1, 0]\\
N : & \Theta^{S_2}_{k,2} - \Theta^{S_1}_{k, 3} + K \cdot 
\sum_{i=1}^{k-1} 
\delta_i^{\DST} &= 
\Lambda^{\DST}_{k} \cdot [1, 0, 0, 0, 0, 0] \\
1: & \Theta^{S_2}_{k, 3}  - \Theta^{S_1}_{k, 4} 
- \delta_k^{\DST} 
+ K \cdot \sum_{i=1}^{k-1} \delta_i^{\DST} &= 
\Lambda^{\DST}_{k} \cdot [-2, 0, 0, 0, 0, 1]
\end{array}
\right. & \forall k \in \{1...3\} \label{eq:ilp_scheduling_example_dst}\\
\left\lbrace
\begin{array}{@{}l@{}r@{}l}%
i^{S_2}: & -\Theta_{k,1}^{S_2} &= \Lambda^{\DTS}_{k} \cdot 
[-1, 0, 1, -1, 1 ,0]\\
i^{S_1}: & \Theta_{k,1}^{S_1} &= \Lambda^{\DTS}_{k} \cdot 
[0, 0, 0, 1, -1, 0]\\
j^{S_1}: & \Theta_{k,2}^{S_1} &= \Lambda^{\DTS}_{k} \cdot 
[0, 1, -1, 0, 0, 0]\\
N : & \Theta^{S_1}_{k,3} - \Theta^{S_2}_{k, 2} + K \cdot 
\sum_{i=1}^{k-1} 
\delta_i^{\DTS} &= 
\Lambda^{\DTS}_{k} \cdot [1, 0, 0, 0, 0, 0]\\
1: & \Theta^{S_1}_{k, 3}  - \Theta^{S_2}_{k, 4}
- \delta_k^{\DTS}
+ K \cdot \sum_{i=1}^{k-1} \delta_i^{\DTS}  &= 
\Lambda^{\DTS}_{k} \cdot [-2, 0, 0, -1, 1, 1]
\end{array}
\right. & \forall k \in \{1...3\} \label{eq:ilp_scheduling_example_dts}
\end{align}
\end{subequations}

To briefly summarize, \Cref{eq:ilp_scheduling_example_delta_setup1,%
eq:ilp_scheduling_example_delta_setup2,%
eq:ilp_scheduling_example_lambda_setup} set up the variables related to 
the dependence constraints.
\Cref{eq:ilp_scheduling_example_dst,eq:ilp_scheduling_example_dts} 
handle the case of $\mathcal{D}_{S_1,S_2}$ and 
$\mathcal{D}_{S_2,S_1}$ respectively for 
\Cref{eq:ilp_scheduling_after_farkas}.
Specifically, 
\Cref{eq:ilp_scheduling_example_dst,eq:ilp_scheduling_example_dts} 
apply Farkas' Lemma for the coefficients of $i^{S_1}, j^{S_1}, i^{S_2}, 
N$ and for the constant terms.
We have labeled each equational constraint in 
\Cref{eq:ilp_scheduling_example_dst,eq:ilp_scheduling_example_dts} 
with 
the corresponding symbolic term (i.e. one of $i^{S_1}, j^{S_1}, i^{S_2}, 
N$) or 
constant term (i.e. denoted by $1$) on the left side of the equation.
\end{changed}

\subsubsection{Non-Sequential Scheduling}
We first define a {\em sequential schedule}.
\begin{definition}[Sequential Schedule]
A schedule $\Theta$ is sequential if for each statement $S$ and its scheduling 
function $\Theta^S$ we have:
\begin{equation}\label{eq:sequential_schedule_requirement}
\Theta^S \cdot 
[\vec{x}^S_1, \vec{p}, 1]^\transpose \neq 
\Theta^S \cdot 
[\vec{x}^S_2, \vec{p}, 1]^\transpose
\quad \forall \vec{x}^S_1, \vec{x}^S_2. \; \vec{x}^S_1 \neq \vec{x}^S_2.
\end{equation}
\end{definition}
Intuitively, a schedule is sequential\rline{finalseqdef} if it maps distinct 
iteration vectors to distinct timestamps.

A valid,\rline{nonsequentialscheduleexplain} dependence-satisfying 
schedule that satisfies 
\Cref{eq:ilpscheduleconvex} may be {\em non-sequential} in that it 
permits multiple statement instances to execute in the same timestep.
Therefore, it is possible for a schedule to demand an unbounded 
number of statement instances to execute at the same timestamp, which does not directly map to physical 
machines with finite resources 
\citep{schedulereductions94,schedulereductionsrealistic}.

In the context of reductions, a reduction may be scheduled to accumulate an unbounded number of values at the same timestamp~\citep{schedulereductions94}.
For example, the (one-dimensional) schedule that assigns $\texttt{S1}[i, 
j]$ to $[i]$ is not realistic, since it assigns the same timestamp to 
$\texttt{S1}[i, j]$ for all $j$, and therefore requires accumulation of a 
potentially unbounded number of values.
\citet{schedulereductionsrealistic} demonstrate a scheduling approach that bounds 
the total accumulations per timestep to target physical machines.

In this paper, we do not extend \citet{ilpschedulemultidim}'s formalization with physical scheduling constraints. 
While this consideration is important for practical scheduling of 
reductions, we only use their scheduling formalization to support our 
formalization of the dependent reduction scheduling problem in 
\Cref{sec:mssrps}.
Our heuristic algorithm in \Cref{sec:heuristicalgo} does not require a 
schedule to have been computed using their scheduling formalization.
Our algorithm instead relies on a sequential schedule that can be computed via any means, including a scheduling algorithm that adopts realistic scheduling constraints.
In \Cref{sec:heuristicdirection} we discuss sequential scheduling in more 
detail.

\section{Background: The Simplifying Reductions Framework}
\label{sec:background-sr}

\todo{explain when SR works for single statement (answer: linear dependence)}
Previous work introduced a core transformation called the {\em simplification 
	transformation} (ST) that can transform a single reduction specified in 
in the polyhedral representation (\Cref{lst:irstmt}) to lower its complexity \cite{sr}. 
The work also introduced a set of enabling transformations that make available 
opportunities to simplify reductions.

For the core transformation, we use an example from 
\Cref{sec:introduction} to 
illustrate 
the transformation. 
For the enabling transformations, we include a brief description of each 
transformation in \Cref{sec:sr_enabling_transformations}.
Finally, \citet{sr} combine all the transformations to provide a dynamic 
programming algorithm to efficiently choose, from an infinite set of configurations and orders for 
the transformations, a sequence of transformations that lead to optimal complexity 
reduction. 

\subsection{Simplification Transformation}
\label{sec:st}

\newcommand{\btmp}{BTmp}
\newcommand{\btmpadd}{{\btmp}Add}

Here we use the example from \Cref{sec:mssr_example} to illustrate how the simplification transformation reduces the complexity of a reduction.  We provide a more complete specification of ST in~\Cref{sec:simplifying-reduction}.

{
\newcommand{\reusevec}{$[1,0]^\transpose$}
\newcommand{\SOneFinalizeAndSTwo}{
	\texttt{S1Fin}: \texttt{B}[i] = \texttt{\btmp}[i] : 
	\polyhedral{[i] : 0  \le i < N } \\
	\texttt{S2}: \texttt{A}[i+1] = f(\texttt{B}[i])  : 
	\polyhedral{[i] : 0  \le i < N - 1 }}
\newcommand{\lstleft}{
	\arraycolsep=0.3pt
	\begin{array}{l}
	\texttt{S1}: \texttt{\btmp}[i] \mathrel{+}=\texttt{A}[j]  : 
	\polyhedral{[i, j] : 0 \le i < N \wedge 0  \le j \le i } \\
	\SOneFinalizeAndSTwo
	\end{array}
}
\newcommand{\lstright}{\hspace{0cm}
\arraycolsep=0.3pt\begin{array}{l}
\texttt{S1Add}: \texttt{\btmpadd}[i] \mathrel{+}= \texttt{A}[j] 
: 
\polyhedral{[i, j] : 0 \le  i < N 
	\wedge i = j} \\
\texttt{S1AddOnly}: \texttt{\btmp}[i] = \texttt{\btmpadd}[i] : \polyhedral{[i] : i  = 
0 } \\
\texttt{S1AddReuse}: \texttt{\btmp}[i] = \texttt{\btmp}[i-1] + 
\texttt{\btmpadd}[i]:  
\polyhedral{[i] : 1 \le i < N } \\
\SOneFinalizeAndSTwo
\end{array}
}
\begin{clisting}[caption={ST in the polyhedral IR for the example in 
\Cref{sec:introduction} (\Cref{eq:ms-scan}), given the reuse vector \reusevec}, 
numbers=none,xleftmargin=0pt,label={lst:srprefixsumir},float={t},aboveskip=1.5ex]
`
\ifthenelse{\isundefined{\isthesis} \AND \isundefined{\ispopl}}{
${\lstleft}\hspace{-0cm}\Rightarrow{\lstright}$
}%
{
$\lstleft$ \\
$\quad\quad\Downarrow$ \\
$\lstright$
}
`
\end{clisting}
}

\Cref{lst:srprefixsumir} illustrates the example of applying ST to \Cref{lst:ms-scan-ir} (i.e., the IR form of \Cref{eq:ms-scan}) 
to produce the optimized version in \refeqscanopt. 
As we mentioned before, core ST operates on a single statement only and only produces a correct result for a dependent reduction if provided with a correct reuse vector.

\paragraph{Original Reduction}
Above the arrow, \Cref{lst:srprefixsumir} presents the reduction in 
\Cref{eq:ms-scan-reduction}
in the polyhedral IR as the statement \texttt{S1}
with domain $\mathcal{P} = \polyhedral[N]{[i, j] : 0 \le i < N \wedge 0 \le j
	\le i}$.  The right hand
side expression is $\texttt{A}[j]$, and $i$ is not a bound variable -- this
means given a fixed $j$, the right hand side's values are the same for
different values of $i$. 

\paragraph{Optimized Reduction} Below the arrow, 
the optimized prefix sum consists of three distinct computations: seeding the initial values of the reduction, computing the root of the reduction, 
and computing the core of the reduction via reuse.

The statement \rline{taddshouldbes1add}$\texttt{S1Add}$ computes $\texttt{\btmpadd}$, which serves to hold the unique seeds of the reduction: the uniquely computed value introduced into the reduction at each iteration $i$. The statement computes over the full space of $i$ and sets $\texttt{\btmpadd[i]}$ to equal $\texttt{A[i]}$.
The statement $\texttt{S1AddOnly}$ computes the root of the reduction, the first (by dependence order) value computed into $\texttt{\btmp[0]}$
and for which the reduction use not reuse any previous computation. Hence,  $\texttt{S1AddOnly}$ sets $\texttt{\btmp[0]}$ to equal $\texttt{\btmpadd[i]}$, the unique element for this iteration.
The statement $\texttt{S1AddReuse}$ computes the core of the reduction incrementally, summing the previous computation $\texttt{\btmp[i - 1]}$ to the unique element for this iteration.

To identify this optimization opportunity and generate the optimized code, the Simplification Transformation identifies a {\em reuse vector} by which shifting the original, unoptimized polyhedron ($\mathcal{P}$) makes plain that consecutive iterations of the polyhedron share the same computation and can be incrementalized.
The reuse vector then drives the transformation of the code.
To elaborate ST's mechanics, we first present how ST leverages a reuse vector to transform the program before then explaining how ST identifies appropriate reuse vectors.

\subsection{From Reuse Vector to Transformed Program}
\label{sec:st-reuse}
\ifthenelse{\isundefined{\ispopl}}%
{
\begin{figure}[ht]
	\centering
	{\includegraphics[width=0.8\linewidth]{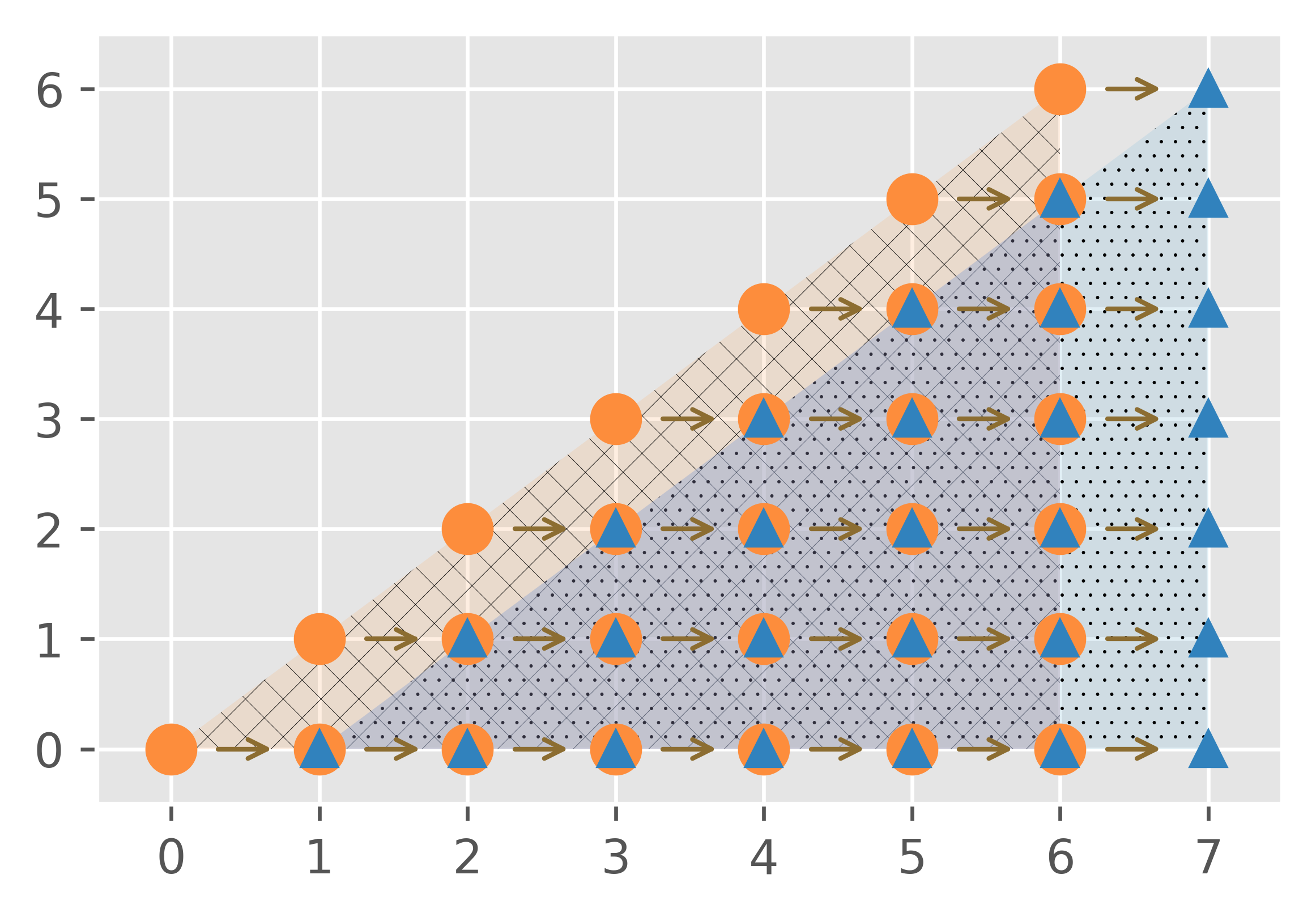}}%
	\caption{Visualization of algorithm on prefix sum example}
	\label{fig:prefixsum}
\end{figure}
}%
{
\begin{wrapfigure}[7]{r}{0.33\linewidth}
\vspace{-8ex}
	\centering
	\raisebox{0pt}[\dimexpr\height-1.2\baselineskip\relax]%
	{\includegraphics[width=\linewidth]{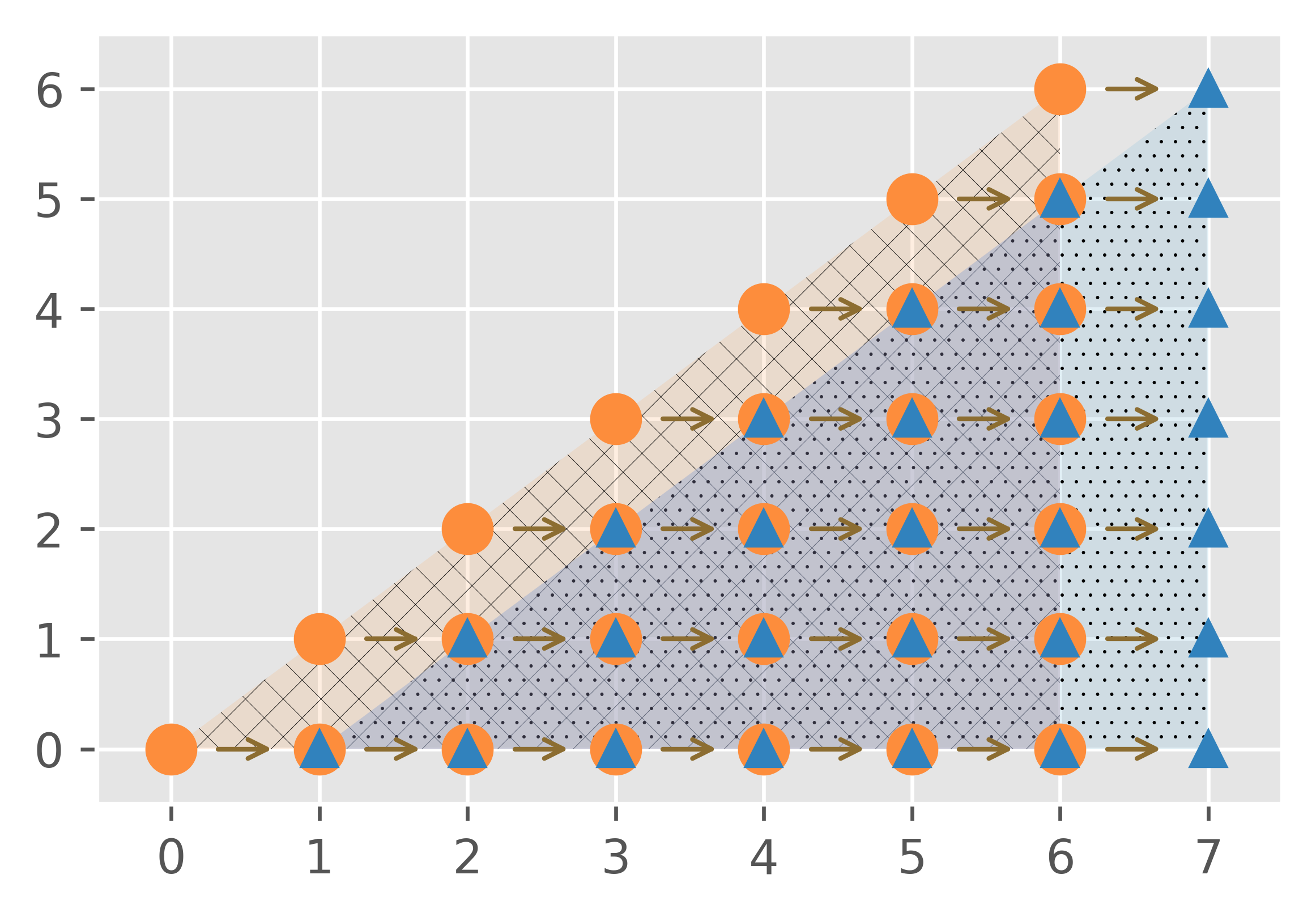}}
	\caption{Visualization of the algorithm on prefix sum example.}
	\label{fig:prefixsum}
\end{wrapfigure}
}

ST takes a reuse vector and manipulates the domain of the reduction to compute the subset of iteration
instances that can be transformed to incrementally compute their results using previously computed results.

Consider the {\em reuse vector} $\vec{r} =[1, 0]^\transpose$,
which can also be represented by the polyhedral relation $\polyhedral{[i, 
j] \rightarrow [i+1, j] : \forall i, j}$.
Conceptually, the reuse vector therefore denotes a shift of all points $[i, j]$ to $[i+1, j]$.
Given the reuse vector $\vec{r} =[1, 0]^\transpose$, ST performs the following steps:

\begin{itemize}[leftmargin=*]
	
	\item {\bf Shift.} The transformation first shifts \texttt{S1}'s polyhedron along the direction of the reuse vector, transforming \polyhedral{[i,j] : 0 \le i < N \wedge 0  \le j \le i } into $\polyhedral{[i, j] : 1 \le i < N + 1 \wedge 0 \le j \le i - 1}$.
\Cref{fig:prefixsum} illustrates the two polyhedrons, with the pre-transformation polyhedron presented  as the crosshatched polyhedron
	with round dots	and the post-shift polyhedron presented as the dot-shaded polyhedron with triangles.

	\item {\bf Intersect.} The transformation next computes the intersection of the
	shifted polyhedron\rline{polyhedraltopolyhedron} with its original 
	polyhedron, yielding $\polyhedral{[i, j] : 1 \le i < N \wedge 0 \le j \le i - 1 }$ (triangles 
	outlined in round dots in \Cref{fig:prefixsum}).
	This polyhedron denotes the subset of points of the original domain $\mathcal{P}$, whose value can be reused from the predecessor points as indicated by the reuse vector. We note that the intersection computation itself does not detect the equivalence of expressions across indices. 
	Instead, ST's reuse vector selection process (\Cref{sec:srconfig}) ensures that instances within the intersection denote equivalent expressions.
i	
	\item {\bf Project.} Finally, the transformation projects the result onto the space of the polyhedron that represents the indices of the left hand side of the array \texttt{BTmp}. Concretely, the transformation applies the projection represented by 
	the polyhedral relation 
	$\polyhedral{[i, j] \rightarrow [i] : \forall i, j}$), yielding the polyhedron 
	$\polyhedral{[i] : 1 \le i < N}$.
{\noindent}This final polyhedron is exactly the domain of elements of \texttt{BTmp} that exhibits reuse along the reuse vector $\vec{r}$. 
This polyhedron corresponds to the domain of the statement $\texttt{S1AddReuse}$, computes $\texttt{BTmp}[i]$ with $\texttt{BTmp}[i] =\texttt{BTmp}[i - 1] + \texttt{BTempAdd}[i]$.

\end{itemize}

The polyhedron $\polyhedral{[i] : 1 \le i < N}$ does not cover the full 
domain of the original reduction. Specifically, it is missing $\texttt{BTmp}[i]$ on the
domain $\polyhedral{[i] : i = 0}$ -- that is, exactly when $i = 0$. The value
of $\texttt{BTmp}[0]$ should be equal to $\texttt{A}[0]$ and the transformation
generates the statement \texttt{S1AddOnly} to perform this computation.

\subsection{Selecting a Reuse Vector}
\label{sec:srconfig}

The simplifying transformation work offers a fully automated technique to identify a reuse vector for a reduction. 
The techniques must specifically (1) identify if computation is {\em shared} among iterations of the reduction, (2) identify if the shared computation is {\em reusable}: if it is possible to transform the program reuse the shared computation), and 3) identify if reusing shared computation is {\em profitable}: the transformation reduces the complexity of a program.

\paragraph{Shared}
ST includes an algorithm to determine the {\em share space}: the set of all vectors such that, for each vector, the evaluations of right hand side expression
of the reduction are the same if shifted along the reuse vector. 

For the prefix sum example in \Cref{sec:st-reuse}, the right hand side expression ($A[j])$) is the same along the direction ~$[1, 0]^\transpose$. 
However, the right hand side expression is also the same along the direction $[-1, 0]^\transpose$ as well as $[2, 0]^\transpose$. 
In fact, any reuse vector $[n, 0]$ for a given integer $n$ is in the share space of this reduction. 
Therefore, generally, the set of reuse vectors is potentially infinite in size.

For a statement $S$, let $\mathcal{S}(S)$ denote the statement's share space. \citet[Section~5.2]{sr} demonstrated how to compute $\mathcal{S}(S)$ 
from $S$.

\paragraph{Inverse}
\begin{wrapfigure}[9]{r}{0.25\linewidth}
	\centering
	\vspace{-1ex}
	\raisebox{0pt}[\dimexpr\height-1.2\baselineskip\relax]%
	{\includegraphics[width=\linewidth]{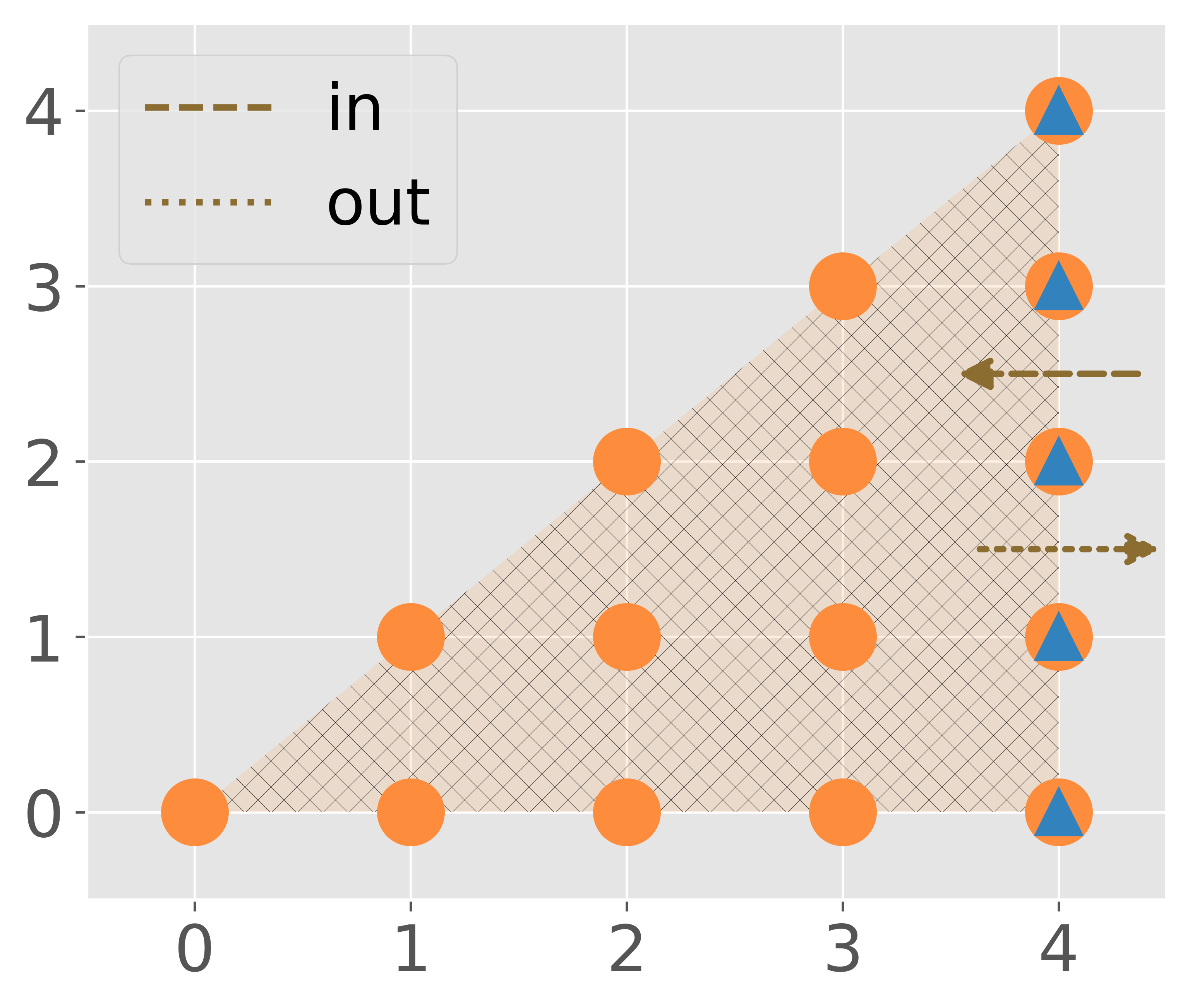}}
	\caption{Visualization of the Inverse constraint}
	\label{fig:in_out}
\end{wrapfigure}
If $\oplus$ does not have 
an inverse, we require that applying ST along a vector $\vec{r}$ does not 
introduce 
statements that require the inverse operator of $\oplus$. 
For example, if $\oplus$ is 
$\min()$ or $\max()$, it does not have an inverse; 
in such cases, \citet{sr} introduce the concept of {\em Boundary Constraints} -- which 
in short is the set of constraints of the domain $\mathcal{P}$ that are 
orthogonal to the projection 
$\textit{proj}$ -- and requires that $\vec{r}$ must be pointing out of 
(instead of pointing into) the boundaries of $\mathcal{P}$ corresponding 
to the Boundary Constraints. 

For example, consider a modification to the prefix sum example 
(\Cref{lst:naiveprefixsum}), where the summation is replaced with $\max$.
The ST application on this modified example along the reuse vector 
$[-1,0]^\transpose$ is invalid, since $\max$ does not have an inverse but the 
transformed program requires an inverse operation in order to compute $B[i]$ from 
$B[i+1]$ (i.e. \Cref{eq:scan-opt-rev-update}).
\Cref{fig:in_out}\rline{pointedoutfigure} illustrates this constraint: for the 
boundary of the 
domain orthogonal to the projection (the right side of the triangular 
domain), the inverse constraint requires the reuse vectors to point out 
of the domain (i.e. to the dotted arrow's direction).

For a statement $S$, let $\mathcal{I}(S)$ denote the set of vectors 
$\vec{r}$ that satisfy the inverse constraint.
\citet[Section~5.4]{sr} demonstrated how to compute $\mathcal{I}(S)$ for a statement $S$.

\paragraph{Profitable} 
Applying ST with a given reuse vector is profitable only if the transformed program 
has lower complexity than the original program.
Of note, the complexity of a program will not increase after applying ST for any $\vec{r}$. 
However, the complexity can stay the same if $\vec{r}$ is chosen along a direction 
where the original polyhedral domain $\mathcal{P}$ has constant thickness -- that is, 
the extent of $\mathcal{P}$ is bounded by some constant not parameterized by the 
input parameters of the program. 

For example, consider an extreme case of the prefix sum example 
(\Cref{lst:naiveprefixsum}, \Cref{lst:optprefixsum}) but with the input parameter $N$ 
fixed to 
some constant -- say $N = 4$. The complexities before and after ST will be the same 
-- 
$\BigO(1)$ -- since both programs will perform a fixed number of computations.

For a reduction statement $\text{S}$ with domain $\mathcal{P}$, let $\mathcal{L}(\mathcal{P})$ denote the set of reuse vectors that reduce the complexity of the reduction.
\citet[Section~4.2]{sr} demonstrated how to compute $\mathcal{L}(\mathcal{P})$ given a domain $\mathcal{P}$.

\paragraph{Valid} In general, the set of {\em valid} reuse vectors are those that satisfy the above properties (i.e., shared, reusable, and profitable).
Formally, if we let the notation $\text{S}.\textit{domain}$ denote the domain of a statement $\text{S}$,
then the set of valid reuse vectors is  $\mathcal{R}(S) = \mathcal{S}(S) \cap \mathcal{I}(S) \cap 
\mathcal{L}(S.\textit{domain})$. Therefore, any $\vec{r} \in \mathcal{R}(S)$ is a valid reuse vector for applying ST.

\subsection{Residual Reduction}
\label{sec:strecursive}

Notice in \Cref{lst:srprefixsumir} that the statement \textsf{S1Add} still contains a 
reduction. 
Although for this example \textsf{S1Add} does not have further ST opportunities, in 
general, the residual reduction might still have available ST opportunities so that ST 
can be applied recursively to all introduced reductions.\footnote{In general, ST can 
also introduce more than one reduction; we include a fuller description of ST in 
\Cref{sec:simplifying-reduction}.}

Even though\rline{final43} the space of valid reuse vectors 
(\mbox{\Cref{sec:srconfig}}) is 
potentially infinite and cannot be simply enumerated,
\mbox{\citet{sr}}  presented an dynamic programming algorithm that partitions this 
space into equivalence classes.
The algorithm then enumerates those equivalence classes to choose reuse vectors. 
The algorithm always terminates and produces a resulting program with optimal 
complexity (i.e. minimal polynomial degree).

\section{Simplifying Dependent Reductions Problem}
\label{sec:mssrps}

In this section, we state the Simplifying Dependent Reductions (SDR) problem. 
In particular, we focus on the core of the Simplifying Reductions approach -- the Simplification Transformation 
in \Cref{sec:st} --
and do not consider the Simplifying Reductions framework's additional {\em enabling transformations}.
These transformations increase available simplification opportunities; we briefly touch on enabling transformations in Section~\ref{sec:mssr:discuss}.

\subsection{Problem Statement}
\label{sec:mssrproblemstatement}
\newcommand{\prog}[1]{\ensuremath{\textit{prog}^{#1}}}

Ideally we would formulate the Simplifying Dependent Reduction (SDR) as \Cref{eq:mssrproblemstatementincorrect}.
\begingroup
\begin{subequations}\label{eq:mssrproblemstatementincorrect}
\setlength{\jot}{0em}
\begin{align}
& \text{\bf minimize} \;\; \text{complexity}(\textit{prog}') \\
& \text{\bf subject to} \;\; \notag \\
& \quad \prog{0} = \textit{prog}, \; \textit{prog}' = \prog{n} \\
& \quad \forall i \in \{1 ... n\} \colon \prog{i} = \text{ST}_{\text{S}_i, 
\vec{r}_i}(\prog{i-1}) 
\\
& \quad \quad \quad \vec{r}_i \in 
\mathcal{S}\left(\text{S}_i\right) \cap 
\mathcal{I}\left(\text{S}_i \right) \cap
\mathcal{L}\left(\text{S}_i.\textit{domain} \right)
\label{line:mssrproblemstatementincorrectrconstraints} \\
\begin{split}
& \quad \exists \text{ schedule } \Theta \text{ of } \prog{n},\;\; \\
& \quad\quad\quad \text{s.t.} \;\; \Theta \text{ satisfies } 
\text{dependence}(\prog{n})
\end{split} \\
& \text{\bf given} \;\; \textit{prog}, \; \text{dependence}(prog) \\
& \text{\bf variables} \;\;   \text{S}_1, \ldots , \text{S}_n, \; \vec{r_1},  \ldots, \vec{r_n}
\end{align}
\end{subequations}
\endgroup
\Cref{eq:mssrproblemstatementincorrect} states that given a program $prog$, 
and all pairwise 
dependences between those 
statements, $\text{dependence}\left(\textit{prog}\right)$, apply a sequence of $n$ ST 
transformations, 
$\text{ST}_{\text{S}_1, \vec{r}_1}, \ldots \allowbreak, \text{ST}_{\text{S}_n, 
\vec{r}_n}$ that minimize 
the complexity of the resulting program, $\textit{prog}'$. 
Here we use $\text{ST}_{\text{S}, \vec{r}}(\textit{prog})$ to denote an ST that is applied on a 
statement S in $\textit{prog}$ along the reuse vector $\vec{r}$.
Each $\text{S}_i$ is a variable that refers to a reduction in $\textit{prog}^{i-1}$.
Note that $\text{S}_i$ can either be a reduction that is in $\textit{prog}^0$, the initial 
input program, or it can be reduction that is introduced by any of the previous
$i-1$ ST applications (i.e. recursive ST in \Cref{sec:strecursive}).
Further, \Cref{line:mssrproblemstatementincorrectrconstraints} requires each $r_i$ to satisfy 
the constraints (i.e. complexity, inverse and sharing, denoted by $ \mathcal{S}(.), 
\mathcal{I}(.), \mathcal{L}(.)$ respectively) as stated 
in \Cref{sec:srconfig}. 

Unfortunately, there are two issues with \Cref{eq:mssrproblemstatementincorrect}:
1) it has infinite space for $\vec{r}_i$
2) it has impractically large space for $\text{S}_i$. 

Firstly, if we assume that $n$ is given and bounded,
\footnote{
The total number of ST applications $n$ must be a finite given the fact that the 
complexity of the input 
program is finite, and that each ST application reduces the complexity.
}
the formulation does not 
readily translate to an executable algorithm.
Specifically, enumeratively searching all possible $\vec{r}_i$ combinations is not 
feasible: each 
$\vec{r}_i$ alone is chosen from an infinite set of vectors, and the entire search 
space is also infinite; therefore the search space of reuse vectors cannot be simply enumerated 
(however it can be partitioned in to equivalence classes as proposed by \citet{sr}).

Secondly, also assuming $n$ is given and bounded, the program relies on a sequence of statements, ($\text{S}_1, \ldots,\text{S}_n$), to specify on which statement in $\prog{i}$ to perform ST. 
Although, unlike the case of $\vec{r}_i$, the number of choices for each $\text{S}_i$ is finitely 
bounded (i.e. by the number of ST-applicable reductions in the program), the combinations of 
all possible $(\text{S}_1, \dots, \text{S}_n)$ has at least $|\text{S}_1|!$ possibilities: assuming the 
best case scenario where 
each ST applications removes one reduction and introduces zero reductions that are potentially 
applicable for further ST applications, 
which imples the $i$-th ST application has $n-i+1$ remaining alternative choices 
of $\text{S}_i$ 
(i.e. $|\text{S}_i| = n - i + 1$). 
Therefore the search space of $\text{S}_i$ is also not practical to navigate with enumerative 
search.

\subsection{Per-face ST Application}
\label{sec:perfacest}

We\rline{rewritelastparagraph51} will resolve these issues with a correct 
formalization in the rest of \Cref{sec:mssrps}. 
Specifically, we show, for a program, a one-to-one correspondence 
between all its 
potential 
ST applications and all {\em faces} of its reductions' domains. 
This correspondance allows a construction of an Integer Bilinear Programming 
(IBP) formulation to SDR, which avoids the explicit enumerative searchs in the 
above issues of \Cref{eq:mssrproblemstatementincorrect}.
\begin{definition}[{\em Face of polyhedral set}]\label{def:polyface}
Let the polyhedral set $\mathcal{P} = \polyhedral[\vec{p}]{[\vec{x}] : M \cdot [\vec{x}, 
\vec{p}, 1] 
\ge 
\vec{0}}$. Let $M_i$ be the $i$-th row of matrix $M$. A face of $\mathcal{P}$ is 
defined as 
$\mathcal{F} = \mathcal{P} \cap \mathcal{B}$ where $\mathcal{B} = 
\polyhedral[\vec{p}]{[\vec{x}] : 
B \cdot [\vec{x}, \vec{p}, 1] = \vec{0}}$ and $\forall i \exists j, B_i = M_j$. 
\end{definition}
In words, a face of $\mathcal{P}$ is $\mathcal{P}$ with a subset of 
(potentially empty or all) inequality constraints of $\mathcal{P}$ changed to 
equality constraints. 

We first make the following observation of ST on a single statement $\text{S}$ with domain 
$\mathcal{P}$: 
if we apply ST on $\text{S}$, we can then recursively apply ST on the newly introduced 
reductions, as in 
\Cref{sec:strecursive}, and this is exactly the root problem of the incorrect formulation 
\Cref{eq:mssrproblemstatementincorrect}: this recursion appears non-terminating. 
We will solve this issue by stating and proving \Cref{thm:stfacecorrespondance} --- to this end, 
we 
first recall  \Cref{thm:subdomaindecompose} from \citet{sr} that we will use in our 
proof. We 
then 
state \Cref{thm:stfacecorrespondance} and give a proof. 

\begin{lemma}[Local Face Correspondance {\cite[Theorem 
3]{sr}}]\label{thm:subdomaindecompose}
	Let $\mathcal{P'}$ be the translation of an $n$-dimensional $\mathcal{P}$ along $\vec{r}$, 
	then 
    $\mathcal{P} - \mathcal{P'} = \uplus_i \mathcal{P}_i$. 
    That\rline{disjunctunionisfinite} is, $\mathcal{P} - \mathcal{P'} $ is a finite 
    union of $\mathcal{P}_i$s. 
    Further, there exists a one-to-one map from $i$ to faces of $\mathcal{P}$ 
    such that each 
	$\mathcal{P}_i$ corresponds uniquely to a $(n-1)$-dimensional face of $\mathcal{P}$. 
\end{lemma}

\begin{lemma}[Global Face Correspondance]\label{thm:stfacecorrespondance}
Each recursive application of ST  is on a subset (a polyhedral set) of $\mathcal{P}$, 
and all subsets correspond exactly one-to-one to all faces of $\mathcal{P}$. 
\end{lemma}

\begin{proof} 
Given a statement S with domain $\mathcal{P}$, ST  
performs a shift of $\mathcal{P}$ along a given reuse vector to $\mathcal{P'}$; new reduction 
statements are introduced over domains  $\mathcal{P} - \mathcal{P'}$ and $\mathcal{P'} - \mathcal{P}$.
Note that these two domains are non-convex half shells around the original domain $\mathcal{P}$, and 
together form a full shell around $\mathcal{P}$. The two shells are both non-convex, however by 
\Cref{thm:subdomaindecompose}, they decompose into convex polyhedral domains, each corresponding 
to a unique $(n-1)$-dimensional face of the $n$-dimensional 
polyhedron $\mathcal{P}$. 

ST is applied recursively on these decomposed $(n-1)$-dimensional faces
and then on the sequence of $(n-i)$-dimensional faces until the 
recursion hits the vertices of $\mathcal{P}$. Therefore, the entire recursion is 
a procedure that enumerates through all faces of a statement $\text{S}$'s full domain 
$\mathcal{P}$, and assigns a reuse vector to 
each face. 
\end{proof}

With\rline{rewritelastparagraph511} \Cref{thm:stfacecorrespondance}, the 
recursive ST application always 
terminates since the number of faces of $\mathcal{P}$ is finite. 
Further, this introduces a {\em per-face application view of ST}.
Specifically, if we think about the algorithm that recursively 
applies ST to an input program to get the final optimized program as is done in 
\citet{sr}, the recursion forms a computation tree, where each node of the tree 
corresponds to a choice of reuse vector of ST. 
\Cref{thm:stfacecorrespondance} says that 
1) this recursion is bounded, and 2) each node in this computation tree 
corresponds one-to-one to a face in the input program's statements' domains.
Under the per-face application view, an algorithm may first 
choose a reuse vector for each face of each statement in the program up-front 
(i.e. assigns the choice of each node in the computation tree up-front), instead of 
recursively requesting the reuse vectors.

\subsection{Integer Bilinear Program Formulation}
\label{sec:mssribpformulation}
With the per-face application view of ST in \Cref{sec:perfacest}, we are now ready 
to give the correct 
formulation of SDR.
The basic idea behind this formulation is to combine previous work on ST for a single 
statement \citep[Section~5]{sr} previous work on the integer linear program 
formulation of polyhedral 
model scheduling \cite{ilpschedulemultidim, pouchet.07.cgo, pouchet.08.pldi}, and 
the 
per-face application view of ST (\Cref{sec:perfacest}). 
We first revisit \Cref{eq:mssrproblemstatementincorrect} and give the correct high level formulation
as \Cref{eq:mssrproblemstatementcorrect}.

This optimization problem minimizes the complexity of $\textit{prog}'$
(\Cref{line:mssrproblemstatementcorrectobjective}), which is a version of \textit{prog} 
transformed by a 
composition of STs applied to each face 
(\Cref{line:mssrproblemstatementsequencest}).
The reuse vectors $\{r_1, \ldots, r_n\}$ that drive each ST must lie in the prescribed 
set that presents sharing ($\mathcal{S}(.)$), satisfies the inverse condition 
($\mathcal{I}(.)$) and reduces complexity ($\mathcal{L}(.)$), as in \Cref{sec:srconfig}
(\Cref{line:mssrproblemstatementcorrectrconstraint}). 
Lastly, there must exist a schedule $\Theta$ that satisfies the 
dependences in $\textit{prog}'$ 
(\Cref{line:mssrproblemstatementcorrecttexistheta}).
\begin{subequations}\label{eq:mssrproblemstatementcorrect}
\setlength{\jot}{0em}
\begin{align}
& \text{\bf minimize} \;\; \text{complexity}( \textit{prog}') 
\label{line:mssrproblemstatementcorrectobjective} \\
& \text{\bf subject to} \;\; \notag \\
& \quad \textit{prog}' = (\text{ST}_{f_1, \vec{r}_{1}}  \circ  \ldots \circ 
\text{ST}_{f_n, \vec{ 
		r}_{n}}) 
(\textit{prog}) \label{line:mssrproblemstatementsequencest}\\
& \quad \forall i \in \{1 ... n\}: \; \vec{r}_i \in 
\mathcal{S}\left(f_i.\textit{stmt} \right)  \cap
\mathcal{I}\left(f_i.\textit{stmt}\right)  \cap
\mathcal{L}\left(f_i\right)
\label{line:mssrproblemstatementcorrectrconstraint} \\
\begin{split}
& \quad \exists \text{ schedule } \Theta \text{ of } \textit{prog'},\\
& \quad\quad\quad \text{s.t.} \;\; \Theta \text{ satisfy } \text{dependence}(prog')
\end{split} \label{line:mssrproblemstatementcorrecttexistheta}  \\
& \text{\bf given} \;\; \textit{prog}, \; \text{dependence}(prog) 
\label{line:mssrproblemstatementcorrectrdepconstraint} \\
& \text{\bf variables} \;\;  \vec{r_1}, \ldots, \vec{r_n}
\end{align}
\end{subequations}

This high level formulation is similar to \Cref{eq:mssrproblemstatementincorrect}, 
except that 
now
1) each reuse vector $\vec{r}_i$ is in one-to-one correspondence with a face $f_i$ 
--- we 
thus have a bounded number of unknown variables for reuse vectors, and 
2) the variables $\text{S}_i$ are eliminated, as the new formulation uses the per-face 
ST view, 
instead of the recursive ST application view.   
Lastly, each reuse vector is still constrained to satisfy the validity constraints 
(i.e. \Cref{line:mssrproblemstatementcorrectrconstraint}).

\subsubsection{Transformed Program}

The program $\textit{prog}'$ is a version of \textit{prog} transformed by the 
composition of STs applied to each face 
(\Cref{line:mssrproblemstatementsequencest}) following the insight from 
\Cref{sec:perfacest}. \rline{testlabel}

{
\newcommand{\reusevec}{$\vec{r}_1 = [r_i, r_j	]^\transpose$}
\renewcommand{\btmp}{\texttt{BTmp}}
\renewcommand{\btmpadd}{\texttt{BTmpAdd}}
\newcommand{\btmpsub}{\texttt{BTmpSub}}
\newcommand{\SOneFinalizeAndSTwo}{
	\texttt{S1Fin}: \texttt{B}[i] = \texttt{\btmp}[i] : 
	\polyhedral{[i] : 0  \le i < N } \\
	\texttt{S2}: \texttt{A}[i+1] = f(\texttt{B}[i])  : 
	\polyhedral{[i] : 0  \le i < N - 1 }}
\newcommand{\lstleft}{
	\arraycolsep=0.3pt
	\begin{array}{l}
	\texttt{S1}: \texttt{\btmp}[i] \mathrel{+}=\texttt{A}[j]  : 
	\polyhedral{[i, j] : 0 \le i < N \wedge 0  \le j \le i } \\
	\SOneFinalizeAndSTwo
	\end{array}
}
\newcommand{\lstright}{\hspace{0cm}
\arraycolsep=0.3pt\begin{array}{l}
\texttt{S1}: \btmp[i] \mathrel{+}= A[j]	: \polyhedral{ [i, j] : r_i = 0 \land r_j = 0 
\land 
i < N \land 
0 \le j \le i } \\
\texttt{S1AddR1Pos}: \btmpadd[i] \mathrel{+}= A[j]	: \polyhedral{ [i, j] : r_i > 0 
\land 
r_j = 0 \land i < N \land i - r_i < j \land 0 \le j \le i } \\
\texttt{S1AddR1Neg}: \btmpadd[i] \mathrel{+}= A[j]	: \polyhedral{ [i, j] : r_i < 0 
\land 
r_j = 0 \land N + r_i \le i < N \land 0 \le j \le i } \\
\texttt{S1SubRNeg}: \btmpsub[i + r_i] \mathrel{+}= A[j]	: \polyhedral{ [i, j] : r_i < 
0 
\land r_j = 0 
\land - r_i \le i < N \land i + r_i < j \le i } \\
\texttt{S1AddOnlyR1Pos}: \btmp[i] = \btmpadd[i]	: \polyhedral{ [i] : r_j = 0 \land 
0 \le 
i < r_i } \\
\texttt{S1AddOnlyR1Neg}: \btmp[i] = \btmpadd[i]	: \polyhedral{ [i] : r_j = 0 \land 
N + 
r_i \le i < N } \\
\texttt{S1AddReuse}: \btmp[i] = \btmp[i - r_i] + \btmpadd[i]	: \polyhedral{ [i] : 
r_i 
> 0 \land r_j = 0 \land r_i \le i < N } \\
\texttt{S1SubRNegReuse}: \btmp[i] = \btmp[i - r_i] - \btmpsub[i]	: \polyhedral{ 
[i] : 
r_i < 
0 \land r_j = 0 \land 0 \le i < N + r_i } \\
\SOneFinalizeAndSTwo
\end{array}
}
\begin{clisting}[caption={Applying ST to all the faces of the polyhedral IR program of 
\Cref{eq:ms-scan}, given the reuse vector \reusevec consisting of integer 
variables $r_i$ and $r_j$. The reuse vectors $\{\vec{r}_{2}, \ldots, \vec{r}_8\}$ 
and 
statements with empty domains are omitted. }, 
numbers=none,xleftmargin=0pt,label={lst:bilpstexample},float={t},aboveskip=1.5ex]
`
\ifthenelse{\isundefined{\isthesis} \AND \isundefined{\ispopl}}{
${\lstleft}\hspace{-0cm}\Rightarrow{\lstright}$
}%
{
$\lstleft$ \\
$\quad\quad\Downarrow$ \\
$\lstright$
}
`
\end{clisting}
}
\begin{changed}{mgreen}
\paragraph{Example}
\Cref{lst:bilpstexample} illustrates\rline{bilptransformprogexample} applications of 
ST to all the faces of the 
program for \Cref{eq:ms-scan}.
Specifically, the top of \Cref{lst:bilpstexample} shows the program \textit{prog} in 
the polyhedral IR of \Cref{eq:ms-scan}, and the bottom of 
\Cref{lst:bilpstexample} shows the program $\textit{prog}'$ after the sequence of ST 
transformations.
Each introduced reduction, by the insight from \Cref{sec:perfacest}, corresponds 
to a face of the reduction \texttt{S1}'s domain.
The reduction \texttt{S1} has domain $\polyhedral{ [i, j] : r_i = 0 \land 
r_j = 0 \land i < N \land 0 \le j \le i } $.
\Cref{eq:ms-scan-faces} gives all the faces $f_i$ of this domain.

\vspace{.5em}
\begin{subequations}
\label{eq:ms-scan-faces}
\begin{minipage}{0.45\linewidth}
\noindent
\begin{align}
f_1 &= \polyhedral{[i, j] : 0 \le i < N \land 0  \le j \le i } \\
f_2 &= \polyhedral{[i, j] : 0 \le i < N \land j = 0 } \\
f_3 &= \polyhedral{[i, j] : 0 \le i < N \land j = i } \\
f_4 &= \polyhedral{[i, j] : i = N - 1 \land 0 \le j \le i }
\end{align}
\end{minipage}\hfil%
\begin{minipage}{0.45\linewidth}
\noindent
\begin{align}
f_5 &= \polyhedral{[i, j] : i = 0 \land j = 0} \\
f_6 &= \polyhedral{[i, j] : i = N - 1 \land j = 0} \\
f_7 &= \polyhedral{[i, j] : i = N - 1 \land j = N - 1} \\
f_8 &=  \polyhedral{[i, j] : 1 = 0}
\end{align}
\end{minipage}
\end{subequations}
\vspace{.5em}

We briefly explain the statements in $\textit{prog}'$ and associate the 
introduced reductions to the faces.
\begin{enumerate}[leftmargin=*]
\item \texttt{S1} in $\textit{prog}'$ is the same as \texttt{S1} in the 
original \textit{prog}, but with 
the extra constraints on its domain that $r_i = 0 \land r_j = 0$. 
This corresponds to the case that $\vec{r}_1 = [0, 0]^\transpose$, which means 
that we do not apply ST at all. This reduction corresponds to the face $f_1$.
\item \texttt{S1AddR1Pos}  corresponds to $f_3$ (the top polyhedron 
in \Cref{fig:s1dep}) for the 
case when $r_1 >0.$ 
\texttt{S1AddR1Neg}  corresponds to $f_3$ (the top polyhedron in 
\Cref{fig:sminus1dep}) for the 
case when $r_1 <0$. 
\item \texttt{S1SubRNeg} is a residual reduction that is required for 
incrementally computing \texttt{B} via subtraction. It is only non-empty for the 
case when $r_1 < 0$.
This reduction corresponds to the face $f_2$.
\item \texttt{S1AddOnlyR1Pos}  initializes \texttt{BTmp} from the result 
of \texttt{S1AddR1Pos} for the case when $r_1 >0$.
\texttt{S1AddOnlyR1Neg}  initializes \texttt{BTmp} from the result of 
\texttt{S1AddR1Neg} for the case when $r_1 <0$.
\item \texttt{S1AddReuse} and \texttt{S1SubRNegReuse} compute
$\texttt{B}$ incrementally by respectively adding or subtracting from the previous 
value of $\texttt{B}[i - r_i]$. 
\item \texttt{SFin} and \texttt{S2} remain the same as in the original 
\textit{prog}.
\end{enumerate}
\end{changed}

\subsubsection{Reuse Constraints}
\label{sec:reuse_constraint_bilp}
The reuse constraints enforce that each $\vec{r}_i$ is chosen from 
$
\mathcal{S}\left(f_i.\textit{stmt} \right) \cap
\mathcal{I}\left(f_i.\textit{stmt}\right) \cap
\mathcal{L}\left(f_i\right)
$.
\citet{sr} demonstrate how to compute each of these sets from $f_i$ 
(\Cref{sec:srconfig}). 
The intersection is a union of polyhedral sets.
We use disjunction to constrain $\vec{r}_i$ to 
belong to one of the polyhedral sets.  
For each polyhedral set, encoding that $\vec{r}_i$ belongs to the polyhedral set is 
then just a simple affine inequality constraint. 

\begin{changednospace}{mgreen}
\paragraph{Example}
We illustrate\rline{bilpreuseconstraintexample} reuse constraints for the example 
of \Cref{lst:bilpstexample}.
First, $\mathcal{S}(\texttt{S1}) = \polyhedral[N]{[i, j] : j = 0}$ since reduction's 
right hand side $A[j]$'s values are the same for different $i$ for any fixed $j$,  for 
$\mathcal{S}(\texttt{S1})$.
Second, $\mathcal{I}(\texttt{S1})$ is the universe set (i.e. does not impose 
any constraint), since the operator of \texttt{S1} is the 
addition operator (i.e. $+$) and it has a well defined inverse operator, subtraction 
(i.e. $-$).
Third, we give $\mathcal{L}(f_i)$ in \Cref{eq:ms-scan-faces-linealty} in 
\Cref{sec:extra_listings}.
Therefore, the intersection $\mathcal{S}(\texttt{S1}) \cap 
\mathcal{I}(\texttt{S1}) \cap \mathcal{L}(f_i) $ is equal to $\polyhedral[N]{[i, j] : j = 0} $ 
for $f_1$ 
 and $f_2$, and is equal to $\polyhedral[N]{[i, j] : i = 0 \land j = 0}$ for $f_{3...8}$.
In summary, we impose the constraint that $\vec{r}_1, \vec{r}_2 \in 
\polyhedral[N]{[r_i, r_j] : r_j = 0 }$, and $\vec{r}_i \in \polyhedral[N]{[r_i, r_j] : r_i 
= 0 \land r_j = 0 }$ for all $i \ge 3$.

\end{changednospace}

\subsubsection{Dependence Constraints}
The dependence constraints enforce that $\Theta$ satisfies the dependences of 
$\textit{prog}'$. Specifically, it requires that for each pair of statements $S_1$ and 
$S_2$ that potentially occur in $\textit{prog}'$, their scheduling functions 
$\Theta^{S_1}, 
\Theta^{S_2}$ satisfy the dependence relation 
$\mathcal{D}_{S_1, S_2}$. %
At a high level, we set up the dependence constraints just the same as in 
\Cref{eq:ilpscheduleconvex};
however, with \Cref{eq:mssrproblemstatementcorrect} we allow the dependence 
matrix $D_{S_1, S_2}$ to contain entries with 
(linear) terms with unknowns from $\vec{r}_1 ... \vec{r}_n$. 
An informative argument for why $D_{S_1, S_2}$ contains these unknown entries 
is:
if we look from the recursive ST view, each application of ST introduces a reuse 
vector variable  $\vec{r}_i$, and the algorithm recurses down to the 
residual reductions -- for the next recursive application, we can think of it as 
taking in a program with both the original program's parameters and the reuse 
vectors introduced by the previous ST application. The residual reductions' 
domains then have space extended by $\vec{r}_1 ... \vec{r}_n$. 

\begin{changednospace}{mgreen}
\paragraph{Example}
We\rline{bilpdependencyexample} formulate the dependence constraints 
(\Cref{line:mssrproblemstatementcorrecttexistheta}) for $\textit{prog}'$ in 
\Cref{lst:bilpstexample}.
This step is the same as in \Cref{sec:ilp_scheduling} 
except that the dependence 
 relations' constraint matrices, instead of being constant matrices, now contain 
 $r_1$ and $r_2$ as integer variables.
For brevity of presentation, we will give the 
example of formulating the 
dependence constraint (i.e. 
\Cref{eq:ilpscheduleconvex,eq:ilp_scheduling_after_farkas}) for statement 
\texttt{S1AddReuse} to itself. 
The rest of the dependence constraints between other pairs of statements can be 
formulated in the same way.

\begin{wrapfigure}[7]{r}{0.45\linewidth}
\centering
\begin{minipage}[t]{\linewidth}%
\medmuskip=1mu%
\small
\vspace{-10pt}
\noindent\begin{equation}
\label{eq:bilp_dependency_relation_constraint_matrix}
\hspace{-9pt}
M^{\mathcal{D}_{S,S}} \cdot
 \begin{bmatrix} i \\ i' \\ N \\ 1  \end{bmatrix} = 
\left[\begin{matrix}0 & 0 & 0 & r_{i} - 1\\1 & 0 & 0 & - r_{i}\\-1 & 0 & 1 & - r_{i} - 
1\\-1 & 1 & 0 & - r_{i}\\1 & -1 & 0 & r_{i}\\0 & 0 & 0 & r_{j}\\0 & 0 & 0 & - 
r_{j}\end{matrix}\right] 
\cdot \begin{bmatrix} i \\ i' \\ N \\ 1  \end{bmatrix}
\end{equation}
\end{minipage}
\end{wrapfigure}

For this example, we let statement $S$ be \texttt{S1AddReuse}.
For this statement, the dependence relation $\mathcal{D}_{S,S}$ between 
instances of  $S$, is 
$\polyhedral[N]{[i] \rightarrow [i'] :  i' = i + r_i \land r_i > 0 
\land r_j = 0 \land r_i 
<= i < N - r_i}$. 
\Cref{eq:bilp_dependency_relation_constraint_matrix} presents the constraint 
matrix $M^{\mathcal{D}_{S,S}}$ of this polyhedral relation.
Notice that $M^{\mathcal{D}_{S,S}}$ contains $r_i$ and $r_j$ as variables in its last 
column.
We can then formulate the dependence constraint as in 
\Cref{eq:ilpscheduleconvex,eq:ilp_scheduling_after_farkas} and get 
\Cref{eq:bilp_example_dss_after_farkas}.

\begin{subequations}
\label{eq:bilp_example_dss_after_farkas}
\allowdisplaybreaks
\small
\newcommand{\DST}{\mathcal{D}_{S,S}}
\medmuskip=2mu
\begin{adjustbox}{max width=\textwidth}
\noindent%
\begin{minipage}{0.4\linewidth}
\begin{equation}
0 \le \delta^{\DST}_k \le 1 \;\; \forall k \in  \{1...3\} 
\end{equation}
\end{minipage}\hfil%
\begin{minipage}{0.25\linewidth}
\begin{equation}
\sum_{k=1}^3 \delta^{\DST}_k = 1
\end{equation}
\end{minipage}\hfil%
\begin{minipage}{0.35\linewidth}
\begin{equation}
\Lambda^{\DST}_k \ge 0 \;\; \forall k 
\in  \{1...3\} 
\end{equation}
\end{minipage}
\end{adjustbox}
\begin{equation}
\label{eq:ilp_scheduling_after_farkas_bilinear}
\left\lbrace
\begin{array}{@{}l@{}r@{}l}%
i: & - \Theta_{k,1}^S &= \Lambda^{\DST}_{k} \cdot [0, 1, -1, -1, 1, 0, 0, 0] \\
i': & \Theta_{k,1}^S &= \Lambda^{\DST}_{k} \cdot 
[0, 0, 0, 1, -1, 0, 0, 0]\\
N : &  K \cdot \sum_{i=1}^{k-1} 
\delta_i^{\DST} &= 
\Lambda^{\DST}_{k} \cdot [0, 0, 1, 0, 0, 0, 0, 0] \\
1: & K \cdot \sum_{i=1}^{k-1} 
\delta_i^{\DST} - \delta_k^{\DST} &= 
\Lambda^{\DST}_{k} \cdot [r_i - 1, -r_i, -r_i-1, -r_i, r_i, r_j, -r_j, 1]
\end{array}
\right. \; \forall k \in \{1...3\}
\end{equation}
\end{subequations}
Notice that \Cref{eq:ilp_scheduling_after_farkas_bilinear} refers to the integer 
variables $r_i$ and $r_j$, multiplying them with 
$\Lambda^{\mathcal{D}_{S,S}}_k$ makes the problem bilinear constrained.
\end{changednospace}

\subsubsection{Objective: Complexity}
\label{sec:mssrcomplexity}
Since we would like to minimize the overall complexity, we need to express our integer bilinear 
program's objective as the complexity of the transformed program. We can compute 
complexity of 
each face by counting the cardinality of each face's domain \cite{barvinok}.  
The cardinality of a face is an Ehrhardt polynomial \cite{ehrhart} in terms of the program parameters. 

\paragraph{Encoding}
If the program only has one parameter, then the degree of the polynomial is a natural choice of a scalar that 
represents the complexity of the program. 

If the program has multiple parameters, then one needs to be careful about comparing complexities: it is 
necessary to be able to compare between $\BigO(M^2 N)$ and $\BigO(M N^2)$ in 
order to minimize 
complexity. 
To this end, we assume that a total ordering is given for all possible polynomial terms of global parameters 
as a sequence of increasing scalars. For example, with two global parameters $M, N$, and maximum 
possible complexity $\BigO(M^2 N^2)$, a total ordering such as $\BigO(1) < \BigO(M) < \BigO(N)  < 
\BigO(MN) < \BigO(M^2 N) <\BigO(M N^2) \allowbreak < \BigO(M^2 N^2)$ is given, and integers $0 ... 6$ are assigned to each big-O term in the previous 
sequence.

\paragraph{Summing Scalar Encodings}
Either the program has a single global parameter or has multiple global parameters, we have a 
mapping from complexities, which are polynomials in terms of global parameters, to their scalar 
encodings. 
Since the final objective is the total complexity of the full transformed program, we need to 
sum the scalar encoding of complexities for all statements, without losing the ability to compare the 
resultants' degrees. 
To that end, we propose to use a simple base-$|S|$ encoding method where $|S|$ is the maximum 
number of statements in the program: 
for a complexity encoded as scalar $c$, we use $|S|^c$ as a term in the final objective.  As 
an example, to sum two complexities represented in scalar $c_1$ and $c_2$,  we compute $|S|^{c_1} 
+ |S|^{c_2}$. We define the  base-$|S|$ sum of $c_i$ as $\sum |S|^{c_i}$. 

\paragraph{Indicator Variable}
In the formulation, we require indicator variables to indicate if ST is disabled along  a certain face 
-- in which case no complexity reduction should be applied for the corresponding domain. We can use the 
big-M method, a well-known ILP modeling trick \citep{ilpbigm}, to 
encode an indicator variable $y \in \{0, 1\}$ for the 
constraint $x = 0$ so  that $y = 1$ iff $x = 0$. \todo{more in appendix?}

\begin{changednospace}{mgreen}
\paragraph{Example}
To\rline{bilpcomplexityexample} formulate the objective by encoding the complexity of the $\textit{prog}'$ in 
\Cref{lst:bilpstexample}, we first 
find the asymptotic complexity of each face by counting its cardinality.
For this example, we have $\textit{complexity}(f_1) = \BigO(N^2)$, 
$\textit{complexity}(f_2) = \BigO(N)$,
$\textit{complexity}(f_3) = \BigO(N)$ and
$\textit{complexity}(f_i) = \BigO(1), \forall i \ge 4$.
Next, we assign an indicator variable to indicate if ST is disabled along a certain 
face -- a face with disabled ST implies that the domain that corresponds to the 
face is non-empty and therefore incurs cost in the objective.
Take the face $f_1$ for example, ST is disabled along this face and incurs cost iff 
$\vec{r}_1 = [0,0]^\transpose$. 
We use $\bm{1}_{\vec{r}_i = \vec{0}} \in \{0, 1\}$ to denote this indicator variable.
Assuming the ordering $\BigO(1) < \BigO(N) < \BigO(N^2)$, we may assign the 
integers $0,1, 2$ to each complexity term.
The maximum number of statements in the program is $|S| = 9$, and we use a 
base-$|S|$ encoding for summing the final objective. 
The final objective then becomes
\begin{equation}
\sum_{i}  |S|^{\textit{complexity}(f_i)} \cdot \bm{1}_{\vec{r}_i = \vec{0}} 
= 9^2 \cdot \bm{1}_{\vec{r}_1 = \vec{0}} + 9^1 \cdot  \bm{1}_{\vec{r}_2 = 
\vec{0}} 
+ 9^1 \cdot  \bm{1}_{\vec{r}_3 = \vec{0}} + 9^0 \cdot \bm{1}_{\vec{r}_4 = 
\vec{0}} + ... 9^0 \cdot \bm{1}_{\vec{r}_8 = \vec{0}}
\end{equation}
\end{changednospace}

\subsection{Discussion}
\label{sec:mssr:discuss}

The above formulation is an integer objective bilinear constrained program. 
The objective is linear because it is an affine combination of the indicator variables. 
The problem is bilinear constrained because: in the original ILP formulation scheduling, the dependence 
matrix (defined in \Cref{sec:dependencerelation}) is multiplied by a vector of unknowns to form a linear 
constraint; however by introducing the unknown reuse vectors $\vec{r}_i$, the dependence matrix 
contains entries that depends on $\vec{r}_i$, thereby making the constraints 
bilinear.

\section{SDR Heuristic Algorithm}
\label{sec:mssrsol}

The problem formulation we present in \Cref{sec:mssribpformulation} is a full 
characterization of the SDR problem. In this work we consider this formulation only as a 
specification instead of a complete solution ---  solving an integer linear objective 
bilinear constrained program is NP-hard. 
The size of the formulation (i.e. total number of constraints and number 
of variables) 
in \Cref{sec:mssribpformulation} is proportional to the number of 
statements, number of faces per statement and the maximal complexity of the program 
-- either one of which could potentially lead to exponential blow up in the size of the 
formulation. Further, our formulation of dependence resolution is based on an ILP 
formulation of multidimensional scheduling, which by itself 
already introduces a tractability challenge as pointed out in 
\citet{ilpschedulemultidim}.

For these reasons, we propose here a sound heuristic solution to SDR.
The key idea behind our heuristic approach is that for a program that 
has an affine \hl{myellow}{sequential} schedule, 
\footnote{
\hlone{msilver}{We\rline{finalremoveappendix} rely on this for the
soundness of our heuristic algorithm.
We acknowledge that enforcing sequentiality may limit 
available parallelism.  
In practice, our algorithm works with an affine sequential schedule computed by 
any algorithm.}
\ifdefined\includesequetialscheduleappendix
\hl{mred}{
We, however, provide a variant of the formalization of 
\mbox{\citet{ilpschedulemultidim}} in 
\mbox{\Cref{sec:sequential_schedule}} 
that 
includes additional constraints that ensure sequentialilty,
if one were to use the ILP approach. We note that our extension is not 
guaranteed to produce a minimal dimension schedule.
}
\fi
}
we can leverage the schedule itself to choose a reuse vector for each ST 
that we apply to the program. 
Specifically, for any reuse vector that is valid for a given face (according
the constraints in  \Cref{sec:mssribpformulation}),
our algorithm chooses either the reuse vector itself or its negation as the reuse vector for the ST. 
This algorithm -- though simple -- is
still optimal for reductions that have inverses -- which spans a broad class of programs --
and always preserves the original dependences of the program.

\subsection{Insights}
\label{sec:heuristicdirection}

The key insights of our algorithm are that 1) choosing any valid 
reuse vector for a given ST results in the same final algorithmic 
complexity for the program (implied by \citet{sr})
and 2) for any valid reuse vector, the direction itself or its negation adheres to 
a \hlone{myellow}{sequential} schedule of the left hand side of the reduction.
We demonstrate these two insights with the following lemmas.

\begin{lemma}
\label{lemma:complexitydecrease}
For any application of ST, the complexity decrease is always the same regardless of the 
actual choice of reuse vector. 
\end{lemma}
\Cref{lemma:complexitydecrease} is implied by (though not enunciated 
in) \citet[Theorem 3 and 
Section 7]{sr}.
We state \Cref{lemma:complexitydecrease} explicitly for the ease of its 
reference and give a 
self-contained proof in \Cref{sec:prooflemma61}. 

Before introducing the next lemma, we first introduce an extended definition of 
scheduling functions.
Recall that the scheduling function of a reduction statement is an affine 
function from the reduction's domain to the timestamp.
We extend the context of a scheduling function from a reduction statement to the 
left hand side of a reduction in a given program as follows.
First the program is augmented by adding to the program a new redirect 
statement $\texttt{A}[\vec{x}] = \texttt{A}'[\vec{x}]$ with the same domain as 
the 
domain of $\texttt{A}$, where $\texttt{A}'$ is a fresh symbol which replaces the 
left hand side array $\texttt{A}$ of the program.
Then the scheduling function of the left hand side of the reduction is simply the 
scheduling function of the newly introduced redirect statement of the left hand side in 
the schedule of the augmented program.

\newcommand{\thetalhs}{\ensuremath{\Theta_{\texttt{A}}}}
\newcommand{\rA}{\vec{r_\texttt{A}} }
\begin{lemma}\label{lemma:proofgreedy}

Given a \hl{myellow}{sequential} 
schedule\rline{sequentialschedulelemma} for the augmented program, 
then for any ST 
application on a reduction whose operator has an inverse and for any
valid reuse vector $\vec{r}$, either $\vec{r}$ or $-\vec{r}$ agrees 
with the \hlone{myellow}{sequential} schedule of the original program and does not 
introduce a dependence 
cycle.
\end{lemma}
\begin{proof}
Consider a reduction S with projection $\textit{proj}$ and left hand side 
array 
\texttt{A}.
Suppose \texttt{A} has an affine \hl{myellow}{sequential} schedule $\thetalhs$, then 
we have that $\thetalhs 
\cdot [\vec{x}, \vec{p}, 1]^\transpose$ is the schedule time for $\texttt{A}[\vec{x}]$. 
Let the vector $\vec{r}$ be in the same space as the domain of S, and we shift the domain 
of S along $\vec{r}$;
let the projected vector of $\vec{r}$ onto the domain of $\texttt{A}$ be 
$\rA = \textit{proj}(\vec{r})$.
Consider $\vec{x}$ and $\vec{x} + \rA$.
Their scheduled timestamps are $\thetalhs \cdot [\vec{x}, \vec{p}, 
1]^\transpose$ 
and $\thetalhs \cdot [\vec{x} + \rA, \vec{p}, 1]^\transpose$, respectively. 
For all $\vec{x}$, 
$\thetalhs \cdot [\vec{x} + \rA, \vec{p}, 1]^\transpose - 
\thetalhs \cdot [\vec{x}, \vec{p}, 1]^\transpose = \thetalhs \cdot 
[\rA, \vec{0}, 0]^\transpose$ 
is a constant that does not depend on $\vec{x}$.
\hl{myellow}{Since $\thetalhs$ 
is\rline{sequentialschedulelemmacase} sequential, the 
difference $\thetalhs \cdot [\rA, \vec{0}, 0]^\transpose$ is non-zero for 
all non-zero $\rA$.}
Therefore, 
$\texttt{A}[\vec{x}]$ is scheduled either always before 
 $\texttt{A}[\vec{x} 
+ \rA]$
or always after $\texttt{A}[\vec{x} 
+ \rA]$. 
We describe these two cases in detail as follows. 

\paragraph{Case 1}
When the first non-zero entry (in accordance with the timestamp 
comparison in \Cref{def:scheduletimestamp}) of $\thetalhs \cdot [\rA, 
\vec{0}, 0]^\transpose$ 
is positive, then $\texttt{A}[\vec{x}]$ is always scheduled before 
$\texttt{A}[\vec{x} + 
\rA]$.
In this case, applying ST with the reuse vector $\vec{r}$ will not 
introduce any 
dependence cycle, since the newly introduced dependence is always 
consistent with 
the original schedule.
\paragraph{Case 2}
When the first non-zero entry of $\thetalhs \cdot [\rA, \vec{0}, 
0]^\transpose$ is 
negative, 
$\texttt{A}[\vec{x}]$ is always scheduled after $\texttt{A}[\vec{x} + \rA]$.
In this case, applying ST with the reuse vector $-\vec{r}$ will not 
introduce any 
dependence cycle, for the same reason as in Case 1.

Further, since $\vec{r}$ chosen this way is always consistent with the original 
schedule, a previous application of ST will not affect a later application of ST --- 
intuitively, a previously applied ST introduces a dependence that can be 
subsumed by an 
enforced dependence according to the original program's schedule; thus later 
a application of ST, as long as it is also consistent with original schedule, will not 
be affected.  
\end{proof}

\subsection{Algorithm}
\label{sec:heuristicalgo}

With justification in \Cref{sec:heuristicdirection}, we now introduce the 
heuristic algorithm in \Cref{fig:heuristicalgo}. 

\begin{figure}[h]
\begin{framed}
\begin{enumerate}
	\item \hl{myellow}{Schedule the augmented program to obtain an 
	initial sequential\rline{heuristicalgorequiresequential} schedule 
	$\Theta$ for all statements and left hand side of reductions}
	\item Apply ST to all faces of all reduction statement's domains; choose the direction 
	that is consistent with $\Theta$ by:
	\begin{enumerate}
		\item First pick any valid reuse vector $\vec{r}$ from the candidate set.
		\item Test if  $\vec{r}$ is consistent with $\Theta$, if not consistent, set $\vec{r} 
		\leftarrow - \vec{r}$, if $-\vec{r}$ is also a valid reuse vector; otherwise, do not 
		apply the current ST.
	\end{enumerate}
\end{enumerate}
\end{framed}
\caption{SDR heuristic algorithm}
\label{fig:heuristicalgo}
\end{figure}

To test if $\vec{r}$ is consistent with $\Theta$, one can compute 
$\thetalhs \cdot [\rA, \vec{0}, 0]^\transpose$, with $\thetalhs$ and 
$\rA$ 
defined as in
\Cref{lemma:proofgreedy},  and then test if the first 
non-zero entry is\rline{finalalgo} positive.  \hlone{myellow}{\hphantom{xxxx}}

\paragraph{Example}\rline{heuristicexample}
\Cref{sec:example} shows an example of the above heuristic algorithm.
Specifically, the middle polyhedron in \Cref{fig:s1} illustrates the 
augmented redirect statement with the same domain as the left hand side of the 
reduction (i.e. $\polyhedral[N]{[i] : 0 \le i < N}$). 
A valid two dimensional \hlone{myellow}{sequential} schedule for the left hand side of 
the reduction is 
$\tau_{\texttt{B}} = \Theta_\texttt{B} \cdot [i, N, 
1]^\transpose = \begin{bmatrix}
0&0&1 \\ 1&0&0
\end{bmatrix} \cdot [i, N, 
1]^\transpose  = [1, i]^\transpose$. 
Then, for the reuse vector $\vec{r} = [1, 0]^\transpose$, 
we have $\Theta_\texttt{B} \cdot \vec{r}_\texttt{B} = \Theta_\texttt{B} 
\cdot [1, 0, 0]^ \transpose= [0, 1]^\transpose$. 
Since the first non-zero entry of $[0, 1]^\transpose$ is positive, ST with 
the reuse vector $[1, 0]^\transpose$ will not introduce a 
dependence cycle (i.e. consistent with the correct optimization in 
\Cref{sec:example}). 

\subsection{Algorithm Analysis}

\paragraph{Heuristic Scheduling} One advantage of the heuristic algorithm in 
\Cref{fig:heuristicalgo} is that the schedule 
$\Theta$ does not need to be obtained from forming and solving the ILP formulation as 
in \Cref{sec:ilp_schedule}, 
and one is free to choose any scheduling algorithm in the polyhedral literature such 
as \citet{zpolysched, someeffsol2affschedpt1, someeffsol2affschedpt2, plutosched}. 
Most of these algorithms, such as the PLUTO scheduler \citep{plutosched}, provide 
scalable solutions to the polyhedral scheduling problem and thus the algorithm in 
\Cref{fig:heuristicalgo} does not present bottleneck due to scheduling. 

\paragraph{Optimality Guarantee} The algorithm is optimal for the SDR problem 
if all the reduction operators have inverses. 
This is because the algorithm considers a basis direction of reuse, and picks the 
direction along that basis that is consistent with the original schedule. 
As long as all reduction operators have inverses, the heuristic algorithm will assign a 
non-zero reuse vector to each face that has valid reuse 
opportunities. In other words, the heuristic 
algorithm maximizes the total number of ST applications among all faces, if all reduce 
operators have inverses. For any ST application along a face, the complexity decrease 
is always the same 
regardless of the choice of reuse vector.
Therefore, maximizing the number of 
ST applications among all faces minimizes the total complexity.

Lastly, if a reduction operator does not have an inverse, thereby restricting the candidate set of directions, then it is possible for our algorithm to produce a non-optimal solution.
Specifically, if an operator does not have an inverse, the valid reuse vector for 
that operator will be restricted to a one sided direction (since  ST 
requires the reuse vector to point out of certain boundaries of the 
domain for a reduction if there is no inverse), instead of both directions 
of the basis. 
It is possible that the original program does not have an unique valid schedule. Consider 
the following scenario: 
one schedule  is consistent with $\vec{r}$, while another schedule is consistent with $- 
\vec{r}$; since the operator does not have an inverse, only the positive  
$\vec{r}$ is valid. Therefore, the initial schedule will affect whether this ST is applied or not -- 
which in turn leads to the suboptimality of the algorithm.

\section{Implementation}
\label{sec:implementation}

We implemented our IR as in \Cref{sec:ir} and the heuristic algorithm as in \Cref{sec:mssrsol} 
using Python. 
We used Integer Set Library (ISL) \citep{isl} for manipulation of polyhedral sets and 
relations. 
We use ISL's scheduler and code generator and generate Cython code 
\cite{cython}, 
which is then\rline{implementcpp} compiled to C++.
\hlone{myellow}{To obtain an \hl{myellow}{affine sequential schedule} for the\rline{finalimpl} original 
augmented 
program, 
we use a PLUTO-like scheduler \mbox{\cite{plutosched}} built into ISL. }

\newcommand{\tool}{Shuffle}

\section{Evaluation}
\label{sec:eval}

The algorithm presented in this work is particularly effective on optimizing
 unoptimized implementations of {\em probabilistic
inference} procedures into efficient implementations, where the inference 
procedures 
have mathematical specifications that naturally translate to our IR. 
The inference procedures are also {\em iterative}, so they contain 
dependent reductions that are not addressed by previous work \citep{sr}. 

{\em Research Question.} 
The goal of this section is to evaluate how effective the heuristic 
algorithm is at 
improving the performance on benchmarks consisting of algorithms described 
above.

\subsection{Evaluation Metrics}\label{sec:evaluationmetrics}
We considered the following two aspects to evaluate the effectiveness of the 
heuristic algorithm.

\subsubsection{Complexity}
We first evaluate asymptotic complexity of the algorithms.
This is an important metric because it determines how well these algorithms scale 
to 
large data sets, and hence how widely they can be applied. 

Optimality is defined regarding programs realizable through transformations 
presented in this work.
We evaluated our implementation in \Cref{sec:implementation} using 
unoptimized implementations of probabilistic inference procedures. 
We present their algorithmic complexities before and after optimization, as well 
as the optimal complexities achievable with transformations in this work, by 
solving the problem formulation in \Cref{sec:mssribpformulation} exactly. 
In addition, we also report the complexities of manual implementations using
transformations that are not in this work. 

We collected the complexities before and after by counting the cardinality of the 
resultant polyhedral domains 
using library implementations in \citet{barvinok}. 
We collected the optimal complexities by inspecting the benchmarks and deriving 
the 
optimal complexities manually. 
We collected complexities of manual implementations by either finding an 
existing 
implementation of the 
algorithm if one exists in the literature or, otherwise, by manually deriving them.

\subsubsection{Runtime} 
Note that a real performance gain is not necessarily implied by a complexity 
improvement because the asymptotic complexity comparison ignores constant 
factors.
The constant factors can be caused by, for example, auxillary variables overhead 
or memory/cache effect induced by ST, and the constant factors could change.
Therefore, we further validate the runtime performance gains due to lowered 
compliexites, by measuring wall-clock time improvements between the optimized 
and unoptimized implementations.

\subsection{Benchmarks}
\begin{wraptable}{r}{0.5\linewidth}
\vspace{-1.2cm}
\caption{Benchmarks: Parameters and Meanings}
\label{tab:benchmarksparamsexplain}
\centering
\small
\renewcommand{\arraystretch}{0.828}
\begin{tabular}{@{}lll@{}}
\toprule
Benchmark(s) & Param & Meaning \\
\midrule
\multirow{2}{*}{GMM-*} & $N$ & \# of observations \\
& $K$ & \# of clusters \\
\cmidrule{1-3}
\multirow{2}{*}{LDA-*} & $W$ & \# of words in corpus \\
& $K$ & \# of topics \\
\cmidrule{1-3}
\multirow{5}{*}{DMM-*} & $W$ & \# of words in corpus \\ 
& $K$ & \# of topics \\
& $D$ & \# of documents in corpus \\
& $A$ & size of alphabet of corpus \\
& $L$ & max length of document \\
\cmidrule{1-3}
\multirow{3}{*}{LBP-Stereo} & $N$ & \# of total pixels \\
& $K$ & max stereo displacement \\ 
& $D$ & \# of neighbors of a pixel \\
\cmidrule{1-3}
\multirow{2}{*}{CoxPH} & $N$ & \# of observations \\
& $K$ & dim of a single observation \\
\bottomrule
\end{tabular}
\end{wraptable}
A subset of the benchmark algorithms considered are identified as 
``model-algorithm'' pairs,
where the model refers to a generative probabilistic model, and the algorithm
refers to a class of algorithm to perform inference on the model. 
We considered 3 models and 3 algorithms.
For models, we considered the Gaussian Mixture Model (GMM)~\cite{gmm}, Latent 
Dirichlet
Allocation (LDA)~\cite{lda} and Dirichlet Multinomial Mixture (DMM)~\cite{dmm}. For
algorithms, we considered Gibbs Sampling (GS)~\cite{gibbssampling}, Metropolis
Hasting (MH)~\cite{mha,mhb} and Likelihood Weighting (LW)~\cite{likelihoodweighting}. 
Thus our benchmarks contain a total of 9 ``model-algorithm'' pairs.

Models and algorithms above have broad applications in the literature. 
The models for LDA \citep{lda, ldagibbs} and DMM \citep{dmm} are popular for 
existing data science 
problems. 
The models for GMM \citep{augurv2, histogram}, LDA \citep{augurv2, histogram}, 
and DMM \citep{histogram} have been used as benchmarks for probabilistic inference 
systems.
Gibbs sampling \citep{gibbssampling}, Metropolis-Hastings \citep{mha, 
mhb}, and Likelihood 
Weighting \citep{likelihoodweighting} 
are all widely used inference algorithms in the literature. 
LDA and DMM are 
particularly valuable 
benchmarks because there are published Gibbs sampling algorithms that researchers have manually 
optimized (\citet{ldagibbs} and \citet{dmmcollapsed}, respectively). 

In addition, we included two other benchmarks, namely the Loopy Belief Propagation 
on a 2D grid 
model for the application of Stereo matching (LBP-Stereo \citep{lbpstereo,lbppldi}
),
 and the Cox proportional hazards model (CoxPH) \citep{coxbook, coxpaper}. 
Loopy Belief Propagation \cite{pattrecogml} is an iterative approximate inference 
algorithm, and its instantiation on the 2D grid model has applications in fields such as 
vision \citep{lbpstereo,lbppldi} and physics \cite{isingphysics}. 
CoxPH is a well known statistical model, which is typically combined with Newton's 
method, an iterative optimization algorithm, for inference on the model's
parameters. 
CoxPH is commonly found in medical applications 
\citep{coxmedical1,coxmedical2}, and mechanical systems \citep{coxmechanics}.

All of the benchmarks have a common feature that they are iterative methods 
specialized to 
a generative probabilistic model. 
The parameters for these benchmarks are listed in 
\Cref{tab:benchmarksparamsexplain}.

\subsection{Results} 
As described in \Cref{sec:evaluationmetrics}, we first evaluated our method by 
analyzing the asymptotic performance 
improvements and reported the results in \Cref{sec:complexityresults}.
Then, we validated the runtime improvements of the benchmarks in 
\Cref{sec:runtimevalidation}.

\subsubsection{Complexity Results}
\label{sec:complexityresults}

\begin{table}[t]
\caption{Benchmarks: Comparison of Complexities}
\label{tab:benchmarkstable}
\centering
\small
\begin{adjustbox}{max width=\textwidth}
\renewcommand{\arraystretch}{0.9}
\begin{tabular}{@{}lllllq@{}}
\toprule
Benchmark & Original & Optimized (Heuristic) & SDR-Optimal & 
Manual & \#IR \\
\midrule
GMM-GS & $\BigO(N^2K^2)$    & $\BigO(NK)$ & $\BigO(NK)$ & $\BigO(NK)$ & 16 \\

GMM-MH       & $\BigO(N(N + K))$ & $\BigO(N)$ & $\BigO(N)$ & $\BigO(N)$ & 16  \\
GMM-LW       & $\BigO(N(N + K))$  & $\BigO(N)$ & $\BigO(N)$ & $\BigO(N)$ & 5 \\
LDA-GS   & $\BigO(W^2K^2)$ & $\BigO(WK)$ & $\BigO(WK)$ & $\BigO(WK)$ & 20 \\
LDA-MH      & $\BigO(W^2K)$  & $\BigO(WK)$ & $\BigO(WK)$ & $\BigO(WK)$ & 42 \\
LDA-LW       & $\BigO(W^2K)$  & $\BigO(W)$ & $\BigO(W)$ & $\BigO(W)$ & 7  \\
DMM-GS & $\BigO(WADK^2 + D^2K^2)$ & $\BigO((W+A)KD)$ & $\BigO((W+A)KD)$ & 
$\BigO(AKD)$ & 40 \\
DMM-MH & $\BigO(D^2K^2 + D(W + A))$ & $\BigO((K + W + A)D)$ & $\BigO((K + W + 
A)D)$ & 
$\BigO((K + L + A)D)$ & 82  \\
DMM-LW & $\BigO((WA+K)D)$ & $\BigO((K + W +A)D)$ & $\BigO((K + W +A)D)$ & 
$\BigO((K + L + 
A)D)$ & 10 \\
LBP-Stereo & $\BigO(N K D^2)$ & $\BigO(N K D)$ & $\BigO(N K D)$ &$\BigO(N K D)$  
& 3 \\
CoxPH & $\BigO(K^2 N^2)$ & $\BigO(K^2 N)$ & $\BigO(K^2 N)$ & $\BigO(K^2 N)$ & 6 
\\
\bottomrule
\end{tabular}
\end{adjustbox}
\end{table}

\Cref{tab:benchmarkstable} summarizes the results on comparison of complexities, 
expressed in terms of the corresponding parameters of each benchmark.

The column ``Original'' gives the complexity of the original program for the 
benchmarks.
The column ``Optimized (Heuristic)'' (later abbreviated as ``Optimized'') gives 
the complexity of the transformed program using the 
heuristic 
implementation in \Cref{sec:implementation}.
The column ``SDR-Optimal'' gives the complexity of the transformed program 
by 
potentially solving 
the problem formulation in \Cref{sec:mssribpformulation} exactly 
\footnote{For\rline{howsolveexactly} all of our benchmarks, it turns out 
that our heuristic 
algorithm is able to assign a non-zero reuse vector to each 
non-degenerate face (i.e. one whose cardinality is not 
$\BigO(1)$). This implies that our heuristic algorithm's solution is 
already the optimal solution to SDR, since it does not have further 
complexity reduction opportunities. 
Therefore, through this argument we obtain the SDR-optimal 
complexities without having to actually construct and solve the bilinear 
program.
}; 
and this is the optimal complexity one can achieve using 
techniques presented in this work.
The column ``Manual'' gives the complexity of a potential optimized manual 
implementation written by a  developer; this means that the complexity reduction 
potentially comes from 
transformations not present in this work.

Comparing the ``Original'' and ``Optimized'' columns,  our approach 
reduces the complexities for all benchmarks.
Comparing the ``Optimized'' and ``SDR-Optimal'' columns, our 
heuristic approach produces complexities same as solving SDR optimally for all 
benchmarks. 
Comparing the ``Optimized'' and ``Manual'' columns, our approach 
produces complexities the same as manual implementations for 8 out of 11 
benchmarks. 
We identified that the 
 3 benchmarks related to DMM require additional data layout modifications which we 
 did not consider in this work, which is a direction of future research.

\subsubsection{Runtime Validation}
\label{sec:runtimevalidation}
{
\newcommand{\tablerow}[4]{%
#1 &%
\ifthenelse{\equal{#3}{timeout}}{\textbf{#3}}{\SI{#3}{#2}} & %
\SI{#4}{#2} & %
{
	\ifthenelse{\equal{#3}{timeout}}{%
		\ifthenelse{\equal{#2}{\s}}{%
			\FPset{\numerator}{43200}
		}{%
			\FPset{\numerator}{43200000}
		}
	}{
		\FPset{\numerator}{#3}
	}%
	\FPset{\denom}{#4}
	\FPdiv{\speedup}{\numerator}{\denom}
	\bf \ifthenelse{\equal{#3}{timeout}}{>}{}
	\num[detect-weight=true,scientific-notation=engineering,round-mode=figures]{%
	\FPprint\speedup}
	 x
}%
}
\begin{table}[t]
\hspace{-.2cm}
\begin{minipage}[t]{0.4\linewidth}
\caption{Benchmarks: Parameter Sizes}
\label{tab:benchmarksparamssizes}
\centering
\small
\renewcommand{\arraystretch}{0.828}
\begin{tabular}[t]{@{}lcc@{}}
\toprule
Benchmark(s) & Parameter & Size \\
\midrule
\multirow{2}{*}{GMM-*} & $N$ & $10000$ \\
& $K$ & $10$ \\
\cmidrule{1-3}
\multirow{2}{*}{LDA-*} & $W$ & $466,000$ \\
& $K$ & $50$ \\
\cmidrule{1-3}
\multirow{5}{*}{DMM-*} & $W$ & $570,000$ \\ 
& $K$ & $4$ \\
& $D$ & $278$ \\
& $A$ & $129$ \\
& $L$ & $3202$ \\
\cmidrule{1-3}
\multirow{3}{*}{LBP-Stereo} & $N$ & $110,592$ \\
& $K$ & $16$ \\ 
& $D$ & $4$ \\
\cmidrule{1-3}
\multirow{2}{*}{CoxPH} & $N$ & $424$ \\
& $K$ & $11$ \\
\bottomrule
\end{tabular}
\end{minipage}\hfil%
\begin{minipage}[t]{0.6\linewidth}
\caption{Runtime evaluations}
\label{tab:shufflebench}
\centering
\begin{tabular}[t]{@{}lcc@{}c@{}}
\toprule
\multirow{2}{1.5cm}{Benchmark}  & 
\multirow{2}{1.5cm}{Original} & 
\multirow{2}{1.5cm}{Optimized (Heuristic)} & 
\multirow{2}{2.5cm}{\centering Speedup}  \\
&&& \\
\midrule
\tablerow{GMM-GS}{\ms}{29.2}{4.1} \\
\tablerow{GMM-MH}{\ms}{0.94}{0.72} \\
\tablerow{GMM-LW}{\s}{2.5}{1.7} \\
\tablerow{LDA-GS }{\ms}{timeout}{7.1} \\
\tablerow{LDA-MH}{\s}{timeout}{0.69} \\
\tablerow{LDA-LW}{\s}{timeout}{41.1} \\
\tablerow{DMM-GS}{\s}{2.2}{0.54} \\
\tablerow{DMM-MH}{\ms}{529}{14.3} \\
\tablerow{DMM-LW}{\s}{71}{1.23} \\
\tablerow{LBP-Stereo}{\s}{16.9}{15.1} \\
\tablerow{CoxPH}{\ms}{69.1}{8.4} \\
\bottomrule
\end{tabular}
\end{minipage}
\end{table}
}%
So far we have evaluated our heuristic algorithm using algorithmic complexity as 
the 
primary factor, which ignores constant factors caused by, for example, auxillary 
variables overhead or memory/cache effect induced by ST. 
In this section, we validate our hypothesis that asymptotic complexity improvements 
dominates potential constant factors improvements for the parameters of these 
benchmarks by timing our benchmarks and comparing the runtimes 
of the naive implementations with the optimized implementations. 
We\rline{implementcpp2} generate C++ code for the IR programs 
corresponding to the 
Original and Optimized (Heuristic)  columns 
in~\Cref{tab:benchmarkstable} . 
We ran these implementations and report timeouts for benchmarks that ran for 
12+ hours. 

\paragraph{Parameter Sizes} 
We collected the typical instantiated values for global parameters from the 
corresponding literature.
Specifically, 
for GMM we use
\citet{augurv2}, for LDA we use \citet{kos}, for DMM we use \citet{twins}, for 
LBP-Stereo we use \citet{lbpstereo} 
and for CoxPH we use \citet[Appendix~D2]{coxbook}.
Based on these prior works, we collected the following parameters for each 
model in \Cref{tab:benchmarksparamssizes}.

\paragraph{Results}
In \Cref{tab:shufflebench}, all benchmarks have non-trivial speedups. 
In particular, for LDA benchmarks all the unoptimized implementations timeout. 
This is because, in \Cref{tab:benchmarkstable}, the complexities of LDA benchmarks 
all improve by a factor of $W \times K$.
With our instantiated parameter values this factor is $466,000 \times 50 = 2.3 
\times 10^7$ -- the largest factor across all benchmarks. This large a factor 
unsurprisingly leads to the timeout of the unoptimzied implementations.
On the other hand, in \Cref{tab:benchmarkstable} LBP-Stereo's complexity only 
improves by a factor of $D$, the number of neighbors of a pixel, which is set to 4 
(i.e. number immediate neighbors of a pixel) in our parameter setting. 
Nonetheless, we observe a speedup of $1.1$x for this benchmark. 
We also note that this speedup scales with the specification of the LBP model 
-- setting to 8 neighboring pixels (i.e. nearby 8 pixels for a pixel at the center of 
a 3 by 3 square) would lead to speed up of of $1.4$x. 
In sum, the observed speedups validates that for these benchmarks and our 
technique,  complexity dominates constant-factor concerns.

\section{Related Work}
\label{sec:relatedwork}

\newcommand{\relatedworkparagraph}[1]{\item[#1]}

\begin{description}[style=unboxed,leftmargin=0cm]

\relatedworkparagraph{Reductions}
\citet{pollyred, ginsbach17} introduced techniques to 
detect reductions from loop based 
code; these techniques can be used as front-ends to our technique for 
conversion into our reduction based IR. 
\citet{rauchwerger99, pollyred, reddrawing, ginsbach17} 
optimized reductions in the 
polyhedral model for considerations such as privatization and 
parallelization. They do not optimize for complexities; however, they can be used as 
optimizing backends for generating efficient code for reductions after applying our 
method. 
\citet{programequivalencereduction} proposed a 
semi-algorithm that decides equivalence between programs with reductions; 
it can be used to check correctness of a program transformation.

\relatedworkparagraph{Simplifying Reductions} 
\citet{optaggarray} introduced a loop based transformation algorithm for reducing
complexities on loop programs. The algorithm uses the Omega calculator
\citep{omega} for analysis on a {\em contributing set}. 
The contributing set based anaylsis is general in the sense that it works for 
non-polyhedral sets as well. 
\citet{optaggarray} used only the direction of loop 
increment to decrease the complexity. 
\citet{sr} generalized the method in 
\citet{optaggarray} and introduced the simplifying reductions framework; 
the framework is more general in that it uses reuse vectors not limited to just the 
direction of loop increments.

\relatedworkparagraph{Scheduling}
We review related works on polyhedral scheduling in \Cref{sec:more_related_works}.

\relatedworkparagraph{Program Complexity}
\citet{determinecomplexity1} and \citet{determinecomplexity2} studied
methods\rline{relatedworkcomplexity} to 
infer the complexity of programs. 
These works can reduce the effort for manually 
analyzing the hand optimized benchmarks, which is particularly helpful 
for enlarging the benchmark suite.

\relatedworkparagraph{Incrementalization in Probabilistic Programming}
The problem of incrementalization occurs in probabilistic programming 
system (PPS), and is known as incrementalized inference. 
Existing work \citep{hakaru10, c3, swift, r2, yang14, zhang19} attempted to 
address the problem of incrementalized inference in PPS. 
However, these techniques are variants/combinations of 
1) tracing JITs, 
2) specialization and caching/memoization, 
3) dynamic dependence analysis, 
4) dynamic program slicing, or 
5) runtime symbolic analysis 
-- in summary, dynamic optimizations. These techniques introduce 
significant runtime overhead for storing dependence graph/traces (which is of size proportional to the 
number of the executed statement instances) and/or performing analysis on those graphs/traces 
dynamically. 
Our technique can be applied to PPS to solve the incrementalized inference problem; however, our technique 
is a static compilation techinique which do not suffer from runtime 
overhead.  

Many\rline{incorpratepps} existing and ongoing work 
\citep{church,augurv2,webppl,pyro,jags,hakaru,venture,shuffle} allow
the user to code in high level DSLs.
Though the details on these systems are out of the scope of the paper, 
our method can be potentially integrated into these systems for 
generating code with efficient complexity, which we discuss more
in \Cref{sec:incorprateintopps}.

\end{description}

\section{Conclusion}
\label{sec:conclusion}

In this work, we introduce the simplifying dependent reductions problem
and provide a heuristic algorithm that --- \hlone{myellow}{given an 
affine\rline{finalconclusion} sequential 
schedule for the program} --- is optimal for reduction operators that
have inverses. These reductions have otherwise only appeared 
as domain- or algorithm-specific optimizations as described in the published
description of standard probabilistic inference algorithms. Our
hope is that this work formally outlines a key general-purpose optimization
opportunity that can be delegated to the compiler, rather than
being a significant piece of manual implementation that stands between
the elaboration of a new probabilistic inference algorithm and its high-performance
implementation. Our results hold the promise that emerging languages 
and systems for this increasingly important class of computations could 
see significant performance improvements
by incorporating our techniques.

\begin{acks}                            %

\newcommand{\acklist}{%
\sortalpha{%
{Alex Renda}%
{Jesse Michel}%
{Jonathan Frankle}%
{Tian Jin}%
{Sriram Krishnamoorthy}%
{Charith Mandis}%
{Riyadh Baghdadi}%
{Sanjay Rajopadhye}%
}}

We would like to thank \acklist, and anonymous reviewers for their helpful comments and suggestions.
This work was supported in part by the Office of Naval Research (ONR-N00014-17-1-2699).
Any opinions, findings, and conclusions or recommendations expressed in this material are those of the author and do not necessarily 
reflect the views of the Office of Naval Research.

\end{acks}

\bibliography{paper}

%%% -*-BibTeX-*-
%%% Do NOT edit. File created by BibTeX with style
%%% ACM-Reference-Format-Journals [18-Jan-2012].

\begin{thebibliography}{14}

%%% ====================================================================
%%% NOTE TO THE USER: you can override these defaults by providing
%%% customized versions of any of these macros before the \bibliography
%%% command.  Each of them MUST provide its own final punctuation,
%%% except for \shownote{}, \showDOI{}, and \showURL{}.  The latter two
%%% do not use final punctuation, in order to avoid confusing it with
%%% the Web address.
%%%
%%% To suppress output of a particular field, define its macro to expand
%%% to an empty string, or better, \unskip, like this:
%%%
%%% \newcommand{\showDOI}[1]{\unskip}   % LaTeX syntax
%%%
%%% \def \showDOI #1{\unskip}           % plain TeX syntax
%%%
%%% ====================================================================

\ifx \showCODEN    \undefined \def \showCODEN     #1{\unskip}     \fi
\ifx \showDOI      \undefined \def \showDOI       #1{#1}\fi
\ifx \showISBNx    \undefined \def \showISBNx     #1{\unskip}     \fi
\ifx \showISBNxiii \undefined \def \showISBNxiii  #1{\unskip}     \fi
\ifx \showISSN     \undefined \def \showISSN      #1{\unskip}     \fi
\ifx \showLCCN     \undefined \def \showLCCN      #1{\unskip}     \fi
\ifx \shownote     \undefined \def \shownote      #1{#1}          \fi
\ifx \showarticletitle \undefined \def \showarticletitle #1{#1}   \fi
\ifx \showURL      \undefined \def \showURL       {\relax}        \fi
% The following commands are used for tagged output and should be
% invisible to TeX
\providecommand\bibfield[2]{#2}
\providecommand\bibinfo[2]{#2}
\providecommand\natexlab[1]{#1}
\providecommand\showeprint[2][]{arXiv:#2}

\bibitem[\protect\citeauthoryear{Atkinson, Yang, and Carbin}{Atkinson
  et~al\mbox{.}}{2018}]%
        {shuffle}
\bibfield{author}{\bibinfo{person}{Eric Atkinson}, \bibinfo{person}{Cambridge
  Yang}, {and} \bibinfo{person}{Michael Carbin}.}
  \bibinfo{year}{2018}\natexlab{}.
\newblock \showarticletitle{Verifying Handcoded Probabilistic Inference
  Procedures}. In \bibinfo{booktitle}{\emph{arXiv e-prints}}.
\newblock


\bibitem[\protect\citeauthoryear{Bondhugula, Hartono, Ramanujam, and
  Sadayappan}{Bondhugula et~al\mbox{.}}{2008}]%
        {plutosched}
\bibfield{author}{\bibinfo{person}{Uday Bondhugula}, \bibinfo{person}{Albert
  Hartono}, \bibinfo{person}{J. Ramanujam}, {and} \bibinfo{person}{P.
  Sadayappan}.} \bibinfo{year}{2008}\natexlab{}.
\newblock \showarticletitle{A Practical Automatic Polyhedral Parallelizer and
  Locality Optimizer}. In \bibinfo{booktitle}{\emph{Conference on Programming
  Language Design and Implementation}}.
\newblock


\bibitem[\protect\citeauthoryear{Cusumano-Towner, Saad, Lew, and
  Mansinghka}{Cusumano-Towner et~al\mbox{.}}{2019}]%
        {gen}
\bibfield{author}{\bibinfo{person}{Marco~F Cusumano-Towner},
  \bibinfo{person}{Feras~A Saad}, \bibinfo{person}{Alexander~K Lew}, {and}
  \bibinfo{person}{Vikash~K Mansinghka}.} \bibinfo{year}{2019}\natexlab{}.
\newblock \showarticletitle{Gen: a General-purpose Probabilistic Programming
  System with Programmable Inference}. In \bibinfo{booktitle}{\emph{Conference
  on Programming Language Design and Implementation}}.
\newblock


\bibitem[\protect\citeauthoryear{Daniel~Huang}{Daniel~Huang}{2017}]%
        {augurv2}
\bibfield{author}{\bibinfo{person}{Greg~Morisett Daniel~Huang,
  Jean-Baptiste~Tristan}.} \bibinfo{year}{2017}\natexlab{}.
\newblock \showarticletitle{Compiling Markov Chain Monte Carlo Algorithms for
  Probabilistic Modeling}. In \bibinfo{booktitle}{\emph{Conference on
  Programming Language Design and Implementation}}.
\newblock


\bibitem[\protect\citeauthoryear{Doerfert, Streit, Hack, and Benaissa}{Doerfert
  et~al\mbox{.}}{2015}]%
        {pollyred}
\bibfield{author}{\bibinfo{person}{Johannes Doerfert}, \bibinfo{person}{Kevin
  Streit}, \bibinfo{person}{Sebastian Hack}, {and} \bibinfo{person}{Zino
  Benaissa}.} \bibinfo{year}{2015}\natexlab{}.
\newblock \showarticletitle{Polly's Polyhedral Scheduling in the Presence of
  Reductions}. In \bibinfo{booktitle}{\emph{International Workshop on
  Polyhedral Compilation Techniques}}.
\newblock


\bibitem[\protect\citeauthoryear{Feautrier}{Feautrier}{1992a}]%
        {someeffsol2affschedpt1}
\bibfield{author}{\bibinfo{person}{Paul Feautrier}.}
  \bibinfo{year}{1992}\natexlab{a}.
\newblock \showarticletitle{Some efficient solutions to the affine scheduling
  problem. I. One-dimensional time}.
\newblock \bibinfo{journal}{\emph{International Journal of Parallel
  Programming}} \bibinfo{volume}{21}, \bibinfo{number}{5} (\bibinfo{date}{Oct.}
  \bibinfo{year}{1992}), \bibinfo{pages}{313--347}.
\newblock


\bibitem[\protect\citeauthoryear{Feautrier}{Feautrier}{1992b}]%
        {someeffsol2affschedpt2}
\bibfield{author}{\bibinfo{person}{Paul Feautrier}.}
  \bibinfo{year}{1992}\natexlab{b}.
\newblock \showarticletitle{Some efficient solutions to the affine scheduling
  problem. Part II. Multidimensional time}.
\newblock \bibinfo{journal}{\emph{International Journal of Parallel
  Programming}} \bibinfo{volume}{21}, \bibinfo{number}{6} (\bibinfo{date}{Dec.}
  \bibinfo{year}{1992}), \bibinfo{pages}{389--420}.
\newblock


\bibitem[\protect\citeauthoryear{Gautam and Rajopadhye}{Gautam and
  Rajopadhye}{2006}]%
        {sr}
\bibfield{author}{\bibinfo{person}{Gautam} {and} \bibinfo{person}{S.
  Rajopadhye}.} \bibinfo{year}{2006}\natexlab{}.
\newblock \showarticletitle{Simplifying Reductions}. In
  \bibinfo{booktitle}{\emph{Symposium on Principles of Programming Languages}}.
\newblock


\bibitem[\protect\citeauthoryear{Hall, Murphy, Amarasinghe, Liao, and Lam}{Hall
  et~al\mbox{.}}{1996}]%
        {interproceduralanalysisforparallelization}
\bibfield{author}{\bibinfo{person}{Mary~W. Hall}, \bibinfo{person}{Brian~R.
  Murphy}, \bibinfo{person}{Saman~P. Amarasinghe}, \bibinfo{person}{Shih~Wei
  Liao}, {and} \bibinfo{person}{Monica~S. Lam}.}
  \bibinfo{year}{1996}\natexlab{}.
\newblock \showarticletitle{Interprocedural analysis for parallelization}. In
  \bibinfo{booktitle}{\emph{Languages and Compilers for Parallel Computing}},
  \bibfield{editor}{\bibinfo{person}{Chua-Huang Huang},
  \bibinfo{person}{Ponnuswamy Sadayappan}, \bibinfo{person}{Utpal Banerjee},
  \bibinfo{person}{David Gelernter}, \bibinfo{person}{Alex Nicolau}, {and}
  \bibinfo{person}{David Padua}} (Eds.). \bibinfo{publisher}{Springer Berlin
  Heidelberg}, \bibinfo{address}{Berlin, Heidelberg}, \bibinfo{pages}{61--80}.
\newblock


\bibitem[\protect\citeauthoryear{Narayanan, Carette, Romano, Shan, and
  Zinkov}{Narayanan et~al\mbox{.}}{2016}]%
        {hakaru}
\bibfield{author}{\bibinfo{person}{Praveen Narayanan}, \bibinfo{person}{Jacques
  Carette}, \bibinfo{person}{Wren Romano}, \bibinfo{person}{Chung{-}chieh
  Shan}, {and} \bibinfo{person}{Robert Zinkov}.}
  \bibinfo{year}{2016}\natexlab{}.
\newblock \showarticletitle{Probabilistic Inference by Program Transformation
  in Hakaru (System Description)}. In \bibinfo{booktitle}{\emph{International
  Symposium on Functional and Logic Programming}}.
\newblock


\bibitem[\protect\citeauthoryear{Pouchet, Bastoul, Cohen, and Cavazos}{Pouchet
  et~al\mbox{.}}{2008}]%
        {pouchet.08.pldi}
\bibfield{author}{\bibinfo{person}{Louis-No{\"e}l Pouchet},
  \bibinfo{person}{C{\'e}dric Bastoul}, \bibinfo{person}{Albert Cohen}, {and}
  \bibinfo{person}{John Cavazos}.} \bibinfo{year}{2008}\natexlab{}.
\newblock \showarticletitle{Iterative optimization in the polyhedral model:
  Part {II}, multidimensional time}. In \bibinfo{booktitle}{\emph{Conference on
  Programming Language Design and Implementation}}.
\newblock


\bibitem[\protect\citeauthoryear{Pouchet, Bastoul, Cohen, and
  Vasilache}{Pouchet et~al\mbox{.}}{2007}]%
        {pouchet.07.cgo}
\bibfield{author}{\bibinfo{person}{Louis-No{\"e}l Pouchet},
  \bibinfo{person}{C{\'e}dric Bastoul}, \bibinfo{person}{Albert Cohen}, {and}
  \bibinfo{person}{Nicolas Vasilache}.} \bibinfo{year}{2007}\natexlab{}.
\newblock \showarticletitle{Iterative optimization in the polyhedral model:
  Part {I}, one-dimensional time}. In \bibinfo{booktitle}{\emph{International
  Symposium on Code Generation and Optimization}}.
\newblock


\bibitem[\protect\citeauthoryear{Pouchet, Bondhugula, Bastoul, Cohen,
  Ramanujam, Sadayappan, and Vasilache}{Pouchet et~al\mbox{.}}{2011}]%
        {ilpschedulemultidim}
\bibfield{author}{\bibinfo{person}{Louis-No\"{e}l Pouchet},
  \bibinfo{person}{Uday Bondhugula}, \bibinfo{person}{C\'{e}dric Bastoul},
  \bibinfo{person}{Albert Cohen}, \bibinfo{person}{J. Ramanujam},
  \bibinfo{person}{P. Sadayappan}, {and} \bibinfo{person}{Nicolas Vasilache}.}
  \bibinfo{year}{2011}\natexlab{}.
\newblock \showarticletitle{Loop Transformations: Convexity, Pruning and
  Optimization}. In \bibinfo{booktitle}{\emph{Symposium on Principles of
  Programming Languages}}.
\newblock


\bibitem[\protect\citeauthoryear{Verdoolaege}{Verdoolaege}{2010}]%
        {isl}
\bibfield{author}{\bibinfo{person}{Sven Verdoolaege}.}
  \bibinfo{year}{2010}\natexlab{}.
\newblock \showarticletitle{isl: An Integer Set Library for the Polyhedral
  Model}. In \bibinfo{booktitle}{\emph{International Congress on Mathematical
  Software}}.
\newblock


\end{thebibliography}


%%% -*-BibTeX-*-
%%% Do NOT edit. File created by BibTeX with style
%%% ACM-Reference-Format-Journals [18-Jan-2012].

\begin{thebibliography}{71}

%%% ====================================================================
%%% NOTE TO THE USER: you can override these defaults by providing
%%% customized versions of any of these macros before the \bibliography
%%% command.  Each of them MUST provide its own final punctuation,
%%% except for \shownote{}, \showDOI{}, and \showURL{}.  The latter two
%%% do not use final punctuation, in order to avoid confusing it with
%%% the Web address.
%%%
%%% To suppress output of a particular field, define its macro to expand
%%% to an empty string, or better, \unskip, like this:
%%%
%%% \newcommand{\showDOI}[1]{\unskip}   % LaTeX syntax
%%%
%%% \def \showDOI #1{\unskip}           % plain TeX syntax
%%%
%%% ====================================================================

\ifx \showCODEN    \undefined \def \showCODEN     #1{\unskip}     \fi
\ifx \showDOI      \undefined \def \showDOI       #1{#1}\fi
\ifx \showISBNx    \undefined \def \showISBNx     #1{\unskip}     \fi
\ifx \showISBNxiii \undefined \def \showISBNxiii  #1{\unskip}     \fi
\ifx \showISSN     \undefined \def \showISSN      #1{\unskip}     \fi
\ifx \showLCCN     \undefined \def \showLCCN      #1{\unskip}     \fi
\ifx \shownote     \undefined \def \shownote      #1{#1}          \fi
\ifx \showarticletitle \undefined \def \showarticletitle #1{#1}   \fi
\ifx \showURL      \undefined \def \showURL       {\relax}        \fi
% The following commands are used for tagged output and should be
% invisible to TeX
\providecommand\bibfield[2]{#2}
\providecommand\bibinfo[2]{#2}
\providecommand\natexlab[1]{#1}
\providecommand\showeprint[2][]{arXiv:#2}

\bibitem[\protect\citeauthoryear{Alias, Darte, Feautrier, and Gonnord}{Alias
  et~al\mbox{.}}{2010}]%
        {determinecomplexity2}
\bibfield{author}{\bibinfo{person}{Christophe Alias}, \bibinfo{person}{Alain
  Darte}, \bibinfo{person}{Paul Feautrier}, {and} \bibinfo{person}{Laure
  Gonnord}.} \bibinfo{year}{2010}\natexlab{}.
\newblock \showarticletitle{{Multi-dimensional Rankings, Program Termination,
  and Complexity Bounds of Flowchart Programs}}. In
  \bibinfo{booktitle}{\emph{{Static Analysis Symposium}}}.
\newblock
\urldef\tempurl%
\url{https://hal.inria.fr/inria-00523298}
\showURL{%
\tempurl}


\bibitem[\protect\citeauthoryear{Atkinson, Yang, and Carbin}{Atkinson
  et~al\mbox{.}}{2018}]%
        {shuffle}
\bibfield{author}{\bibinfo{person}{Eric Atkinson}, \bibinfo{person}{Cambridge
  Yang}, {and} \bibinfo{person}{Michael Carbin}.}
  \bibinfo{year}{2018}\natexlab{}.
\newblock \showarticletitle{Verifying Handcoded Probabilistic Inference
  Procedures}. In \bibinfo{booktitle}{\emph{arXiv e-prints}}.
\newblock


\bibitem[\protect\citeauthoryear{Behnel, Bradshaw, Citro, Dalcin, Seljebotn,
  and Smith}{Behnel et~al\mbox{.}}{2011}]%
        {cython}
\bibfield{author}{\bibinfo{person}{Stefan Behnel}, \bibinfo{person}{Robert
  Bradshaw}, \bibinfo{person}{Craig Citro}, \bibinfo{person}{Lisandro Dalcin},
  \bibinfo{person}{Dag~Sverre Seljebotn}, {and} \bibinfo{person}{Kurt Smith}.}
  \bibinfo{year}{2011}\natexlab{}.
\newblock \showarticletitle{Cython: The Best of Both Worlds}.
\newblock \bibinfo{journal}{\emph{Computing in Science and Engg.}}
  \bibinfo{volume}{13}, \bibinfo{number}{2} (\bibinfo{date}{March}
  \bibinfo{year}{2011}), \bibinfo{pages}{31–39}.
\newblock


\bibitem[\protect\citeauthoryear{Benabderrahmane, Pouchet, Cohen, and
  Bastoul}{Benabderrahmane et~al\mbox{.}}{2010}]%
        {polyhedralmodelismoreapplicable}
\bibfield{author}{\bibinfo{person}{Mohamed-Walid Benabderrahmane},
  \bibinfo{person}{Louis-No\"{e}l Pouchet}, \bibinfo{person}{Albert Cohen},
  {and} \bibinfo{person}{C{\'e}dric Bastoul}.} \bibinfo{year}{2010}\natexlab{}.
\newblock \showarticletitle{The Polyhedral Model is More Widely Applicable Than
  You Think}. In \bibinfo{booktitle}{\emph{European Conference on Theory and
  Practice of Software, International Conference on Compiler Construction}}.
\newblock


\bibitem[\protect\citeauthoryear{Bingham, Chen, Jankowiak, Obermeyer, Pradhan,
  Karaletsos, Singh, Szerlip, Horsfall, and Goodman}{Bingham
  et~al\mbox{.}}{2019}]%
        {pyro}
\bibfield{author}{\bibinfo{person}{Eli Bingham}, \bibinfo{person}{Jonathan~P.
  Chen}, \bibinfo{person}{Martin Jankowiak}, \bibinfo{person}{Fritz Obermeyer},
  \bibinfo{person}{Neeraj Pradhan}, \bibinfo{person}{Theofanis Karaletsos},
  \bibinfo{person}{Rohit Singh}, \bibinfo{person}{Paul Szerlip},
  \bibinfo{person}{Paul Horsfall}, {and} \bibinfo{person}{Noah~D. Goodman}.}
  \bibinfo{year}{2019}\natexlab{}.
\newblock \showarticletitle{Pyro: Deep Universal Probabilistic Programming}.
\newblock \bibinfo{journal}{\emph{Journal of Machine Learning Research}}
  \bibinfo{volume}{20}, \bibinfo{number}{28} (\bibinfo{year}{2019}).
\newblock


\bibitem[\protect\citeauthoryear{Bishop}{Bishop}{2006}]%
        {pattrecogml}
\bibfield{author}{\bibinfo{person}{Christopher~M. Bishop}.}
  \bibinfo{year}{2006}\natexlab{}.
\newblock \bibinfo{booktitle}{\emph{Pattern Recognition and Machine Learning
  (Information Science and Statistics)}}.
\newblock \bibinfo{publisher}{Springer-Verlag}, \bibinfo{address}{Berlin,
  Heidelberg}.
\newblock


\bibitem[\protect\citeauthoryear{Blei, Ng, and Jordan}{Blei
  et~al\mbox{.}}{2003}]%
        {lda}
\bibfield{author}{\bibinfo{person}{David~M. Blei}, \bibinfo{person}{Andrew~Y.
  Ng}, {and} \bibinfo{person}{Michael~I. Jordan}.}
  \bibinfo{year}{2003}\natexlab{}.
\newblock \showarticletitle{Latent Dirichlet Allocation}.
\newblock \bibinfo{journal}{\emph{Journal of Machine Learning Research}}
  \bibinfo{volume}{3} (\bibinfo{date}{Jan.} \bibinfo{year}{2003}),
  \bibinfo{pages}{993--1022}.
\newblock


\bibitem[\protect\citeauthoryear{Bondhugula, Hartono, Ramanujam, and
  Sadayappan}{Bondhugula et~al\mbox{.}}{2008}]%
        {plutosched}
\bibfield{author}{\bibinfo{person}{Uday Bondhugula}, \bibinfo{person}{Albert
  Hartono}, \bibinfo{person}{J. Ramanujam}, {and} \bibinfo{person}{P.
  Sadayappan}.} \bibinfo{year}{2008}\natexlab{}.
\newblock \showarticletitle{A Practical Automatic Polyhedral Parallelizer and
  Locality Optimizer}. In \bibinfo{booktitle}{\emph{Conference on Programming
  Language Design and Implementation}}.
\newblock


\bibitem[\protect\citeauthoryear{Collard, Barthou, and Feautrier}{Collard
  et~al\mbox{.}}{1995}]%
        {fuzzyarray}
\bibfield{author}{\bibinfo{person}{Jean-Fran\c{c}ois Collard},
  \bibinfo{person}{Denis Barthou}, {and} \bibinfo{person}{Paul Feautrier}.}
  \bibinfo{year}{1995}\natexlab{}.
\newblock \showarticletitle{Fuzzy Array Dataflow Analysis}. In
  \bibinfo{booktitle}{\emph{Symposium on Principles and Practice of Parallel
  Programming}}.
\newblock


\bibitem[\protect\citeauthoryear{Collett}{Collett}{1993}]%
        {coxmedical1}
\bibfield{author}{\bibinfo{person}{D Collett}.}
  \bibinfo{year}{1993}\natexlab{}.
\newblock \bibinfo{booktitle}{\emph{Modelling Survival Data in Medical
  Research}}.
\newblock \bibinfo{publisher}{Chapman \& Hall}, \bibinfo{address}{New York}.
\newblock


\bibitem[\protect\citeauthoryear{Cox}{Cox}{1972}]%
        {coxpaper}
\bibfield{author}{\bibinfo{person}{D.~R. Cox}.}
  \bibinfo{year}{1972}\natexlab{}.
\newblock \showarticletitle{Regression Models and Life-Tables}.
\newblock \bibinfo{journal}{\emph{Journal of the Royal Statistical Society:
  Series B (Methodological)}} \bibinfo{volume}{34}, \bibinfo{number}{2}
  (\bibinfo{year}{1972}), \bibinfo{pages}{187--202}.
\newblock


\bibitem[\protect\citeauthoryear{Cusumano-Towner, Saad, Lew, and
  Mansinghka}{Cusumano-Towner et~al\mbox{.}}{2019}]%
        {gen}
\bibfield{author}{\bibinfo{person}{Marco~F Cusumano-Towner},
  \bibinfo{person}{Feras~A Saad}, \bibinfo{person}{Alexander~K Lew}, {and}
  \bibinfo{person}{Vikash~K Mansinghka}.} \bibinfo{year}{2019}\natexlab{}.
\newblock \showarticletitle{Gen: a General-purpose Probabilistic Programming
  System with Programmable Inference}. In \bibinfo{booktitle}{\emph{Conference
  on Programming Language Design and Implementation}}.
\newblock


\bibitem[\protect\citeauthoryear{Daniel~Huang}{Daniel~Huang}{2017}]%
        {augurv2}
\bibfield{author}{\bibinfo{person}{Greg~Morisett Daniel~Huang,
  Jean-Baptiste~Tristan}.} \bibinfo{year}{2017}\natexlab{}.
\newblock \showarticletitle{Compiling Markov Chain Monte Carlo Algorithms for
  Probabilistic Modeling}. In \bibinfo{booktitle}{\emph{Conference on
  Programming Language Design and Implementation}}.
\newblock


\bibitem[\protect\citeauthoryear{Doerfert, Streit, Hack, and Benaissa}{Doerfert
  et~al\mbox{.}}{2015}]%
        {pollyred}
\bibfield{author}{\bibinfo{person}{Johannes Doerfert}, \bibinfo{person}{Kevin
  Streit}, \bibinfo{person}{Sebastian Hack}, {and} \bibinfo{person}{Zino
  Benaissa}.} \bibinfo{year}{2015}\natexlab{}.
\newblock \showarticletitle{Polly's Polyhedral Scheduling in the Presence of
  Reductions}. In \bibinfo{booktitle}{\emph{International Workshop on
  Polyhedral Compilation Techniques}}.
\newblock


\bibitem[\protect\citeauthoryear{Ehrhardt}{Ehrhardt}{1967}]%
        {ehrhart}
\bibfield{author}{\bibinfo{person}{E. Ehrhardt}.}
  \bibinfo{year}{1967}\natexlab{}.
\newblock \showarticletitle{Sur un probl\`{e}me de G\'{e}om\'{e}trie
  Diophantienne Lin\'{e}aire. II}.
\newblock \bibinfo{journal}{\emph{Journal f\"{u}r die Reine und Angewandte
  Mathematik}}  \bibinfo{volume}{1967} (\bibinfo{year}{1967}),
  \bibinfo{pages}{25--49}.
\newblock
Issue 227.


\bibitem[\protect\citeauthoryear{Feautrier}{Feautrier}{1988}]%
        {feautrierarrayexpansion}
\bibfield{author}{\bibinfo{person}{P. Feautrier}.}
  \bibinfo{year}{1988}\natexlab{}.
\newblock \showarticletitle{Array Expansion}. In
  \bibinfo{booktitle}{\emph{International Conference on Supercomputing}}.
\newblock


\bibitem[\protect\citeauthoryear{Feautrier}{Feautrier}{1992a}]%
        {someeffsol2affschedpt1}
\bibfield{author}{\bibinfo{person}{Paul Feautrier}.}
  \bibinfo{year}{1992}\natexlab{a}.
\newblock \showarticletitle{Some efficient solutions to the affine scheduling
  problem. I. One-dimensional time}.
\newblock \bibinfo{journal}{\emph{International Journal of Parallel
  Programming}} \bibinfo{volume}{21}, \bibinfo{number}{5} (\bibinfo{date}{Oct.}
  \bibinfo{year}{1992}), \bibinfo{pages}{313--347}.
\newblock


\bibitem[\protect\citeauthoryear{Feautrier}{Feautrier}{1992b}]%
        {someeffsol2affschedpt2}
\bibfield{author}{\bibinfo{person}{Paul Feautrier}.}
  \bibinfo{year}{1992}\natexlab{b}.
\newblock \showarticletitle{Some efficient solutions to the affine scheduling
  problem. Part II. Multidimensional time}.
\newblock \bibinfo{journal}{\emph{International Journal of Parallel
  Programming}} \bibinfo{volume}{21}, \bibinfo{number}{6} (\bibinfo{date}{Dec.}
  \bibinfo{year}{1992}), \bibinfo{pages}{389--420}.
\newblock


\bibitem[\protect\citeauthoryear{Fung and Chang}{Fung and Chang}{1989}]%
        {likelihoodweighting}
\bibfield{author}{\bibinfo{person}{Robert~M. Fung} {and}
  \bibinfo{person}{Kuo-Chu Chang}.} \bibinfo{year}{1989}\natexlab{}.
\newblock \showarticletitle{Weighing and Integrating Evidence for Stochastic
  Simulation on Bayesian Networks}. In \bibinfo{booktitle}{\emph{Conference on
  Uncertainty in Artificial Intelligence}}.
\newblock


\bibitem[\protect\citeauthoryear{Gautam and Rajopadhye}{Gautam and
  Rajopadhye}{2006}]%
        {sr}
\bibfield{author}{\bibinfo{person}{Gautam} {and} \bibinfo{person}{S.
  Rajopadhye}.} \bibinfo{year}{2006}\natexlab{}.
\newblock \showarticletitle{Simplifying Reductions}. In
  \bibinfo{booktitle}{\emph{Symposium on Principles of Programming Languages}}.
\newblock


\bibitem[\protect\citeauthoryear{Gelman, Lee, and Guo}{Gelman
  et~al\mbox{.}}{2015}]%
        {stan}
\bibfield{author}{\bibinfo{person}{Andrew Gelman}, \bibinfo{person}{Daniel
  Lee}, {and} \bibinfo{person}{Jiqiang Guo}.} \bibinfo{year}{2015}\natexlab{}.
\newblock \showarticletitle{Stan: A probabilistic programming language for
  Bayesian inference and optimization}.
\newblock \bibinfo{journal}{\emph{Journal of Educational and Behavioral
  Statistics}} \bibinfo{volume}{40}, \bibinfo{number}{5}
  (\bibinfo{year}{2015}), \bibinfo{pages}{530--543}.
\newblock


\bibitem[\protect\citeauthoryear{Geman and Geman}{Geman and Geman}{1984}]%
        {gibbssampling}
\bibfield{author}{\bibinfo{person}{Stuart Geman} {and} \bibinfo{person}{Donald
  Geman}.} \bibinfo{year}{1984}\natexlab{}.
\newblock \showarticletitle{Stochastic Relaxation, Gibbs Distributions, and the
  Bayesian Restoration of images}.
\newblock \bibinfo{journal}{\emph{Transactions on Pattern Analysis and Machine
  Intelligence}}  \bibinfo{volume}{6} (\bibinfo{date}{Nov.}
  \bibinfo{year}{1984}), \bibinfo{pages}{721--741}.
\newblock
Issue 6.


\bibitem[\protect\citeauthoryear{Ginsbach and O'Boyle}{Ginsbach and
  O'Boyle}{2017}]%
        {ginsbach17}
\bibfield{author}{\bibinfo{person}{Philip Ginsbach} {and}
  \bibinfo{person}{Michael F.~P. O'Boyle}.} \bibinfo{year}{2017}\natexlab{}.
\newblock \showarticletitle{Discovery and Exploitation of General Reductions: A
  Constraint Based Approach}. In \bibinfo{booktitle}{\emph{International
  Symposium on Code Generation and Optimization}}.
\newblock


\bibitem[\protect\citeauthoryear{Goodman, Mansinghka, Roy, Bonawitz, and
  Tenenbaum}{Goodman et~al\mbox{.}}{2008}]%
        {church}
\bibfield{author}{\bibinfo{person}{Noah~D. Goodman}, \bibinfo{person}{Vikash~K.
  Mansinghka}, \bibinfo{person}{Daniel~M. Roy}, \bibinfo{person}{Keith
  Bonawitz}, {and} \bibinfo{person}{Joshua~B. Tenenbaum}.}
  \bibinfo{year}{2008}\natexlab{}.
\newblock \showarticletitle{Church: A language for generative models}. In
  \bibinfo{booktitle}{\emph{Conference on Uncertainty in Artificial
  Intelligence}}.
\newblock


\bibitem[\protect\citeauthoryear{Goodman and Stuhlm\"{u}ller}{Goodman and
  Stuhlm\"{u}ller}{2014}]%
        {webppl}
\bibfield{author}{\bibinfo{person}{Noah~D Goodman} {and}
  \bibinfo{person}{Andreas Stuhlm\"{u}ller}.} \bibinfo{year}{2014}\natexlab{}.
\newblock \showarticletitle{{The Design and Implementation of Probabilistic
  Programming Languages}}. \bibinfo{howpublished}{\url{http://dippl.org}}.
\newblock
\newblock
\shownote{Accessed: 2020-10-30.}


\bibitem[\protect\citeauthoryear{Grauer-Gray and Cavazos}{Grauer-Gray and
  Cavazos}{2011}]%
        {lbppldi}
\bibfield{author}{\bibinfo{person}{Scott Grauer-Gray} {and}
  \bibinfo{person}{John Cavazos}.} \bibinfo{year}{2011}\natexlab{}.
\newblock \showarticletitle{Optimizing and Auto-tuning Belief Propagation on
  the GPU}. In \bibinfo{booktitle}{\emph{Workshop on Languages and Compilers
  for Parallel Computing}}.
\newblock


\bibitem[\protect\citeauthoryear{Griffiths and Steyvers}{Griffiths and
  Steyvers}{2004}]%
        {ldagibbs}
\bibfield{author}{\bibinfo{person}{T. Griffiths} {and} \bibinfo{person}{M.
  Steyvers}.} \bibinfo{year}{2004}\natexlab{}.
\newblock \showarticletitle{Finding Scientific Topics}.
\newblock \bibinfo{journal}{\emph{Proceedings of the National Academy of
  Sciences}} \bibinfo{volume}{101}, \bibinfo{number}{suppl. 1}
  (\bibinfo{date}{April} \bibinfo{year}{2004}), \bibinfo{pages}{5228--5235}.
\newblock


\bibitem[\protect\citeauthoryear{Gupta, Daegon, and Rajopadhye}{Gupta
  et~al\mbox{.}}{2007}]%
        {zpolysched}
\bibfield{author}{\bibinfo{person}{Gautam Gupta}, \bibinfo{person}{Kim Daegon},
  {and} \bibinfo{person}{Sanjay Rajopadhye}.} \bibinfo{year}{2007}\natexlab{}.
\newblock \showarticletitle{Scheduling in the Z-Polyhedral Model}.
\newblock \bibinfo{journal}{\emph{International Parallel and Distributed
  Processing Symposium}}.
\newblock


\bibitem[\protect\citeauthoryear{Gupta, Rajopadhye, and Quinton}{Gupta
  et~al\mbox{.}}{2002}]%
        {schedulereductionsrealistic}
\bibfield{author}{\bibinfo{person}{Gautam Gupta}, \bibinfo{person}{Sanjay
  Rajopadhye}, {and} \bibinfo{person}{Patrice Quinton}.}
  \bibinfo{year}{2002}\natexlab{}.
\newblock \showarticletitle{Scheduling Reductions on Realistic Machines}. In
  \bibinfo{booktitle}{\emph{Symposium on Parallel Algorithms and
  Architectures}}.
\newblock


\bibitem[\protect\citeauthoryear{Hastings}{Hastings}{1970}]%
        {mhb}
\bibfield{author}{\bibinfo{person}{W.~K. Hastings}.}
  \bibinfo{year}{1970}\natexlab{}.
\newblock \showarticletitle{Monte Carlo Sampling Methods Using Markov Chains
  and Their Applications}.
\newblock \bibinfo{journal}{\emph{Biometrika}} \bibinfo{volume}{57},
  \bibinfo{number}{1} (\bibinfo{date}{April} \bibinfo{year}{1970}),
  \bibinfo{pages}{97--109}.
\newblock


\bibitem[\protect\citeauthoryear{Holmes, Harris, and Quince}{Holmes
  et~al\mbox{.}}{2012}]%
        {dmm}
\bibfield{author}{\bibinfo{person}{Ian Holmes}, \bibinfo{person}{Keith Harris},
  {and} \bibinfo{person}{Christopher Quince}.} \bibinfo{year}{2012}\natexlab{}.
\newblock \showarticletitle{Dirichlet Multinomial Mixtures: Generative Models
  for Microbial Metagenomics}.
\newblock \bibinfo{journal}{\emph{PLOS ONE}}  \bibinfo{volume}{7}
  (\bibinfo{date}{2} \bibinfo{year}{2012}).
\newblock
Issue 2.


\bibitem[\protect\citeauthoryear{Iooss, Alias, and Rajopadhye}{Iooss
  et~al\mbox{.}}{2014}]%
        {programequivalencereduction}
\bibfield{author}{\bibinfo{person}{Guillaume Iooss},
  \bibinfo{person}{Christophe Alias}, {and} \bibinfo{person}{Sanjay
  Rajopadhye}.} \bibinfo{year}{2014}\natexlab{}.
\newblock \showarticletitle{On Program Equivalence with Reductions}. In
  \bibinfo{booktitle}{\emph{International Static Analysis Symposium}}.
\newblock


\bibitem[\protect\citeauthoryear{{Jian Sun}, {Nan-Ning Zheng}, and {Heung-Yeung
  Shum}}{{Jian Sun} et~al\mbox{.}}{2003}]%
        {lbpstereo}
\bibfield{author}{\bibinfo{person}{{Jian Sun}}, \bibinfo{person}{{Nan-Ning
  Zheng}}, {and} \bibinfo{person}{{Heung-Yeung Shum}}.}
  \bibinfo{year}{2003}\natexlab{}.
\newblock \showarticletitle{Stereo matching using belief propagation}.
\newblock \bibinfo{journal}{\emph{Transactions on Pattern Analysis and Machine
  Intelligence}} \bibinfo{volume}{25}, \bibinfo{number}{7}
  (\bibinfo{date}{July} \bibinfo{year}{2003}), \bibinfo{pages}{787--800}.
\newblock


\bibitem[\protect\citeauthoryear{Kikuchi}{Kikuchi}{1951}]%
        {isingphysics}
\bibfield{author}{\bibinfo{person}{Ryoichi Kikuchi}.}
  \bibinfo{year}{1951}\natexlab{}.
\newblock \showarticletitle{A Theory of Cooperative Phenomena}.
\newblock \bibinfo{journal}{\emph{Physical Review}}  \bibinfo{volume}{81}
  (\bibinfo{date}{March} \bibinfo{year}{1951}), \bibinfo{pages}{988--1003}.
\newblock
Issue 6.


\bibitem[\protect\citeauthoryear{Kiselyov}{Kiselyov}{2016}]%
        {hakaru10}
\bibfield{author}{\bibinfo{person}{Oleg Kiselyov}.}
  \bibinfo{year}{2016}\natexlab{}.
\newblock \showarticletitle{Probabilistic Programming Language and its
  Incremental Evaluation}. In \bibinfo{booktitle}{\emph{Asian Symposium on
  Programming Languages and Systems}}.
\newblock


\bibitem[\protect\citeauthoryear{Liu}{Liu}{1994}]%
        {gibbscollapsed}
\bibfield{author}{\bibinfo{person}{Jun~S. Liu}.}
  \bibinfo{year}{1994}\natexlab{}.
\newblock \showarticletitle{The Collapsed Gibbs Sampler in Bayesian
  Computations with Applications to a Gene Regulation Problem}.
\newblock \bibinfo{journal}{\emph{J. Amer. Statist. Assoc.}}
  \bibinfo{volume}{89}, \bibinfo{number}{427} (\bibinfo{date}{Sept.}
  \bibinfo{year}{1994}), \bibinfo{pages}{958--966}.
\newblock


\bibitem[\protect\citeauthoryear{Liu, Stoller, Li, and Rothamel}{Liu
  et~al\mbox{.}}{2005}]%
        {optaggarray}
\bibfield{author}{\bibinfo{person}{Yanhong~A. Liu}, \bibinfo{person}{Scott~D.
  Stoller}, \bibinfo{person}{Ning Li}, {and} \bibinfo{person}{Tom Rothamel}.}
  \bibinfo{year}{2005}\natexlab{}.
\newblock \showarticletitle{Optimizing Aggregate Array Computations in Loops}.
\newblock \bibinfo{journal}{\emph{ACM Transactions on Programming Languages and
  Systems}} \bibinfo{volume}{27}, \bibinfo{number}{1} (\bibinfo{date}{Jan.}
  \bibinfo{year}{2005}), \bibinfo{pages}{91--125}.
\newblock


\bibitem[\protect\citeauthoryear{Mansingkha, Schaechtle, Handa, Radul, Chen,
  and Rinard}{Mansingkha et~al\mbox{.}}{2018}]%
        {venture}
\bibfield{author}{\bibinfo{person}{Vikash Mansingkha}, \bibinfo{person}{Ulrich
  Schaechtle}, \bibinfo{person}{Shivam Handa}, \bibinfo{person}{Alexey Radul},
  \bibinfo{person}{Yutian Chen}, {and} \bibinfo{person}{Martin Rinard}.}
  \bibinfo{year}{2018}\natexlab{}.
\newblock \showarticletitle{Probabilistic Programming with Programmable
  Inference}. In \bibinfo{booktitle}{\emph{Conference on Programming Language
  Design and Implementation}}.
\newblock


\bibitem[\protect\citeauthoryear{{Metropolis}, {Rosenbluth}, {Rosenbluth},
  {Teller}, and {Teller}}{{Metropolis} et~al\mbox{.}}{1953}]%
        {mha}
\bibfield{author}{\bibinfo{person}{N. {Metropolis}}, \bibinfo{person}{A.~W.
  {Rosenbluth}}, \bibinfo{person}{M.~N. {Rosenbluth}}, \bibinfo{person}{A.~H.
  {Teller}}, {and} \bibinfo{person}{E. {Teller}}.}
  \bibinfo{year}{1953}\natexlab{}.
\newblock \showarticletitle{{Equation of State Calculations by Fast Computing
  Machines}}.
\newblock \bibinfo{journal}{\emph{Journal of Chemical Physics}}
  \bibinfo{volume}{21}, \bibinfo{number}{6} (\bibinfo{year}{1953}),
  \bibinfo{pages}{1087--1092}.
\newblock


\bibitem[\protect\citeauthoryear{Murphy}{Murphy}{2012}]%
        {gmm}
\bibfield{author}{\bibinfo{person}{Kevin~P. Murphy}.}
  \bibinfo{year}{2012}\natexlab{}.
\newblock \bibinfo{booktitle}{\emph{Machine Learning: A Probabilistic
  Perspective}}.
\newblock \bibinfo{publisher}{MIT Press}, \bibinfo{address}{Cambridge,
  Massachusets}.
\newblock


\bibitem[\protect\citeauthoryear{Narayanan, Carette, Romano, Shan, and
  Zinkov}{Narayanan et~al\mbox{.}}{2016}]%
        {hakaru}
\bibfield{author}{\bibinfo{person}{Praveen Narayanan}, \bibinfo{person}{Jacques
  Carette}, \bibinfo{person}{Wren Romano}, \bibinfo{person}{Chung{-}chieh
  Shan}, {and} \bibinfo{person}{Robert Zinkov}.}
  \bibinfo{year}{2016}\natexlab{}.
\newblock \showarticletitle{Probabilistic Inference by Program Transformation
  in Hakaru (System Description)}. In \bibinfo{booktitle}{\emph{International
  Symposium on Functional and Logic Programming}}.
\newblock


\bibitem[\protect\citeauthoryear{Nemhauser and Wolsey}{Nemhauser and
  Wolsey}{1988}]%
        {ilpbigm}
\bibfield{author}{\bibinfo{person}{George~L. Nemhauser} {and}
  \bibinfo{person}{Laurence~A. Wolsey}.} \bibinfo{year}{1988}\natexlab{}.
\newblock \showarticletitle{Integer and Combinatorial Optimization}.
\newblock


\bibitem[\protect\citeauthoryear{Newman}{Newman}{2008}]%
        {kos}
\bibfield{author}{\bibinfo{person}{David Newman}.}
  \bibinfo{year}{2008}\natexlab{}.
\newblock \showarticletitle{Bag of Words Dataset}. In
  \bibinfo{booktitle}{\emph{UCI Machine Learning Respository}}.
\newblock


\bibitem[\protect\citeauthoryear{Nori, Ozair, Rajamani, and Vijaykeerthy}{Nori
  et~al\mbox{.}}{2015}]%
        {r2}
\bibfield{author}{\bibinfo{person}{Aditya~V. Nori}, \bibinfo{person}{Sherjil
  Ozair}, \bibinfo{person}{Sriram~K. Rajamani}, {and} \bibinfo{person}{Deepak
  Vijaykeerthy}.} \bibinfo{year}{2015}\natexlab{}.
\newblock \showarticletitle{Efficient Synthesis of Probabilistic Programs}. In
  \bibinfo{booktitle}{\emph{Conference on Programming Language Design and
  Implementation}}.
\newblock


\bibitem[\protect\citeauthoryear{Padua}{Padua}{2011}]%
        {omega}
\bibfield{editor}{\bibinfo{person}{David Padua}} (Ed.).
  \bibinfo{year}{2011}\natexlab{}.
\newblock \bibinfo{booktitle}{\emph{Omega Calculator}}.
\newblock \bibinfo{publisher}{Springer US}, \bibinfo{address}{Boston, MA},
  \bibinfo{pages}{1355--1355}.
\newblock
\showISBNx{978-0-387-09766-4}
\urldef\tempurl%
\url{https://doi.org/10.1007/978-0-387-09766-4\_2303}
\showDOI{\tempurl}


\bibitem[\protect\citeauthoryear{Plummer}{Plummer}{2015}]%
        {jags}
\bibfield{author}{\bibinfo{person}{Martyn Plummer}.}
  \bibinfo{year}{2015}\natexlab{}.
\newblock \bibinfo{booktitle}{\emph{JAGS Version 4.0.0 user manual}}.
\newblock \bibinfo{publisher}{Addison-Wesley}, \bibinfo{address}{Reading,
  Massachusetts}.
\newblock


\bibitem[\protect\citeauthoryear{Pouchet, Bastoul, Cohen, and Cavazos}{Pouchet
  et~al\mbox{.}}{2008}]%
        {pouchet.08.pldi}
\bibfield{author}{\bibinfo{person}{Louis-No{\"e}l Pouchet},
  \bibinfo{person}{C{\'e}dric Bastoul}, \bibinfo{person}{Albert Cohen}, {and}
  \bibinfo{person}{John Cavazos}.} \bibinfo{year}{2008}\natexlab{}.
\newblock \showarticletitle{Iterative optimization in the polyhedral model:
  Part {II}, multidimensional time}. In \bibinfo{booktitle}{\emph{Conference on
  Programming Language Design and Implementation}}.
\newblock


\bibitem[\protect\citeauthoryear{Pouchet, Bastoul, Cohen, and
  Vasilache}{Pouchet et~al\mbox{.}}{2007}]%
        {pouchet.07.cgo}
\bibfield{author}{\bibinfo{person}{Louis-No{\"e}l Pouchet},
  \bibinfo{person}{C{\'e}dric Bastoul}, \bibinfo{person}{Albert Cohen}, {and}
  \bibinfo{person}{Nicolas Vasilache}.} \bibinfo{year}{2007}\natexlab{}.
\newblock \showarticletitle{Iterative optimization in the polyhedral model:
  Part {I}, one-dimensional time}. In \bibinfo{booktitle}{\emph{International
  Symposium on Code Generation and Optimization}}.
\newblock


\bibitem[\protect\citeauthoryear{Pouchet, Bondhugula, Bastoul, Cohen,
  Ramanujam, Sadayappan, and Vasilache}{Pouchet et~al\mbox{.}}{2011}]%
        {ilpschedulemultidim}
\bibfield{author}{\bibinfo{person}{Louis-No\"{e}l Pouchet},
  \bibinfo{person}{Uday Bondhugula}, \bibinfo{person}{C\'{e}dric Bastoul},
  \bibinfo{person}{Albert Cohen}, \bibinfo{person}{J. Ramanujam},
  \bibinfo{person}{P. Sadayappan}, {and} \bibinfo{person}{Nicolas Vasilache}.}
  \bibinfo{year}{2011}\natexlab{}.
\newblock \showarticletitle{Loop Transformations: Convexity, Pruning and
  Optimization}. In \bibinfo{booktitle}{\emph{Symposium on Principles of
  Programming Languages}}.
\newblock


\bibitem[\protect\citeauthoryear{{Rauchwerger} and {Padua}}{{Rauchwerger} and
  {Padua}}{1999}]%
        {rauchwerger99}
\bibfield{author}{\bibinfo{person}{L. {Rauchwerger}} {and}
  \bibinfo{person}{D.~A. {Padua}}.} \bibinfo{year}{1999}\natexlab{}.
\newblock \showarticletitle{The LRPD test: speculative run-time parallelization
  of loops with privatization and reduction parallelization}.
\newblock \bibinfo{journal}{\emph{Transactions on Parallel and Distributed
  Systems}} \bibinfo{volume}{10}, \bibinfo{number}{2} (\bibinfo{date}{Feb.}
  \bibinfo{year}{1999}), \bibinfo{pages}{160--180}.
\newblock


\bibitem[\protect\citeauthoryear{Reddy, Kruse, and Cohen}{Reddy
  et~al\mbox{.}}{2016}]%
        {reddrawing}
\bibfield{author}{\bibinfo{person}{C. Reddy}, \bibinfo{person}{M. Kruse}, {and}
  \bibinfo{person}{A. Cohen}.} \bibinfo{year}{2016}\natexlab{}.
\newblock \showarticletitle{Reduction drawing: Language constructs and
  polyhedral compilation for reductions on GPUs}. In
  \bibinfo{booktitle}{\emph{International Conference on Parallel Architecture
  and Compilation Techniques}}.
\newblock


\bibitem[\protect\citeauthoryear{Redon and Feautrier}{Redon and
  Feautrier}{1994}]%
        {schedulereductions94}
\bibfield{author}{\bibinfo{person}{Xavier Redon} {and} \bibinfo{person}{Paul
  Feautrier}.} \bibinfo{year}{1994}\natexlab{}.
\newblock \showarticletitle{Scheduling Reductions}. In
  \bibinfo{booktitle}{\emph{International Conference on Supercomputing}}.
\newblock


\bibitem[\protect\citeauthoryear{Resnik and Hardisty}{Resnik and
  Hardisty}{2010}]%
        {dmmcollapsed}
\bibfield{author}{\bibinfo{person}{Philip Resnik} {and} \bibinfo{person}{Eric
  Hardisty}.} \bibinfo{year}{2010}\natexlab{}.
\newblock \bibinfo{booktitle}{\emph{Gibbs Sampling for the Uninitiated}}.
\newblock \bibinfo{type}{{T}echnical {R}eport}.
  \bibinfo{institution}{University of Maryland Institute of Advanced Computer
  Studies}.
\newblock


\bibitem[\protect\citeauthoryear{Ritchie, Stuhlm\"{u}ller, and Goodman}{Ritchie
  et~al\mbox{.}}{2016}]%
        {c3}
\bibfield{author}{\bibinfo{person}{Daniel Ritchie}, \bibinfo{person}{Andreas
  Stuhlm\"{u}ller}, {and} \bibinfo{person}{Noah Goodman}.}
  \bibinfo{year}{2016}\natexlab{}.
\newblock \showarticletitle{C3: Lightweight Incrementalized MCMC for
  Probabilistic Programs using Continuations and Callsite Caching}. In
  \bibinfo{booktitle}{\emph{International Conference on Artificial Intelligence
  and Statistics}}.
\newblock


\bibitem[\protect\citeauthoryear{Rubiano}{Rubiano}{2017}]%
        {determinecomplexity1}
\bibfield{author}{\bibinfo{person}{Thomas Rubiano}.}
  \bibinfo{year}{2017}\natexlab{}.
\newblock \emph{\bibinfo{title}{{Implicit Computational Complexity and
  Compilers}}}.
\newblock \bibinfo{thesistype}{Ph.D. Dissertation}.
\newblock


\bibitem[\protect\citeauthoryear{Saouter and Quinton}{Saouter and
  Quinton}{1993}]%
        {sarecomputability}
\bibfield{author}{\bibinfo{person}{Yannick Saouter} {and}
  \bibinfo{person}{Patrice Quinton}.} \bibinfo{year}{1993}\natexlab{}.
\newblock \showarticletitle{Computability of Recurrence Equations}.
\newblock \bibinfo{journal}{\emph{Theoretical Computer Science}}
  \bibinfo{volume}{116}, \bibinfo{number}{2} (\bibinfo{date}{Aug.}
  \bibinfo{year}{1993}), \bibinfo{pages}{317--337}.
\newblock


\bibitem[\protect\citeauthoryear{Schrijver}{Schrijver}{1986}]%
        {theoryofilp}
\bibfield{author}{\bibinfo{person}{Alexander Schrijver}.}
  \bibinfo{year}{1986}\natexlab{}.
\newblock \bibinfo{booktitle}{\emph{Theory of Linear and Integer Programming}}.
\newblock \bibinfo{publisher}{John Wiley \& Sons, Inc.}, \bibinfo{address}{New
  York, NY, USA}.
\newblock


\bibitem[\protect\citeauthoryear{{Susto}, {Schirru}, {Pampuri}, {McLoone}, and
  {Beghi}}{{Susto} et~al\mbox{.}}{2015}]%
        {coxmechanics}
\bibfield{author}{\bibinfo{person}{G.~A. {Susto}}, \bibinfo{person}{A.
  {Schirru}}, \bibinfo{person}{S. {Pampuri}}, \bibinfo{person}{S. {McLoone}},
  {and} \bibinfo{person}{A. {Beghi}}.} \bibinfo{year}{2015}\natexlab{}.
\newblock \showarticletitle{Machine Learning for Predictive Maintenance: A
  Multiple Classifier Approach}.
\newblock \bibinfo{journal}{\emph{Transactions on Industrial Informatics}}
  \bibinfo{volume}{11}, \bibinfo{number}{3} (\bibinfo{date}{June}
  \bibinfo{year}{2015}), \bibinfo{pages}{812--820}.
\newblock


\bibitem[\protect\citeauthoryear{Therneau}{Therneau}{2013}]%
        {coxbook}
\bibfield{author}{\bibinfo{person}{Patricia~M Therneau, Terry~M.;Grambsch}.}
  \bibinfo{year}{2013}\natexlab{}.
\newblock \bibinfo{booktitle}{\emph{Modeling Survival Data: Extending the Cox
  Model}}.
\newblock \bibinfo{publisher}{Springer}, \bibinfo{address}{New York}.
\newblock


\bibitem[\protect\citeauthoryear{Tran, Hoffman, Saurous, Brevdo, Murphy, and
  Blei}{Tran et~al\mbox{.}}{2017}]%
        {edward}
\bibfield{author}{\bibinfo{person}{Dustin Tran}, \bibinfo{person}{Matthew~D
  Hoffman}, \bibinfo{person}{Rif~A Saurous}, \bibinfo{person}{Eugene Brevdo},
  \bibinfo{person}{Kevin Murphy}, {and} \bibinfo{person}{David~M Blei}.}
  \bibinfo{year}{2017}\natexlab{}.
\newblock \showarticletitle{Deep probabilistic programming}. In
  \bibinfo{booktitle}{\emph{International Conference on Learning
  Representations}}.
\newblock


\bibitem[\protect\citeauthoryear{Turnbaugh, Hamady, Yatsunenko, Cantarel,
  Duncan, Ley, Sogin, Jones, A~Roe, Affourtit, Egholm, Henrissat, C~Heath,
  Knight, and I~Gordon}{Turnbaugh et~al\mbox{.}}{2008}]%
        {twins}
\bibfield{author}{\bibinfo{person}{Peter Turnbaugh}, \bibinfo{person}{Micah
  Hamady}, \bibinfo{person}{Tanya Yatsunenko}, \bibinfo{person}{Brandi
  Cantarel}, \bibinfo{person}{Alexis Duncan}, \bibinfo{person}{Ruth Ley},
  \bibinfo{person}{Mitchell Sogin}, \bibinfo{person}{Joe Jones},
  \bibinfo{person}{Bruce A~Roe}, \bibinfo{person}{Jason Affourtit},
  \bibinfo{person}{Michael Egholm}, \bibinfo{person}{Bernard Henrissat},
  \bibinfo{person}{Andrew C~Heath}, \bibinfo{person}{Rob Knight}, {and}
  \bibinfo{person}{Jeffrey I~Gordon}.} \bibinfo{year}{2008}\natexlab{}.
\newblock \showarticletitle{A core gut microbiome in obese and lean twins}.
\newblock \bibinfo{journal}{\emph{Nature}}  \bibinfo{volume}{457}
  (\bibinfo{date}{12} \bibinfo{year}{2008}), \bibinfo{pages}{480--4}.
\newblock


\bibitem[\protect\citeauthoryear{Verdoolaege}{Verdoolaege}{2010}]%
        {isl}
\bibfield{author}{\bibinfo{person}{Sven Verdoolaege}.}
  \bibinfo{year}{2010}\natexlab{}.
\newblock \showarticletitle{isl: An Integer Set Library for the Polyhedral
  Model}. In \bibinfo{booktitle}{\emph{International Congress on Mathematical
  Software}}.
\newblock


\bibitem[\protect\citeauthoryear{Verdoolaege}{Verdoolaege}{2016}]%
        {presburgerformulaandpolyhedralcompilation}
\bibfield{author}{\bibinfo{person}{Sven Verdoolaege}.}
  \bibinfo{year}{2016}\natexlab{}.
\newblock \bibinfo{title}{Presburger Formulas and Polyhedral Compilation}.
\newblock
\newblock
\urldef\tempurl%
\url{https://doi.org/10.13140/RG.2.1.1174.6323}
\showDOI{\tempurl}


\bibitem[\protect\citeauthoryear{Verdoolaege, Nikolov, and
  Stefanov}{Verdoolaege et~al\mbox{.}}{2013}]%
        {ondemandarraydatafolow}
\bibfield{author}{\bibinfo{person}{Sven Verdoolaege}, \bibinfo{person}{Hristo
  Nikolov}, {and} \bibinfo{person}{Todor Stefanov}.}
  \bibinfo{year}{2013}\natexlab{}.
\newblock \showarticletitle{On Demand Parametric Array Dataflow Analysis}. In
  \bibinfo{booktitle}{\emph{International Workshop on Polyhedral Compilation
  Techniques}}.
\newblock


\bibitem[\protect\citeauthoryear{Verdoolaege, Seghir, Beyls, Loechner, and
  Bruynooghe}{Verdoolaege et~al\mbox{.}}{2007}]%
        {barvinok}
\bibfield{author}{\bibinfo{person}{Sven Verdoolaege}, \bibinfo{person}{Rachid
  Seghir}, \bibinfo{person}{Kristof Beyls}, \bibinfo{person}{Vincent Loechner},
  {and} \bibinfo{person}{Maurice Bruynooghe}.} \bibinfo{year}{2007}\natexlab{}.
\newblock \showarticletitle{Counting Integer Points in Parametric Polytopes
  Using Barvinok's Rational Functions}.
\newblock \bibinfo{journal}{\emph{Algorithmica}} \bibinfo{volume}{48},
  \bibinfo{number}{1} (\bibinfo{date}{May} \bibinfo{year}{2007}),
  \bibinfo{pages}{37--66}.
\newblock


\bibitem[\protect\citeauthoryear{Walia, Narayanan, Carette, Tobin-Hochstadt,
  and Shan}{Walia et~al\mbox{.}}{2019}]%
        {histogram}
\bibfield{author}{\bibinfo{person}{Rajan Walia}, \bibinfo{person}{Praveen
  Narayanan}, \bibinfo{person}{Jacques Carette}, \bibinfo{person}{Sam
  Tobin-Hochstadt}, {and} \bibinfo{person}{Chung{-}chieh Shan}.}
  \bibinfo{year}{2019}\natexlab{}.
\newblock \showarticletitle{From High-level Inference Algorithms to Efficient
  Code}. In \bibinfo{booktitle}{\emph{International Conference on Functional
  Programming}}.
\newblock


\bibitem[\protect\citeauthoryear{White, Reid, Harris, Harries, and Stone}{White
  et~al\mbox{.}}{2016}]%
        {coxmedical2}
\bibfield{author}{\bibinfo{person}{Nicola White}, \bibinfo{person}{Fiona Reid},
  \bibinfo{person}{Adam Harris}, \bibinfo{person}{Priscilla Harries}, {and}
  \bibinfo{person}{Patrick Stone}.} \bibinfo{year}{2016}\natexlab{}.
\newblock \showarticletitle{A Systematic Review of Predictions of Survival in
  Palliative Care: How Accurate Are Clinicians and Who Are the Experts?}
\newblock \bibinfo{journal}{\emph{PLOS ONE}} \bibinfo{volume}{11},
  \bibinfo{number}{8} (\bibinfo{date}{08} \bibinfo{year}{2016}),
  \bibinfo{pages}{1--20}.
\newblock


\bibitem[\protect\citeauthoryear{Wu, Li, Russell, and Bodik}{Wu
  et~al\mbox{.}}{2016}]%
        {swift}
\bibfield{author}{\bibinfo{person}{Yi Wu}, \bibinfo{person}{Lei Li},
  \bibinfo{person}{Stuart Russell}, {and} \bibinfo{person}{Rastislav Bodik}.}
  \bibinfo{year}{2016}\natexlab{}.
\newblock In \bibinfo{booktitle}{\emph{International Joint Conferences on
  Artificial Intelligence}}.
\newblock


\bibitem[\protect\citeauthoryear{Yang, Hanrahan, and Goodman}{Yang
  et~al\mbox{.}}{2014}]%
        {yang14}
\bibfield{author}{\bibinfo{person}{Lingfeng Yang}, \bibinfo{person}{Patrick
  Hanrahan}, {and} \bibinfo{person}{Noah Goodman}.}
  \bibinfo{year}{2014}\natexlab{}.
\newblock \showarticletitle{{Generating Efficient MCMC Kernels from
  Probabilistic Programs}}. In \bibinfo{booktitle}{\emph{International
  Conference on Artificial Intelligence and Statistics}}.
\newblock


\bibitem[\protect\citeauthoryear{Yuki, Gupta, Kim, Pathan, and Rajopadhye}{Yuki
  et~al\mbox{.}}{2013}]%
        {alphaz}
\bibfield{author}{\bibinfo{person}{Tomofumi Yuki}, \bibinfo{person}{Gautam
  Gupta}, \bibinfo{person}{DaeGon Kim}, \bibinfo{person}{Tanveer Pathan}, {and}
  \bibinfo{person}{Sanjay Rajopadhye}.} \bibinfo{year}{2013}\natexlab{}.
\newblock \showarticletitle{AlphaZ: A System for Design Space Exploration in
  the Polyhedral Model}. In \bibinfo{booktitle}{\emph{Workshop on Languages and
  Compilers for Parallel Computing}}.
\newblock


\bibitem[\protect\citeauthoryear{Zhang and Xue}{Zhang and Xue}{2019}]%
        {zhang19}
\bibfield{author}{\bibinfo{person}{Jieyuan Zhang} {and}
  \bibinfo{person}{Jingling Xue}.} \bibinfo{year}{2019}\natexlab{}.
\newblock \showarticletitle{Incremental Precision-Preserving Symbolic Inference
  for Probabilistic Programs}. In \bibinfo{booktitle}{\emph{Conference on
  Programming Language Design and Implementation}}.
\newblock


\end{thebibliography}

\end{document}

\makeatletter\@input{xx.tex}\makeatother

\title{Simplifying Dependent Reductions in the Polyhedral Model (Appendix)}     
\author{Cambridge Yang}
\affiliation{
  \institution{MIT CSAIL}            %
  \country{USA}                    %
}
\email{camyang@csail.mit.edu}          %

\author{Eric Atkinson}
\affiliation{
  \institution{MIT CSAIL}            %
  \country{USA}                    %
}
\email{eatkinson@csail.mit.edu}          %

\author{Michael Carbin}
\affiliation{
  \institution{MIT CSAIL}            %
  \country{USA}                    %
}
\email{mcarbin@csail.mit.edu}          %

\appendix

\section{Extra Listings and Equations}
\label{sec:extra_listings}

\begin{clisting}[caption={Alternative optimized PS (right-to-left)}, 
label=lst:optprefixsumr2l]
for(i = 0; i < N; i++)
 B[N - 1] += A[i]
for(i = N-2; i >= 0; i--)
 B[i] = B[i+1] - A[i]
\end{clisting}

\begin{figure}[h]
\begin{subequations}
\allowdisplaybreaks
\label{eq:ms-scan-faces-linealty}
\begin{align}
\mathcal{L}(f_1) &= \polyhedral[N]{[i, j] : 1 = 1} \\
\mathcal{L}(f_2) &= \polyhedral[N]{[i, j] : j = 0} \\
\mathcal{L}(f_3) &= \polyhedral[N]{[i, j] : j = i} \\
\mathcal{L}(f_4) &= \polyhedral[N]{[i, j] : i = 0 \land j = 0} \\
\mathcal{L}(f_5) &= \polyhedral[N]{[i, j] : i = 0 \land j = 0} \\
\mathcal{L}(f_6) &= \polyhedral[N]{[i, j] : i = 0 \land j = 0} \\
\mathcal{L}(f_7) &= \polyhedral[N]{[i, j] : i = 0 \land j = 0}\\
\mathcal{L}(f_8) &= \polyhedral[N]{[i, j] : i = 0 \land j = 0}
\end{align}
\end{subequations}
\caption{$\mathcal{L}(.)$ for the example in \Cref{sec:reuse_constraint_bilp}}
\end{figure}

\section{Simplifying Reduction}
\label{sec:simplifying-reduction}

{
	
\newcommand{\Paddonly}{\mathcal{P}_{\text{add}}^{\mathbf{u}} - \mathcal{P}_{\text{int}}^{\mathbf{u}}}

\newcommand{\Preuseonly}{\mathcal{P}_{\text{int}}^{\mathbf{u}} - 
	(\mathcal{P}_{\text{add}}^{\mathbf{u}} \cup \mathcal{P}_{\text{sub}}^{\mathbf{u}})}

\newcommand{\Paddreuse}{\mathcal{P}_{\text{add}}^{\mathbf{u}} \cap 
	(\mathcal{P}_{\text{int}}^{\mathbf{u}} - \mathcal{P}_{\text{sub}}^{\mathbf{u}})}

\newcommand{\Preusesub}{\mathcal{P}_{\text{sub}}^{\mathbf{u}} \cap 
	(\mathcal{P}_{\text{int}}^{\mathbf{u}} - \mathcal{P}_{\text{add}}^{\mathbf{u}})}

\newcommand{\Paddreusesub}{\mathcal{P}_{\text{sub}}^{\mathbf{u}} \cap 
	\mathcal{P}_{\text{int}}^{\mathbf{u}} \cap \mathcal{P}_{\text{add}}^{\mathbf{u}} }

\begin{figure}[ht]
$
\begin{array}{lll}
\texttt{l-add-only: } & \textsf{LHS}[\mathbf{u}] = \textsf{ADD}[\mathbf{u}] & : \Paddonly \\
\texttt{l-reuse-only: } & \textsf{LHS}[\mathbf{u}] = \textsf{LHS}[T_r^{\mathbf{u}}(\mathbf{u})] & : 
\Preuseonly \\
\texttt{l-add-reuse: } & \textsf{LHS}[\mathbf{u}] = \textsf{ADD}[\mathbf{u}] \oplus 
\textsf{LHS}[T_r^{\mathbf{u}}(\mathbf{u})] & : \Paddreuse \\
\texttt{l-reuse-sub: } & \textsf{LHS}[\mathbf{u}] = \textsf{LHS}[T_r^{\mathbf{u}}(\mathbf{u})] \ominus 
\textsf{SUB}[\mathbf{u}] & : \Preusesub \\
\texttt{l-add-reuse-sub: } & \textsf{LHS}[\mathbf{u}] =  \textsf{ADD}[\mathbf{u}] \oplus 
\textsf{LHS}[T_r^{\mathbf{u}}(\mathbf{u})] \ominus \textsf{SUB}[\mathbf{u}] & : \Paddreusesub \\
\texttt{ladd: } & \textsf{ADD}[\mathbf{u}] \mathrel{\oplus}= \textit{expr} & : 
\mathcal{P}_{\text{add}} 
\\
\texttt{lsub: } & \textsf{SUB}[\mathbf{u}] \mathrel{\oplus}= \subst{\textit{expr}}{T_r(\freevars(\textit{expr}))} & : 
(\mathcal{P}_{\text{int}}^{\mathbf{u}})^{\mathbf{s}} \cap \mathcal{P}_{\text{sub}}
\end{array}
$
\caption{Simplifying Reduction in the Polyhedral Model}
\vspace{-.1cm}
\label{fig:srtransform}
\end{figure}

\subsection{Simplifying Reduction in Polyhedral Model}
\label{sec:sr}
\label{sec:srtransformation}

Consider an IR statement:
\[
\texttt{label: } \textsf{LHS}[\mathbf{u}] \mathrel{\oplus}=  \textit{expr} 
\; : \; \mathcal{P}
\]

Simplifying reduction (SR) transforms this statement into an equivalent form as
in \Cref{fig:srtransform}.  The transformation takes in one parameter, a
nonzero constant vector $\vec{r}$, representing the direction of reuse, which we will explain shortly.

We first define some notations:
\begin{itemize}
\item we use $p^\mathbf{a}$ to denote projecting $p$ onto space $\mathbf{a}$; the superscript acts 
effectively as an projection function; $p$ can either be a point, an affine transformation or a 
polyhedral set of points.
\item $T_r(x)$ is an affine translation transformation (under homogeneous coordinates). That is, if $x$ 
is a vector $\vec{x}$ representing a point, $T_r$ shifts $\vec{x}$ to $\vec{x} + \vec{r}$. If $x$ is a 
polyhedron 
$\mathcal{P}$,  $T_r$ shifts all points in $\mathcal{P}$ by $+\vec{r}$. 
\end{itemize}

Then, let $\mathcal{P}' = T_{-r}(\mathcal{P})$, i.e. $\mathcal{P}'$ is $\mathcal{P}$
shifted by $-\vec{r}$, we define the following symbols in \Cref{fig:srtransform}:
\[ 
\mathcal{P}_{\text{add}} = \mathcal{P} - \mathcal{P}' \quad\;\;
\mathcal{P}_{\text{sub}} = \mathcal{P}' - \mathcal{P} \quad\;\;
\mathcal{P}_{\text{int}} = \mathcal{P} \cap \mathcal{P}' \quad\;\;
\]

\paragraph{Explanation}
The core intuition behind ST is to realize {\em reuse} of the right hand side 
$\textit{expr}$.
Specifically, we require a choice of $\vec{r}$ so that it presents {\em sharing} for the 
RHS expression, that is: 
$$\llbracket \subst{\textit{expr}}{T_r(\freevars(\textit{expr}))} =
\textit{expr} \rrbracket$$. 
In other words, the value of $\textit{expr}$
is the same for any point $\mathbf{v}$ and its shifted counterpart
$T_r(\mathbf{v})$. This way, we can avoid evaluation of $\textit{expr}$ by simply 
copying from $\subst{\textit{expr}}{T_r(\allowbreak\freevars(\textit{expr}))}$, 
whenever possible. The first 
five statements 
\texttt{l-add-only} through \texttt{l-add-reuse-sub} computes $\textsf{LHS}$ this way and reuse 
$\subst{\textit{expr}}{T_r(\allowbreak\freevars(\textit{expr}))}$ along $\vec{r}$. The 
domains of the five 
statements prescribe the set of points according to each statement's semantics. 

To make this concrete, first notice the following:
\begin{itemize}
\item $\mathcal{P}_{\text{add}}^\mathbf{u}$ is the set of indices that receives
$expr$'s values evaluated in $\mathcal{P}_{\text{add}}$
\item $\mathcal{P}_{\text{sub}}^\mathbf{u}$ is the set of indices that receives $expr$'s
values evaluated in $\mathcal{P}_{\text{sub}}$
\item $\mathcal{P}_{\text{int}}^\mathbf{u}$ is
the set of indices that receives $expr$'s values evaluated in
$\mathcal{P}_{\text{int}}$. Receiving value from the intersection means that it is possible to reuse from 
the index point shifted by $\vec{r}$.
\end{itemize}

We then explain each of the first five statements in turn:

\paragraph{Reuse only} Consider the domain of \texttt{l-reuse-only},
$\Preuseonly$, can be read as: the set of indices that
receive value from intersection, but does not receive from \textsf{ADD} or
\textsf{SUB}, and this is precisely the set of points that can be directly
copied along $\vec{r}$. Thus, \texttt{l-reuse-only} performs just this copy operation: \( 
\textsf{LHS}[\mathbf{u}] = \textsf{LHS}[T_r^{\mathbf{u}}(\mathbf{u})] \).

\paragraph{Add Only} \texttt{l-add-only}'s domain $\Paddonly$ can be read as: the set of indices 
that receive value from 
\textsf{ADD}, but does not recieve value from intersection. One can verify that $\Paddonly = \Paddonly 
- \mathcal{P}_{\text{sub}}^{\mathbf{u}}$ this also implies that the 
set  also does not recieve value from \textsf{SUB}. Therefore, the statement 
just copies from \textsf{ADD}

\paragraph{Add and Reuse}
\texttt{l-add-reuse}'s domain, $\Paddreuse$, can be read as: the set of indices that receive value from 
\textsf{ADD} and the intersection, but does not receive value from \textsf{SUB}. Therefore the 
statement reuses value  along $\vec{r}$, and increments with value calculated from \textsf{ADD}.

\paragraph{Sub and Reuse}
\texttt{l-reuse-sub}'s domain, $\Preusesub$, can be read as: the set of indices that recieves value from 
\textsf{SUB} and the intersection, but does not receive value from \textsf{ADD}. Therefore the 
statement reuses value along $\vec{r}$, and decrements with value calculated from \textsf{SUB}.

\paragraph{Add, Reuse and Sub}
\texttt{l-add-reuse-sub}'s domain, $\Paddreusesub$. can be read as: the set of indices that 
receive value from both \textsf{ADD}, the intersection and \textsf{SUB}. Therefore, the statement reuse 
along $\vec{r}$, increments with \textsf{ADD} and decrements with \textsf{SUB}.

} 

\paragraph{Residual Reductions}

The statements \texttt{ladd} and \texttt{lsub} are themselves reductions, and we will call them {\em 
residual reductions} after SR transformation. They compute additional values that are requested by the 
top five statements. The residual reduction accumulates the same right-side expression as the original 
reduction, but with domains that are subsets of the original domain.

\subsection{Configuration of Simplifying Reduction}

As mentioned in \Cref{sec:srconfig}, we need to consider three constraints -- complexity, 
sharing and dependence -- when choosing the reuse vector $\vec{r}$. Here we discuss the 
constraints in more detail. 

\paragraph{Complexity}
A program's complexity is a function over its input parameters.

The complexity after one SR transformation is equal to the total sum of all cardinalities of 
domain sizes of the statements after the transformation. The complexity of the first five 
statements combined together is equivalent to iterating points of \textsf{LHS} array, and 
therefore it will always remain unchanged, since we will always need to compute answer for 
each point of the \textsf{LHS}. As shown in \cite{sr}, in order for one step of SR to 
be meaningful, in the sense that it decreases the complexity, we need that 
$|\mathcal{P}_{\text{add}}| + 
|(\mathcal{P}_{\text{int}}^{\mathbf{u}})^{\mathbf{s}} \cap \mathcal{P}_{\text{sub}}| < 
|\mathcal{P}|$. 

\paragraph{Sharing} 
Fully determining all possible $\vec{r}$ that presents 
sharing for the right-hand side expression is not decidable: 
for an arbitrary RHS expression $expr$ as an uninterpreted function: we can encode 
the problem as $\exists \vec{r} \forall \mathbf{v}. \llbracket \textit{expr} \rrbracket =  
\llbracket \allowbreak \subst{\allowbreak \textit{expr}}{T_r(\freevars(\textit{expr}))} 
\rrbracket $, 
and this is not decidable in general. 
However, we can still heuristically deduce valid reuse vectors, if we know the internals 
structure of 
$expr$.
\cite{sr} proposed a heuristical approach by computing a polyhedral set
$\mathcal{S}(S)$ called {\em share space}, which is the intersection of
the nullspaces of the dependence functions of all the subexpressions of
$expr$ in statement $S$, and selecting any vector $r \in 
\mathcal{S}(S)$.

\section{Enabling Transformations}
\label{sec:sr_enabling_transformations}

In this section we briefly review the enabling transformations introduced in \cite{sr}. Since 
these transformations are important to fully utilize ST for a single reduction, we encourage readers 
to find more details of these transformations in \cite{sr}.

\paragraph{Reduction Decomposition}
For reduction with projection function $\textit{proj}$, we can potentially decompose 
$\textit{proj}= \textit{proj1} \circ \textit{proj2}$, where $\circ$ denotes function composition. It 
is possible to break the reduction into two statements: the first statement with projection 
$\textit{proj2}$ produces an intermediate output, followed by a reduction with projection 
$\textit{proj1}$ that returns the original output. The first statement could lead to a larger share 
space than the original reduction, and therefore reduction decomposition enables 
ST.

\paragraph{Same Operator Transformation}
It's possible to lift inner expressions out of reductions to increase share space.

\paragraph{Distributivity Transformation}
It's possible to utilizie distributivity of an operator to lift inner expression out of reductions to 
increase share space.

\paragraph{Higher Order Operator Transformation}
It's possible to collapse along the entire reuse space, if the reduce operator $\oplus$ has an 
higher order operator $\otimes$. 

\section{Proof of Lemma~6.1}
\label{sec:prooflemma61}
\begin{proof}[Proof of \Cref{lemma:complexitydecrease}]
For any ST application, the reduction's $n$-dimensional domain $\mathcal{P}$ is 
reduced to the  
two half shells $\mathcal{P} - \mathcal{P}'$ and  $\mathcal{P}' - \mathcal{P}$. 
The two half shells decompose into convex polyhedral sets corresponding to all 
$(n-1)$-dimensional faces of $\mathcal{P}$. Further, for each decomposed convex 
polyhedral set, the thickness of the set, which is defined as the spanned width of the set 
orthogonal to its corresponding face, is a 
constant dependent solely on the ST's reuse vector and the face's orientation. 
Therefore, the cardinality of each decomposed polyhedral set is just the cardinality of 
the face multiplied by some constant. 
It then follows that for any two STs with two non-zero reuse vectors, their resultant 
residual reductions' complexities are 
always the same, and equal the sum of the cardinalities of all the faces of $\mathcal{P}$ 
multiplied by some constant. 
\end{proof}
\section{More Related Works and Discussions}
\label{sec:relatedworksanddiscussions}

\subsection{More Related Works}
\label{sec:more_related_works}

\paragraph{ILP Scheduling} Previous work \citep{ilpschedulemultidim, 
pouchet.07.cgo, 
pouchet.08.pldi} gave an ILP formulation of the scheduling problem. Specifically, 
\cite{ilpschedulemultidim} 
showsed how to construct constraints for a convex ILP 
problem to find an $m$-dimensional schedule for a program. Moreover, this 
formulation of constraints 
allows one to incorporate a desired objective to be optimized.
I\texttt{}n this work, we 
used the complexity of the 
final transformed program as the objective and we showed how to encode such 
objective as a affine 
expression in \Cref{sec:mssrcomplexity}. 

\paragraph{Heuristic Scheduling} 
There are also other scheduling methods such as the ones in 
\cite{someeffsol2affschedpt1, 
someeffsol2affschedpt2, plutosched} that use heuristics to schedule a program. 
These methods are 
usually more scalable than an ILP formulation. In this work, we use ideas from the 
ILP formulation to 
formulate the SDR problem, while our provided heuristic algorithm does not 
depend on using the ILP 
formulation for scheduling. Instead, we use the PLUTO \citep{isl, plutosched} 
heuristic scheduling 
algorithm in our implementation. 

\subsection{Incorpration of Technique into Probabilistic Programming System}
\label{sec:incorprateintopps}

We elaborate in limited detail on how we envision our techniques can be 
incorporated into Probabilistic Programming Systems.

In our vision, a user would not need to write probailistic programming algorithm 
in our polyhedral IR. 
However, it would be the responsibility of the probabilistic
programming system to provide a mapping from the source code of the
probabilistic program and the inference algorithm to an intermediate
representation that is amenable to our technique. There are several
probabilistic programming systems that compile probabilistic programs
and their inference algorithms to an intermediate representation, thus
making the resulting programs amenable to a variety of compiler
optimizations \citep{hakaru,gen,augurv2}.  To then specifically apply our
technique, we envision two possibilities:

\begin{enumerate}
\item Polyhedral IR: in this instance, the probabilistic programming system uses a 
polyhedral IR. We anticipate that the primary extension
of a standard polyhedral IR (e.g. IR used by \cite{pollyred}) to support our
work would be to add explicit reductions to the IR. The probabilistic
programming system would then need to lower reductions in the source
program to reductions in the IR. Alternatively, the system could use
reduction recognition techniques 
\cite[Section~6.1]{interproceduralanalysisforparallelization} to recover
reductions if it uses a more standard lowering process.
\item Non-polyhedral IR: in this instance, the probabilistic programming system 
uses a more traditional IR, such as three address code. Such an
approach would still need to identify and distinguish reductions as in
the polyhedral IR case. Where the approach differs is that a
non-polyhedral IR can make it challenging to access the full
capabilities of our technique. For example, the set of available reuse
directions may be limited to the iteration directions of existing
loops. To this end, we have implemented a limited version of our
heuristic algorithm in Shuffle, a probabilistic programming language
for typesafe programmable inference \cite{shuffle}. In this language, users
provide a probabilistic program and inference algorithm written in a
high-level language. Shuffle then translates the program to a
non-polyhedral IR with distinguished reductions that we then optimize
with our heuristic. 
\end{enumerate}

\ifdefined\includesequetialscheduleappendix
{\color{blue}
\section{Sequential schedule}
\label{sec:sequential_schedule}

To ensure the soundness of our heuristic algorithm, 
we require the schedule to be fully sequential, 
meaning no two instances of a statement are 
scheduled at the same timestamp. 
Specifically, for a statement $S$, we require that its schedule, $\Theta^S$, 
satisfies 
the following:
\begin{equation}\label{eq:sequential_schedule_requirement}
\Theta^S \cdot 
[\vec{x}^S_1, \vec{p}, 1]^\transpose \neq 
\Theta^S \cdot 
[\vec{x}^S_2, \vec{p}, 1]^\transpose
\quad \forall \vec{x}^S_1, \vec{x}^S_2. \; \vec{x}^S_1 \neq \vec{x}^S_2.
\end{equation}
A sufficient condition for \Cref{eq:sequential_schedule_requirement} is 
that $\Theta^S$ is {\em injective}, mapping unique iterations to unique 
timestamps.
That is that
\begin{equation}
\label{eq:is_injective}
\Theta^S \cdot [\vec{x}^S_1, \vec{p}_1, 1]^\transpose \neq  
\Theta^S \cdot [\vec{x}^S_2, \vec{p}_2, 1]^\transpose
\quad \forall \vec{x}^S_1,  \vec{p}_1, \vec{x}^S_2, \vec{p}_2. \; 
[\vec{x}^S_1, \vec{p}_1, 1]^\transpose \neq
[\vec{x}^S_2, \vec{p}_2, 1]^\transpose .
\end{equation}

\paragraph{Concept} If one desires to use an ILP formalization for scheduling 
(\Cref{sec:ilp_schedule}), 
then it is 
possible to ensure that $\Theta^S$ is injective by requiring that the last $d = 
\max_{\forall S} 
\left(|\vec{x}^S| + |\vec{p}| + 1\right)$ rows of $\Theta^S$ are the identity 
matrix.
Conceptually, this requirement results in appending the iteration vector 
to its otherwise standardly computed timestamp.
Therefore, different iterations that would otherwise be mapped to the same 
timestamp are then mapped to unique timestamps according to their 
iteration (therefore ensuring that \Cref{eq:is_injective} is satisfied).

We note that while the requirement is sufficient, it is not guaranteed to produce a 
minimal dimension schedule. 
For example, if the dependences in the original 
program by themselves force a sequential schedule, then the additional rows are 
not necessary.

\paragraph{Approach} 
To ensure squentiality for the ILP formulation in \Cref{eq:ilpscheduleconvex}, 
we first require that the schedule timestamp's dimension $m$ is greater than the 
maximum of the dimenions of the iterations vectors of all statements 
\Cref{eq:sequential_schedule_dimension_constraint}. 
\begin{equation}\label{eq:sequential_schedule_dimension_constraint}
m \ge \max_{\forall S} |\vec{x}^S| + |\vec{p}| + 1
\end{equation}

We next augment the original ILP formulation 
(\Cref{eq:ilpscheduleconvex}) with an additional 
constraint to ensure injectivity  (\Cref{eq:enforce_identity_constraint}).
\begin{subequations}\label{eq:ilpscheduleconvexrepeat}
\begin{align}
\forall \mathcal{D}_{S_1, S_2}, & \forall k \in \{1...m\},  
\delta_k^{\mathcal{D}_{S_1, 
S_2}} \in \{0, 1\}  \label{eq:ilpscheduleconvexrepeata}\\
\forall \mathcal{D}_{S_1, S_2}, & \sum_{k = 1}^m  
\delta_k^{\mathcal{D}_{S_1, S_2}} = 1 \label{eq:ilpscheduleconvexrepeatb} \\
\begin{split}
\label{eq:ilpscheduleconvexrepeatc}
\forall \mathcal{D}_{S_1, S_2}, & \forall k \in \{1 ... m\}, \forall 
[\vec{x}^{S_1}, \vec{x}^{S_2}, \vec{p}] 
\in \mathcal{D}_{S_1, S_2} \\
& \Theta^{S_2}_k \cdot \begin{bmatrix} \vec{x}^{S_2} \\ \vec{p} \\ 1 
\end{bmatrix}
- \Theta^{S_1}_k \cdot \begin{bmatrix} \vec{x}^{S_1} \\ \vec{p} \\ 1 
\end{bmatrix} 
\ge \delta_k^{\mathcal{D}_{S_1, S_2}} - \sum_{i=1}^{k-1} 
\delta_i^{\mathcal{D}_{S_1, S_2}}  (K \vec{p} 
+ K)
\end{split} \\
\label{eq:enforce_identity_constraint}
\forall \Theta^S, & \Theta^S_{\left(m - |\vec{x}^S + \vec{p} + 1|\right): m} = I
\end{align}
\end{subequations}
\Cref{eq:ilpscheduleconvexrepeata,eq:ilpscheduleconvexrepeatb,eq:ilpscheduleconvexrepeatc}
  exactly duplicate \Cref{eq:ilpscheduleconvex}.
Additionally, \Cref{eq:enforce_identity_constraint} requires that the last 
$|\vec{x}^S + \vec{p} + 1|$ rows of $\Theta^S$ are the identity matrix $I$.

\paragraph{Existence of Sequential Schedule}
We next demonstrate that if there exists a potentially non-sequential schedule to 
\Cref{eq:ilpscheduleconvex}, then there exists a sequential schedule to 
\Cref{eq:ilpscheduleconvexrepeat}.
\begin{theorem}[Sequential schedule ILP formulation]
\label{thm:sequential_schedule_ilp}
Given a program in the polyhedral IR, if there exists a valid (potentially 
non-sequential) schedule $\hat{\Theta}^S$ 
of dimension $\hat{m}$
to the original ILP formulation  (\Cref{eq:ilpscheduleconvex}), and let 
$m = \hat{m} + (\max_{\forall S} |\vec{x}^S| + |\vec{p}| + 1)$, then there exists a 
valid sequential schedule $\Theta$ 
of dimension $m$ to the augmented ILP formulation
(\Cref{eq:ilpscheduleconvexrepeat}).
\end{theorem}
To prove \Cref{thm:sequential_schedule_ilp}, we first demonstrate
\Cref{lemma:schedule_function_injectivity,%
lemma:sequential_schedule_completeness}, which show the existence and 
sequentiality of the solution to 
\Cref{eq:ilpscheduleconvexrepeat}, respectively.

\begin{lemma}[Scheduling function completeness]
\label{lemma:sequential_schedule_completeness}
Given any program in the polyhedral IR, 
if there exists a valid schedule $\hat{\Theta}$ of 
dimension $\hat{m}$, 
then for any $m > \hat{m}$, define the schedule $\Theta$ as:
(i)  for each statement $S$, $\Theta^S$ is given by
$\Theta^S \triangleq
\begin{bmatrix}\hat{\Theta}^S  \\ M^S
\end{bmatrix}$, and 
(ii) $M^S$ is an arbitrary matrix of shape $(m - \hat{m})  \times 
(|\vec{x}^S| + |\vec{p}| + 1|)$ .
Then $\Theta$ is a valid schedule of dimension $m$ for the given program.
\end{lemma}
In other words, if we append arbitrary rows to end to each 
matrix in a valid 
schedule $\hat{\Theta}$, then the obtained schedule $\Theta$ is also valid.
\begin{proof}
Since $\hat{\Theta}$ is a valid schedule for the given program, it must satisfy all 
dependences in the program. 
That is, for all dependence relation $\mathcal{D}_{S_1, S_2}$ between statements 
$S_1$ and $S_2$,
and for all instances $[\vec{x}^{S_1}, \vec{p}, 1]^T$ and 
$[\vec{x}^{S_2}, \vec{p}, 1]^T$ such that $[\vec{x}^{S_1}, \vec{x}^{S_2}, \vec{p}] 
\in \mathcal{D}_{S_1, S_2}$, we have
\begin{equation}
\hat{\Theta}^{S_1} \cdot [\vec{x}^{S_1}, \vec{p}, 1]^\transpose < 
\hat{\Theta}^{S_2} \cdot [\vec{x}^{S_2}, \vec{p}, 1]^\transpose,
\end{equation}
that is, the schedule timestamp of $\vec{x}^{S_1}$ is less than the schedule 
timestamp of $\vec{x}^{S_2}$.

If we append an arbitrary matrix $M^S$ of $m - \hat{m}$ rows to 
$\hat{\Theta}^S$, then the 
inequality
\begin{equation}
\label{eq:extend_thetahat_valid}
\begin{bmatrix}
\hat{\Theta}^{S_1} \\
M^{S_1}
\end{bmatrix} \cdot [\vec{x}^{S_1}, \vec{p}, 1]^\transpose < 
\begin{bmatrix}
\hat{\Theta}^{S_1} \\
M^{S_2}
\end{bmatrix} \cdot [\vec{x}^{S_2}, \vec{p}, 1]^\transpose,
\end{equation}
still holds because of the entry-by-entry schedule timestamp comparison.
This means that the schedule matrix $\Theta^S = \begin{bmatrix}\hat{\Theta}^S  
\\ M^S \end{bmatrix}$ is a valid schedule matrix of dimension $m$ for statement 
$S$.

Since \Cref{eq:extend_thetahat_valid} holds for each statement $S$, $\Theta$ is a 
valid schedule of dimension $m$ for the given program.
\end{proof}

\begin{lemma}[Scheduling function injectivity]
\label{lemma:schedule_function_injectivity}
For any schedule function $\Theta^S$ that satisfies
\Cref{eq:enforce_identity_constraint}:
\begin{equation}
\Theta^S \cdot [\vec{x}_1, \vec{p}_1, 1]^\transpose \neq  
\Theta^S \cdot [\vec{x}_2, \vec{p}_2, 1]^\transpose
\quad \forall \vec{x}_1,  \vec{p}_1, \vec{x}_2, \vec{p}_2. \; 
[\vec{x}_1, \vec{p}_1, 1]^\transpose \neq
[\vec{x}_2, \vec{p}_2, 1]^\transpose .
\end{equation}
\end{lemma}
\begin{proof}
By the additional constraint in \Cref{eq:enforce_identity_constraint}, 
$\Theta^S$ has the structure that the last $|\vec{x}^S| + |\vec{p}| + 1$ rows are 
the identity matrix.
\Cref{eq:sequential_schedule_theta_illustration} illustrates the structure 
of $\Theta^S$. 
\begin{equation}
\label{eq:sequential_schedule_theta_illustration}
\Theta^S 
= 
\left [ \begin{array}{*{3}{c}}
{\theta}_{1, 1} & ... & {\theta}_{1, |\vec{x}| + |\vec{p}| + 1 } \\
\vdots & ... & \vdots \\
{\theta}_{m-\left(|\vec{x}| + |\vec{p}| + 1\right) - 1, 1} & ... & 
{\theta}_{m-\left(|\vec{x}| + |\vec{p}| + 1\right) - 1, 
|\vec{x}| + |\vec{p}| + 1 } 
\\
\hline
1 & 0 ... & 0 \\
0 & 1 ... & 0 \\
\vdots & ... & \vdots \\
0 & ... & 1
\end{array}
\right ]
{
\begin{array}{@{}l@{}}
\left\}\begin{array}{@{}c@{}}\null\\\vphantom{\vdots}\\\null\end{array}\right.
 m- (|\vec{x}| + |\vec{p}| + 1) \\
\left\}\begin{array}{@{}c@{}}\null\\\vphantom{\vdots}\\\null\\\null\end{array}\right.
 
|\vec{x}| + 
|\vec{p}| + 1
\end{array}
}
\end{equation}

Therefore, for a given $\Theta^S$ and
for any pair of $[\vec{x}_1, \vec{p}_1, 1]^T$ and $[\vec{x}_2, 
\vec{p}_1, 
1]^T$ such that $[\vec{x}_1, \vec{p}_2, 1]^T \neq [\vec{x}_2, 
\vec{p}_2, 
1]^T$, we have 
\begin{equation}
\label{eq:sequential_schedule_time_neq}
\Theta^S \cdot [\vec{x}_1, \vec{p}_1, 1]^\transpose = 
\begin{bmatrix}
\Theta^S_{1:m-\left(|\vec{x}| + |\vec{p}| + 1\right)} \cdot [\vec{x}_1, 
\vec{p}_1, 1]^\transpose \\
\vec{x}_1 \\ \vec{p}_1 \\
1 
\end{bmatrix} \neq 
\begin{bmatrix}
\Theta^S_{1:m-\left(|\vec{x}| + |\vec{p}| + 1\right)} \cdot [\vec{x}_2, 
\vec{p}, 
1]^\transpose \\
\vec{x}_2 \\ \vec{p}_2\\
1 
\end{bmatrix} = 
\Theta^S \cdot [\vec{x}_2, \vec{p}_2, 1]^\transpose .
\end{equation}
The inequality in \Cref{eq:sequential_schedule_time_neq} holds because 
the left and right timestamps differ in the last $|\vec{x}| + |\vec{p}| + 1$ rows.
\end{proof}

We are now ready to prove \Cref{thm:sequential_schedule_ilp}.
\begin{proof}[Proof for \Cref{thm:sequential_schedule_ilp}]
By \Cref{lemma:sequential_schedule_completeness}, solution $\Theta$ exists and 
is a valid schedule.
By \Cref{lemma:schedule_function_injectivity}, each $\Theta^S$ in $\Theta$ is 
injective; thus, $\Theta$ is a sequential schedule.
In sum, we have that $\Theta$ is a valid sequential schedule.
\end{proof}
}
\fi

\bibliography{paper}